\documentclass[11pt]{article}

\usepackage{fullpage}
\usepackage{amsmath,amsfonts,amsthm,mathrsfs,mathpazo,xspace,hyperref,graphicx,bm}
\usepackage{needspace}

\newtheorem{theorem}{Theorem}[section]
\newtheorem{proposition}[theorem]{Proposition}
\newtheorem{lemma}[theorem]{Lemma}

\newtheorem{corollary}[theorem]{Corollary}

\newtheorem{question}{Open problem}[section]

\newcommand{\beq}{\begin{eqnarray}}
\newcommand{\eeq}{\end{eqnarray}}

\newcommand{\ket}[1]{|#1\rangle}
\newcommand{\bra}[1]{\langle#1|}
\newcommand{\Tr}{\mbox{\rm Tr}}
\newcommand{\Id}{\ensuremath{\mathop{\rm Id}\nolimits}}
\newcommand{\Ex}{\textsc{E}}

\newcommand{\setft}[1]{\mathrm{#1}}

\newcommand{\Pos}{\setft{Pos}}
\newcommand{\Unitary}{\setft{U}}

\newcommand{\Bil}{\setft{B}}
\newcommand{\cBil}{cb}
\newcommand{\Lin}{\setft{L}}
\newcommand{\Bounded}{\mathcal{B}}
\newcommand{\Ball}{\setft{Ball}}

\newcommand{\val}{\omega}
\newcommand{\bias}{\beta}

\newcommand{\C}{\ensuremath{\mathbb{C}}}
\newcommand{\N}{\ensuremath{\mathbb{N}}}

\newcommand{\K}{\ensuremath{\mathbb{K}}}
\newcommand{\R}{\ensuremath{\mathbb{R}}}

\newcommand{\mH}{\ensuremath{\mathcal{H}}}
\newcommand{\h}{\ensuremath{\mathcal{H}}}
\newcommand{\prob}{\ensuremath{\mathcal{P}}}
\newcommand{\class}{\ensuremath{\mathcal{P}_{\mathcal{C}}}}
\newcommand{\ent}{\ensuremath{\mathcal{P}_{\mathcal{Q}}}}

\newcommand{\setx}{\textbf{X}}
\newcommand{\cardx}{\textsc{x}}
\newcommand{\sety}{\textbf{Y}}
\newcommand{\cardy}{\textsc{y}}

\newcommand{\seta}{\textbf{A}}
\newcommand{\carda}{\textsc{a}}
\newcommand{\setb}{\textbf{B}}
\newcommand{\cardb}{\textsc{b}}

\DeclareMathOperator{\poly}{poly}

\newcommand{\eps}{\varepsilon}

\begin{document}

\title{Survey on Nonlocal Games and Operator Space Theory}
\author{Carlos Palazuelos\thanks{Instituto de Ciencias Matem\'aticas (ICMAT), Facultad de Ciencias Matem\'aticas, Universidad Complutense de Madrid, Spain. Email: \texttt{cpalazue@mat.ucm.es}} \quad and \quad Thomas Vidick\thanks{Department of Computing and Mathematical Sciences, California Institute of Technology, Pasadena, CA, USA. Email: \texttt{vidick@cms.caltech.edu}}}
\date{}
\maketitle

\abstract{This review article is concerned with a recently uncovered connection between \emph{operator spaces}, a noncommutative extension of Banach spaces, and \emph{quantum nonlocality}, a striking phenomenon which underlies many of the applications of quantum mechanics to information theory, cryptography and algorithms. Using the framework of nonlocal games, we relate measures of the nonlocality of quantum mechanics to certain norms in the Banach and operator space categories. We survey recent results that exploit this connection to derive large violations of Bell inequalities, study the complexity of the classical and quantum values of games and their relation to Grothendieck inequalities, and quantify the nonlocality of different classes of entangled states.}

\tableofcontents


\section{Introduction}

The development of quantum mechanics has from its outset been closely intertwined with advances in mathematics. Perhaps most prominently, the modern theory of operator algebras finds its origins in John von Neumann's successful unification of Heisenberg's ``matrix mechanics'' and Schr\"odinger's ``wave mechanics'' in the early 1920s. It should therefore not be unexpected that subsequent advances in operator algebras would yield insights into quantum mechanics, while questions motivated by the physics of quantum systems would stimulate the development of the mathematical theory. In practice however, aside from algebraic quantum field theory the more exotic types of algebras considered by von Neumann have (arguably) had a limited impact on physicist's daily use of quantum mechanics, this in spite of the theory's increasingly pervasive role in explaining a range of phenomena from the theory of superconductors to that of black holes. 

This survey is concerned with the emergence of a promising new connection at this interface. One of the most intriguing aspects of quantum mechanics is the existence of so-called ``entangled'' quantum states. For our purposes the joint state of two or more quantum mechanical systems is called entangled whenever it is possible, by locally measuring each of the systems whose state is described,  to generate correlations that cannot be reproduced by any ``local hidden-variable model'' (LHV) whereby any local observable should \emph{a priori} have a locally well-defined outcome distribution. Depending on its precise formulation, this phenomenon can lead to the most confusing interpretations --- as evidenced by Einstein, Podolsky and Rosen's famous ``paradox''~\cite{epr}. It took more than thirty years and the work of many researchers, including most prominently Bell~\cite{Bell:64a}, to establish the underlying \emph{nonlocality} of quantum mechanics on firm formal grounds. Increasingly convincing experimental demonstrations have followed suit, from Aspect's experiments~\cite{Aspect81} to the recentl ``all-loopholes-closed'' experiments~\cite{hensen2015experimental}. 

The essential mathematical construction that underlies the study of quantum nonlocality is the tensor product of two Hilbert spaces $\mH_1$, $\mH_2$. Each Hilbert space (or rather its unit sphere) represents the state space of a quantum mechanical system. The tensor product of the two spaces provides the state space for joint systems $\ket{\psi}\in\mH_1\otimes\mH_2$. A \emph{local observable}\footnote{The term ``observable'' is used simultaneously for a quantity, such as spin or momentum, and the operator that measures the quantity.} of the system can be represented by an operator  $A\in\Bounded({\mH_1})$ or $B\in\Bounded({\mH_2})$ of norm $1$, whose action on $\ket{\psi}$ is obtained via the application of $A\otimes \Id_{\mH_2}$ or $\Id_{\mH_1}\otimes B$ respectively. 

A \emph{Bell experiment} involves a state $\ket{\psi}\in \mathcal{H}_1\otimes \mathcal{H}_2$ and collections of local quantum observables $\{A_x\}\subset \Bounded(\mathcal{H}_1)$ and $\{B_y\}\subset \Bounded(\mH_2)$. The nonlocality of quantum mechanics is evidenced through the \emph{violation} of a  \emph{Bell inequality} by this state and observables. More precisely, Bell's work establishes the following as the fundamental question in the study of nonlocality:

\begin{quote}
Suppose $M=(M_{xy})\in \R^{n\times m}$ (the \emph{Bell functional}) is such that 
\begin{equation}\label{eq:class-value-intro}
\sup_{a,b}\, \Big|\sum_{xy} M_{xy} \int_\lambda a_x(\lambda) b_y(\lambda) d\lambda\Big|, 
\end{equation}
where the supremum is over all probability spaces $\Omega$, probability measures $d\lambda$ on $\Omega$ and measurable functions $a_x,b_y:\Omega\to [-1,1]$, is at most $1$ (the \emph{Bell inequality}). How large can (the \emph{violation} of the Bell inequality)
\begin{equation}\label{eq:quant-value-intro}
\sup_{\{A_x\},\{B_y\}} \,\Big|\sum_{xy} M_{xy} \bra{\psi} A_x\otimes B_y\ket{\psi}\Big|,
\end{equation}
where now the supremum is taken over all Hilbert spaces $\mH_1$ and $\mH_2$, states $\ket{\psi}\in \mathcal{H}_1\otimes \mathcal{H}_2$ and local observables $A_x\in \Bounded(\mathcal{H}_1)$, $B_y\in \Bounded(\mH_2)$, be?
\end{quote}

The expert in Banach space theory will immediately recognize that the first quantity~\eqref{eq:class-value-intro} is precisely the \emph{injective norm} of the tensor $M\in\ell_1^n(\R) \otimes_\epsilon \ell_1^m(\R)$, or equivalently the norm of $M$ when it is viewed as a bilinear form $\ell_\infty^n(\R)\times\ell_\infty^m(\R)\to \R$. But what does the second quantity~\eqref{eq:quant-value-intro} correspond to? Colloquially, quantum mechanics is often thought of as a ``non-commutative'' extension of classical mechanics, and it is perhaps not so surprising that the answer should come out of a non-commutative extension of Banach space theory --- the theory of \emph{operator spaces}. Initiated in Ruan's thesis (1988), operator space theory is a \emph{quantized extension} of Banach space theory whose basic objects are no longer the elements of a Banach space $X$ but sequences of matrices $M_d(X)$, $d\geq 1$, with entries in $X$ (\cite{EffrosRuanBook}, \cite{PisierBook}). 
This very recent theory provides us with the matching quantities: in operator space terminology~\eqref{eq:quant-value-intro} is precisely the \emph{minimal} norm of the tensor $M$, or equivalently the \emph{completely bounded} norm $\|\cdot\|_{cb}$ of the associated bilinear form. Thus the fundamental question in the study of nonlocality expressed above finds a direct reformulation as a basic question in operator space theory:

\begin{quote}
Suppose that $M:\ell_\infty^n(\R)\times\ell_\infty^m(\R)\to \R$ is such that $\|M\|\leq 1$. How large can $\|M\|_{cb}$ be? 
\end{quote}

This striking connection, only unearthed very recently~\cite{PerezWPVJ08tripartite}, has already led to a number of exciting developments, with techniques from the theory of operator spaces used to provide new tools in quantum information theory and results from quantum information theory applied to derive new results in operator space theory. Without attempting to be exhaustive, the aim of this survey is to give a flavor for this active area of research by explaining the connection in its most elementary --- and convincing --- instantiations, while simultaneously describing some of its most advanced developments. Highlights of the topics covered include:
\begin{itemize}
\item A discussion of the tight connections between correlation inequalities (or \emph{XOR games}),  Grothendieck's fundamental theorem of the theory of tensor norms, Tsirelson's bound and semidefinite programming, 
\item Unbounded violations for three-player XOR games as well as two-player games with large answer sizes inspired by operator space theory, and their implications for problems of interest in computer science such as parallel repetition,
\item A proof of the operator space Grothendieck inequality inspired by the notion of embezzlement in quantum information theory,
\item A brief exposition of the thorny problem of infinite-dimensional entanglement and its relation to Connes' embedding conjecture in the theory of $C^*$-algebras.
\end{itemize}

The survey is organized as follows. We start with a section introducing the basic concepts that form the basis of the survey. Although we will assume that the reader has a minimal background in functional analysis and quantum information, we provide all the required definitions, including all those that involve operator space theory, that go beyond elementary concepts such as Banach spaces or density matrices. Explaining the title of the survey, we adopt a viewpoint on the theory of Bell inequalities in terms of \emph{nonlocal games}. Aside from its playful aspects this powerful framework from computer science has the advantage of connecting the study to important questions, and results, in complexity theory. 

Section~\ref{sec:xor-games} is concerned with Bell correlation inequalities, or, in the language of games, XOR games. This is the setting in which the connection between nonlocality and operator space theory is the tightest and most striking. Since at that level the theory is well-established we state most results without proof; however we also  include proof sketches for more recent results on multipartite XOR games and unbounded violations thereof. 

In Section~\ref{sec:2p1r} we first extend the results of the previous section to arbitrary two-player games, where operator spaces still provide a very tight description of the main quantities of interest. We then treat the case of general Bell inequalities, for which the connection is less accurate, but still sufficient to derive interesting results. We provide bounds, from above and from below, on the largest violations that can be achieved in this general setting. Virtually all such results known originate in the operator space viewpoint; some of them are new to this survey, in which case we include complete proofs. 

In Section~\ref{sec:Entanglement} we look more closely at the question of entanglement, and in particular the kinds of states, and their dimension, that can lead to large violations. This leads us to a discussion of further connections with deep problems in functional analysis, including the operator space Grothendieck inequality and Connes' embedding conjecture. 

A number of open problems that we view as interesting are interspersed throughout, with the hope that the reader will be enticed into contributing to the further development of the exciting connections that are brought forward in this survey.


\section{Preliminaries}

\subsection{Notation}

Throughout we use boldfont for finite sets $\setx,\sety,\seta,\setb$, small capital letters $\cardx,\cardy,\carda,\cardb$ for their cardinality and $x\in\setx,y\in\sety,a\in\seta,b\in\setb$ for elements. The sets $\setx,\sety$ will usually denote sets of \emph{questions}, or \emph{inputs}, to a game or Bell inequality, and $\seta,\setb$ will denote sets of \emph{answers}, or \emph{outputs}.

For a set $X$, $M_d(X)$ is the set of $d\times d$ matrices with entries in $X$. When we simply write $M_d$ we always mean $M_d(\C)$. Algebraically, $M_d(X)=M_d\otimes X$. We use $\Pos(\C^d)$ to denote the semidefinite positive linear operators on $\C^d$.

Given a vector space $\mathcal{V}$ we will use $\{e_i\}$ to denote an arbitrary fixed orthonormal basis of $\mathcal{V}$, identified as a ``canonical basis'' for $\mathcal{V}$. 
We write $\{E_{ij}\}$ for the canonical basis of $M_d$.

Write $\ell_p^n$ for the $n$-dimensional complex $\ell_p$ space, and $\ell_p^n(\R)$ when we want to view it as real. $S_p^n$ is the  $n$-dimensional Schatten $p$-space. Unless specified otherwise the space $\C^n$  will always be endowed with the Hilbertian norm and identified with $\ell_2^n$. The inner product between two elements $x,y\in\ell_2^n$ is denoted by $x\cdot y$. We write $\Ball (X)$ for the closed unit ball of the normed space $X$.

Given Banach spaces $X,Y$ and $Z$, $\Lin(X,Y)$ is the set of linear maps from $X$ to $Y$. We also write $\Lin(X)$ for $\Lin(X,X)$, and $\Id_X\in\Lin(X)$ for the identity map on $X$.  $\Bil(X,Y;Z)$ is the set of bilinear maps from $X\times Y $ to $Z$. If $Z=\C$ we also write $\Bil(X,Y)$. We write $\|\cdot\|_{X}$ for the norm on $X$. When we write $\|x\|$ without explicitly specifying the norm we always mean the ``natural'' norm on $x$: the Banach space norm if $x\in X$ or the implied operator norm if $x\in\Lin(X,Y)$ or $x\in\Bil(X,Y;Z)$. 

Given a linear map $T:X\rightarrow Y$ between Banach spaces, $T$ is \emph{bounded} if its norm is finite: $\|T\|:=\sup\{\|T(x)\|_Y: \|x\|_X\leq 1\}<\infty$. We let $\Bounded(X,Y)$ denote the Banach space of bounded linear maps from $X$ to $Y$, and write $\Bounded(X)$ for $\Bounded (X,X)$. Similarly, $B\in\Bil(X,Y;Z)$ is said bounded if $\|B\|:=\sup\{\|B(x,y)\|_Z: \|x\|_X,\|y\|_Y\leq 1\}<\infty$, and we write $\Bounded(X,Y;Z)$ for the set of such maps. 
If $X$ is a Banach space, we write $X^* = \Lin(X,\C)$ for its dual.

\subsection{Operator spaces}
\label{sec:prelims-os}
We introduce all notions from Banach space theory and operator space theory needed to understand the material presented in this survey. For the reader interested in a more in-depth treatment of the theory we refer to the standard books~\cite{Def93},~\cite{PisierBook} and~\cite{PaulsenBook}. 

An \emph{operator space} $X$ is a closed subspace of the space of all bounded operators on a complex Hilbert space $\h$. For any such subspace the operator norm on $\Bounded(\h)$ automatically induces a sequence of \emph{matrix norms} $\|\cdot\|_d$ on $M_d(X)$ via the inclusions  $M_d(X) \subseteq M_d(\Bounded(\h))\simeq \Bounded(\h^{\oplus d})$. That is, a matrix $A$ with entries $A_{ij}\in X$ is interpreted as a bounded operator on $\h^{\oplus d}$ through $A(\sum_j x_j \otimes e_j)=\sum_{ij} A_{ij}(x_j)\otimes e_i$, and equipped with the corresponding operator norm.

Ruan's Theorem \cite{Ruan88} characterizes those sequences of norms which can be obtained in this way. It provides an alternative definition of an operator space as a complex Banach space $X$ equipped with a sequence of matrix norms $(M_d(X),\|\cdot\|_d)$ satisfying the following two conditions: for every pair of integers $(c,d)$,
\begin{itemize}
\item $\|v \oplus w\|_{c+d} = \max\{\|v\|_c,\|w\|_d \}$ for every $v\in M_c(X)$ and $w\in M_d(X)$, 
\item $\|\alpha v \beta\|_d \leq \|\alpha\| \, \|v\|_d \, \|\beta\|$ for every $\alpha, \beta \in M_d$ and $v\in M_d(X)$. 
\end{itemize}
A sequence of matrix norms satisfying both conditions, or alternatively an explicit inclusion of $X$ into $B(\h)$, which automatically yields such a sequence, is called an \emph{operator space structure} (o.s.s.) on $X$.

A given Banach space $X$ can be endowed with different o.s.s. Important examples are the \emph{row} and \emph{column} o.s.s.\ on $\ell_2^n$. The space $\ell_2^n$ can be viewed as a subspace $R_n$ of $M_n$ via the map $e_i \mapsto E_{1i}=|1\rangle\langle i|$, $i=1,\cdots, n$, where we use the Dirac notation $\ket{i} = e_i$ and $\bra{i}=e_i^*$ that is standard in quantum information theory. Thus each vector $u\in\ell_2^n$ is identified with the matrix $A_u$ which has that vector as its first row and zero elsewhere; clearly $\|u\|_2=\|A_u\|$ and the embedding is norm-preserving. Or we can use the map $e_i \mapsto E_{i1}=|i\rangle\langle 1|$, $i=1,\cdots, n$, identifying $\ell_2^n$ with the subspace $C_n$ of $M_n$ of matrices that have all but their first column set to zero. 
An element $A\in M_d(R_n)$ is a $d\times d$ matrix whose each entry is an $n\times n$ matrix that is $0$ except in its first row. Alternatively, $A$ can be seen as an $n\times n$ matrix whose first row is made of $d\times d$ matrices $A_1,\ldots,A_n$ and all other entries are zero. The operator $A$ acts on $M_d(M_n)\simeq M_{dn}$ by block-wise matrix multiplication, and one immediately verifies that the corresponding sequence of norms can be expressed as 
\begin{align*}
\Big\|\sum_{i=1}^nA_i\otimes \ket{1}\bra{i}\Big\|_{M_d(R_n)}=\Big\|\sum_{i=1}^nA_iA_i^*\Big\|^\frac{1}{2}_{M_d}.
\end{align*}
Similarly, for $M_d(C_n)$ we obtain
\begin{align*}
\Big\|\sum_{i=1}^nA_i\otimes |i\rangle\bra{1}\Big\|_{M_d(C_n)}=\Big\|\sum_{i=1}^nA_i^*A_i\Big\|^\frac{1}{2}_{M_d}.
\end{align*}
These two expressions make it clear that the two o.s.s.\ they define on $\ell_2^n$ can be very different; consider for instance the norms of $A = \sum_{i=1}^n  \ket{i}\bra{1} \otimes e_i \in M_n(\ell_2^n)$. 

Certain Banach spaces have a natural o.s.s. This happens for the case of C$^*$-algebras which, by the GNS representation, have a canonical inclusion in $\Bounded(\h)$ for a certain Hilbert space $\h$ obtained from the GNS construction. Two examples will be relevant for this survey. The first is the space $\Bounded(\h)$ itself, for which the natural inclusion is the identity. Note that when $\h=\C^n$, $\Bounded(\C^n)$ is identified with $M_n$ and $M_d(M_n)$ with $M_{dn}$. The second example is $\ell_\infty^n =(\C^n, \|\cdot \|_\infty)$, for which the natural inclusion is given by the map $e_i \mapsto E_{ii}=|i\rangle\langle i|$, $i=1,\cdots, n$ embedding an element of $\ell_\infty^n$ as the diagonal of an $n$-dimensional matrix. This inclusion yields the sequence of norms
\begin{align}\label{infinity o.s.s.}
\Big\|\sum_{i=1}^nA_i\otimes |i\rangle\bra{i}\Big\|_{M_d(\ell_\infty^n)}=\sup_{i\in\{1,\ldots, n\}} \big\{ \|A_i\|_{M_d}  \big\}.
\end{align}
Bounded linear maps are the natural morphisms in the category of Banach spaces in the sense that two Banach spaces $X,Y$ can be identified whenever there exists an isomorphism $T:X\rightarrow Y$ such that $\|T\|\|T^{-1}\|=1$; in this case we say that $X$ and $Y$ are \emph{isometric}. When considering operator spaces the norm  on linear operators should take into account the sequence of matrix norms defined by the o.s.s. Given a linear map between operator spaces $T:X\rightarrow Y$, let $T_d$ denote the linear map
$$T_d:\,v=(v_{ij})\in M_d(X)\,\mapsto \,(\Id_{M_d} \otimes T)(v)=(T(v_{ij}))_{i,j}\in M_d(Y).$$
The map $T$ is said to be \emph{completely bounded} if its completely bounded norm if finite:
$$
\|T\|_{cb} \,:=\, \sup_{d\in\N} \|T_d\|\,<\, \infty.
$$
Two operator spaces are said \emph{completely isometric} whenever there exists an isomorphism $T:X\rightarrow Y$ such that $\|T\|_{cb}\|T^{-1}\|_{cb}=1$. The spaces $R_n$ and $C_n$ introduced earlier are isometric Banach spaces: if $t$ denotes the transpose map,
$$\|t:R_n\rightarrow C_n\|\,=\,\|t:C_n\rightarrow R_n\|=1.$$
But they are \emph{not} completely isometric; in particular it is a simple exercise to verify that
$$\|t:R_n\rightarrow C_n\|_{cb}\,=\,\|t:C_n\rightarrow R_n\|_{cb}\,=\,\sqrt{n},$$
and in fact any isomorphism $u:R_n\rightarrow C_n$ verifies $\|u\|_{cb}\|u^{-1}\|_{cb}\geq n$.

For C$^*$-algebras ${A}$, ${B}$ with unit a linear map $T:{A}\rightarrow {B}$ is \emph{completely positive} if $T_d(x)=(\Id_{M_d}\otimes T) (x)$ is a positive element in $M_d({B})$ for every $d$ and every positive element $x\in M_d({A})$, and it is \emph{unital} if $T(\Id_{A})=\Id_{B}$. Although every completely positive map is trivially positive, the converse is not true (a typical example is given by the transpose map on $M_n$). However, if $A$ is a commutative C$^*$-algebra (such as $\ell_\infty^n$) it is straightforward to check that positivity implies complete positivity. The following very standard lemma will be used often, and we include a short proof. 

\begin{lemma}\label{lem:cp-unital}
Let $T:A\rightarrow B$ be a completely positive and unital map between C$^*$-algebras. Then 
$$\|T\|\,=\,\|T\|_{cb}\,=\,1.$$
\end{lemma}

\begin{proof}
The inequalities $1\leq \|T\|\leq \|T\|_{cb}$ hold trivially for any unital map. To show the converse inequalities, first recall that an element $z$ in a unital $C^*$-algebra $D$ verifies
\begin{align}\label{positivity and norm one}
\|z\|\leq 1 \text{    if and only if   }\begin{pmatrix} \Id_{D}  & z\\z^* & \Id_{D}\end{pmatrix}\text{    is a positive element in   } M_2(D).
\end{align}This can be easily seen by considering the canonical inclusion of the C$^*$-algebra $M_2({D})$ in $M_2(\Bounded(\h))=\Bounded(\h\oplus \h)$. 
Now given an element $x\in M_d({A})$ such that $\|x\|\leq 1$ we use the positivity of 
$$\begin{pmatrix} \Id_{M_d({A})}  & x\\ x^* & \Id_{M_d({A})}\end{pmatrix}\in M_2(M_d({A}))$$
 and the fact that $T$ is completely positive and unital to conclude that 
$$T_{2d} \begin{pmatrix} \Id_{M_d({A})}  & x\\x^* & \Id_{M_d({A})}\end{pmatrix}=\begin{pmatrix}  \Id_{M_d({B})} & T_d (x)\\T_d(x)^* & \Id_{M_d({B})}\end{pmatrix}$$
 is a positive element in $M_2(M_d({B}))$. Using~(\ref{positivity and norm one}) again we conclude that $\|T_d (x)\|\leq 1$ for every $d$, thus $\|T\|_{cb}\leq 1$. 
\end{proof}


Given an operator space $X$ it is possible to define an o.s.s.\ on $X^*$, the dual space of $X$. The norms on $M_d(X^*)$ are specified through the natural identification with the space of linear maps from $X$ to $M_d$, according to which to an element $x=\sum_i a_i\otimes x_i^*\in M_d(X^*)$ is associated the map 
$$T^x:\,v\in X\mapsto \sum_i x_i^*(v)a_i\in M_d.$$
This leads to the sequence of norms
\begin{align}\label{dual o.s.s.}
\|x\|_{M_d(X^*)}=\|T^x:X\rightarrow M_d\|_{cb},\qquad d\geq 1.
\end{align}
For every linear map $T:X\rightarrow M_d$ it holds that $\|T\|_{cb}=\|T_d\|$~\cite[Proposition 8.1]{PaulsenBook}, so in this case the completely bounded norm is attained by considering amplifications up to dimension $d$. We encourage the reader to check the (completely isometric) identifications $C_n^* =R_n$ and $ R_n^*=C_n$. 

Duality allows us to introduce a natural o.s.s.\ on the spaces $\ell_1^n$ and $S_1^n$ as dual spaces of $\ell_\infty^n$ and $M_n$ respectively. From~(\ref{infinity o.s.s.}) and~(\ref{dual o.s.s.}) one immediately obtains that\footnote{
Note that here the o.s.s.\ is defined without explicitly specifying an embedding of $\ell_1^n$ in some $\Bounded(\h)$. Although Ruan's theorem assures that such an embedding must exist that leads to the sequence of norms (\ref{o.s.s. ell_1}), finding that embedding explicitly can be a difficult problem. In particular, for $\ell_1^n$ the simplest embedding is based on the universal C$^*$-algebra associated to the free group with $n$ generators $C^*(\mathbb F_n)$.} 
\begin{align}\label{o.s.s. ell_1}
\Big\|\sum_{i=1}^nA_i\otimes |i\rangle\bra{i}\big\|_{M_d(\ell_1^n)}=\sup \big\|\sum_{i=1}^nA_i\otimes B_i\Big\|_{M_{d^2}},
\end{align}
where the supremum runs over all families of operators $\{B_i\}_i$ in $M_d$ such that $\sup_i\|B_i\|\leq 1$. Note that, by convexity, the supremum in~\eqref{o.s.s. ell_1} can be restricted to families of unitaries $\{U_i\}_i$ in $M_d$.
An o.s.s.\ on $S_1^n$ can similarly be defined by introducing norms on $M_d(S_1^n)$ through~(\ref{dual o.s.s.}). The explicit form of these norms will not play a role in this survey, and we omit the (cumbersome) definition.
\subsection{Bilinear forms and tensor products}
\label{sec:bil-tensor}
As in the case of linear maps, working with bilinear maps on operator spaces requires the introduction of a norm on such maps which captures the o.s.s. Given a bilinear form on operator spaces $B:X\times Y\rightarrow \C$, for every $d$ define a bilinear operator 
$$B_d=B\otimes\Id_{M_d} \otimes \Id_{M_d}:M_d(X)\times  M_d(Y)\rightarrow M_{d^2}$$
by $B(a\otimes x,b\otimes y)=B(x,y)a\otimes b$ for every $a,b\in M_d$, $x\in X$, $y\in Y$. We say that $B$ is \emph{completely bounded} if its completely bounded norm is finite: $$\|B\|_{cb}\,:=\,\sup_d\|B_d\|<\infty.$$
A weaker condition is that $B$ is \emph{tracially bounded}:
$$\|B\|_{tb}\,:=\,\sup \Big\{\big|\langle \psi_d|B_d(A,B)|\psi_d\rangle\big|: d\in \N, A\in \Ball(M_d(X)), B\in \Ball(M_d(Y))\Big\}<\infty,$$
where here $|\psi_d\rangle=\frac{1}{\sqrt{d}}\sum_{i=1}^d|ii\rangle$ is the maximally entangled state in dimension $d$. 

The natural one-to-one correspondence between bilinear forms and tensor products will play an important role in this survey. It is obtained through the identifications
\begin{align}\label{eq:bil-tensor}
\Bil(X,Y)=(X\otimes Y)^*=L(X,Y^*).
\end{align}
Here the first equality is obtained by identifying a bilinear form $B:X\times Y\rightarrow \C$ with the linear form $S_B:X\otimes Y\rightarrow \C$ defined by $S_B(x\otimes y)=B(x,y)$ for every $x,y$, and the second equality by identifying $B$ with the linear map $T_B:X\rightarrow Y^*$ defined by $\langle T_B(x),y\rangle=B(x,y)$ for every $x,y$. 
If the spaces $X$, $Y$ are finite dimensional, one also has the natural algebraic identification
\begin{align}\label{bilinear-tensor}
\Bil(X,Y)=X^*\otimes Y^*.
\end{align}
To be more precise, fix bases $(e_i)$ and $(f_j)$ for $X$ and $Y$ respectively, and let $(e_i^*)$ and $(f_j^*)$ be the dual bases. If $B:X\times Y\rightarrow \C$ is a bilinear form such that $B(e_i,f_j)=b_{i,j}$ for every $i,j$, its associated tensor is given by $\hat{B}=\sum_{i,j}b_{i,j}e_i^*\otimes f_j^* \in X^*\otimes Y^*$. Conversely, given an element $x=\sum_{s,t}x_{s,t}e_s^*\otimes f_t^*\in X^*\otimes Y^*$, its natural action on $X\times Y$ is defined by $\tilde{x}(e_i,f_j)=\sum_{s,t}x_{s,t}e_s^*(e_i)f_t^*(f_j)=x_{i,j}$ for every $i,j$.

The identifications~\eqref{bilinear-tensor} can be made isometric. Given two Banach spaces $X$, $Y$, define the \emph{injective tensor norm} of $z\in X\otimes Y$ as
\begin{align}\label{Def injective norm}
\|z\|_{X\otimes_\epsilon Y}=\sup_{x^*\in\Ball(X^*),\,y^*\in\Ball(Y^*)}\big|\langle z, x^*\otimes y^*\rangle\big|.
\end{align}
With this definition whenever $X$ and $Y$ are finite-dimensional it holds that for every bilinear form $B:X\times Y\rightarrow \C$ we have 
\begin{equation}\label{eq:bil-epsilon-norm}
\|B\|=\|\hat{B}\|_{X^*\otimes_\epsilon Y^*},
\end{equation}
where $\hat{B}\in X^*\otimes  Y^*$ is the tensor associated to the bilinear form $B$ according to~(\ref{bilinear-tensor}).

The correspondence can be extended to the case of operator spaces $X$, $Y$. Define the \emph{minimal tensor norm} of $z\in X\otimes Y$ as $\|z\|_{X\otimes_{min} Y}=\sup_d  \|z\|_{X\otimes_{min_d} Y}$, where
\begin{align}\label{Def min-norm}
\|z\|_{X\otimes_{min_d} Y}=\sup_{T\in\Lin(X,M_d),S\in\Lin(Y,M_d):\,\|T\|_{cb},\|S\|_{cb}\leq 1} \big\|(T\otimes S)(z)\big\|_{M_{d^2}}.
\end{align}
With this definition it is again easily verified that for any bilinear form $B:X\times Y\rightarrow \C$,
$$\|B\|_{cb}=\|\hat{B}\|_{X^*\otimes_{min} Y^*}.$$
Finally for the case of the tracially bounded norm we define  $\|z\|_{X\otimes_{\psi-min} Y}=\sup_d \|z\|_{X\otimes_{\psi-min_d} Y}$, where 
\begin{align}\label{Def tb norm}
\|z\|_{X\otimes_{\psi-min_d} Y}=\sup_{k\leq d,\,T\in\Lin(X,M_k),S\in\Lin(Y,M_k):\,\|T\|_{cb},\|S\|_{cb}\leq 1} \big|\langle\psi_k|(T\otimes S)(z)|\psi_k\rangle\big|,
\end{align}
in which case it holds that 
$$\|B\|_{tb}=\|\hat{B}\|_{X^*\otimes_{\psi-min} Y^*}.$$

\subsection{Bell functionals and multiplayer games}\label{Sec: Bell ineq and multiplayer}

In this section we introduce the basic definitions and notation associated with the main objects this survey is concerned with, Bell functionals and multiplayer games. It will soon be apparent that the two names cover essentially the same concept, although important differences do arise in some contexts, as will be pointed out later in the survey. 

\subsubsection{Bell functionals}

Bell functionals are linear forms acting on multipartite conditional probability distributions, and Bell inequalities are bounds on the largest value taken by a Bell functional over a restricted set of distributions. For clarity we focus on the bipartite case, but the framework extends easily. Given finite sets $\setx$ and $\seta$ denote by $\prob(\seta|\setx)$ the set of conditional distributions from $\setx$ to $\seta$, 
\begin{align*}
\prob(\seta|\setx) &= \Big\{P_{\seta|\setx}=\big(P_{\seta|\setx}(a|x)\big)_{x,a} \in \R_+^{\carda\cardx}:\,\forall x\in \setx,\sum_{a\in \seta}P_{\seta|\setx}(a|x)=1\Big\}.
\end{align*}
 We often drop the subscript in $P_{\seta|\setx}$ whenever the underlying sets are clear from the context. When considering bipartite conditional distributions we use the shorthand $\seta\setb$ for $\seta\times\setb$, i.e. $\prob(\seta\setb|\setx\sety)=\prob(\seta\times\setb|\setx\times\sety)$. 

A \emph{Bell functional} $M$ is simply a linear form on $\R^{\carda\cardb\cardx\cardy}$. Any such functional is specified by a family of real coefficients $M=(M_{x,y}^{a,b})_{x,y;a,b}\in\R^{\seta\times\setb\times\setx\times\sety}$, and its action on $\prob(\seta\setb|\setx\sety)$ is given by 
\begin{equation}\label{eq:val-def-bell}
P\in \prob(\seta\setb|\setx\sety) \mapsto \val(M;P) :=\sum_{x,y;a,b}M_{x,y}^{a,b}P(a,b|x,y)\in\R.
\end{equation} 
We refer to $\cardx$ and $\carda$ (resp. $\cardy$ and $\cardb$) as the number of \emph{inputs} and \emph{outputs} to and from the first (resp. second) \emph{system} acted on by the Bell functional. 
A \emph{Bell inequality} is an upper bound on the largest value that the expression~\eqref{eq:val-def-bell} can take when restricted to the subset of $\prob(\seta\setb|\setx\sety)$ corresponding to  
\emph{classical} conditional distributions. Informally, classical distributions are those that can be implemented locally with the help of shared randomness. Formally, they correspond to the convex hull of product distributions,
\begin{align}\label{eq:def-classical-prob}
\class(\seta\setb|\setx\sety) &= \text{Conv}\big\{P_{\seta|\setx}\times P_{\setb|\sety}:\, P_{\seta|\setx}\in\prob(\seta|\setx),\,P_{\setb|\sety}\in\prob(\setb|\sety)\big\}.
\end{align}
Thus a Bell inequality is an upper bound on the quantity 
\begin{equation}\label{eq:class-value-def}
\omega(M) := \sup_{P\in \class(\seta\setb|\setx\sety)} |\val(M;P)|.
\end{equation}
We refer to $\omega(M)$ as the \emph{classical value} of the functional $M$. The second value associated to a Bell functional is its \emph{entangled value}, which corresponds to its supremum over the subset of $\prob(\seta\setb|\setx\sety)$ consisting of those distributions that can be implemented locally using measurements on a bipartite quantum state: 
\begin{align*}
\ent(\seta\setb|\setx\sety) &= \Big\{ \big(\bra{\psi} A_x^a\otimes B_y^b \ket{\psi}\big)_{x,y,a,b}:\, d\in\N,\,\ket{\psi}\in \Ball(\C^d\otimes\C^d),\, A_x^a,B_y^b\in\Pos(\C^d),\\
&\hskip5cm\,\sum_a A_x^a = \sum_b B_y^b=\Id \,\,\forall(x,y)\in\setx\times\sety\Big\}.
\end{align*}
Here the constraints $A_x^a\in\Pos(\C^d)$, $\sum_a A_x^a = \Id$ for every $x$ correspond to the general notion of measurement called \emph{positive operator-valued measurement} (POVM) in quantum information. Note that given a POVM $\{A_x^a\}$ the linear map $T_x:\ell_\infty^\carda\to M_d$, $T_x: e_a \mapsto A_x^a$ is completely positive and unital, thus by Lemma~\ref{lem:cp-unital} $\|T\|_{cb}=1$; conversely any completely positive and unital map $T:\ell_\infty^\carda\to M_d$ defines a POVM.

With this definition of the set $\ent(\seta\setb|\setx\sety)$, the entangled value of $M$ is defined as
\begin{equation}\label{eq:def-entangled-value}
\omega^*(M) := \sup_{P\in \ent(\seta\setb|\setx\sety)} |\val(M;P)|.
\end{equation}
Since $\class(\seta\setb|\setx\sety)\subseteq \ent(\seta\setb|\setx\sety)$ it is clear that $\omega(M)\leq\omega^*(M)$ in general. A \emph{Bell inequality violation} corresponds to the case when the inequality is strict: those $M$ such that $\omega(M)<\omega^*(M)$ serve as ``witnesses'' that the set of quantum conditional distributions is strictly larger than the classical set. 

\subsubsection{Multiplayer games}

\emph{Multiplayer games} are the sub-class of Bell functionals $M$ such that all coefficients $M_{x,y}^{a,b}$ are non-negative and satisfy the normalization condition $\sum_{x,y;a,b} M_{x,y}^{a,b}=1$. For this reason we will sometimes refer to games as ``Bell functionals with non-negative coefficients''. Beyond a mere change of language (inputs and outputs to the systems will be referred to as \emph{questions} and \emph{answers} to the \emph{players}), the fact that such functionals can be interpreted as games allows for fruitful connections with the extensive literature on this topic developed in computer science and leads to many interesting constructions.

A two-player one-round game $G=(\setx,\sety,\seta,\setb,\pi,V)$ is specified by finite sets $\setx,\sety,\seta,\setb$, a probability distribution $\pi:\setx\times \sety\rightarrow [0,1]$ and a payoff function $V: \seta\times \setb\times \setx\times \sety\rightarrow [0,1]$.\footnote{Sometimes the function $V$ is required to take values in $\{0,1\}$. Our slightly more general definition can be interpreted as allowing for randomized predicates.} One should think of the ``game'' as proceeding as follows. There are three interacting parties, the \emph{referee} and two \emph{players}, customarily named ``Alice'' and ``Bob''. The referee initiates the game by selecting a pair of questions $(x,y)\in\setx\times\sety$ according to the distribution $\pi$, and  sends $x$ to Alice and $y\in\sety$ to Bob. Given their question, each of the players is required to provide the referee with an answer $a\in\seta$ and $b\in\setb$ respectively. Finally the referee  declares  that the players ``win'' the game with probability precisely $V(a,b,x,y)$; alternatively one may say that the players are attributed a ``payoff'' $V(a,b,x,y)$ for their answers. The value of the game is the largest probability with which the players can win the game, where the probability is taken over the referee's choice of questions, the player's strategy (an element of $\prob(\seta\setb|\setx\sety)$), and the randomness in the referee's final decision; alternatively the value can be interpreted as the maximum expected payoff that can be achieved by the players. 

Any game induces a Bell functional $G_{x,y}^{a,b}=\pi(x,y)V(a,b,x,y)$ for every $x,y,a,b$. This allows to extend the definitions of classical and entangled values given earlier to corresponding quantities for games. Precisely, given a game $G$ we can define
\begin{align}
\omega(G) := \sup_{P\in \class(\seta\setb|\setx\sety)} \Big|\sum_{x,y;a,b} \pi(x,y)V(a,b,x,y)P(a,b|x,y)\Big|,\notag\\
\omega^*(G) := \sup_{P\in \ent(\seta\setb|\setx\sety)} \Big|\sum_{x,y;a,b} \pi(x,y)V(a,b,x,y)P(a,b|x,y)\Big|.\label{eq:def-eval}
\end{align}
Conversely, any Bell functional with non-negative coefficients can be made into a game by setting $\pi(x,y) = \sum_{a,b} M_{x,y}^{a,b}$ and $V(a,b,x,y)= M_{x,y}^{a,b}/\pi(x,y)$.


\section{XOR games}
\label{sec:xor-games}

\emph{XOR games} are arguably the simplest and best-understood class of multiplayer games that are interesting from the point of view of quantum nonlocality, i.e. the study of the set $\ent(\seta\setb|\setx\sety)$ and its multipartite generalizations. The properties of XOR games turn out to be very tightly connected to fundamental results in the theory of tensor products of Banach spaces as well as operator spaces, making them a perfect starting point for this survey. 

Recall the general definitions for games introduced in Section~\ref{Sec: Bell ineq and multiplayer}. Two-player XOR games correspond to the restricted family of games for which the answer alphabets $\seta$ and $\setb$ are binary, $\seta=\setb=\{0,1\}$, and the payoff function $V(a,b,x,y)$ depends only on $x$, $y$, and the parity of $a$ and $b$. We further restrict our attention to functions that take the form $V(a,b,x,y)=\frac{1}{2}(1+(-1)^{a\oplus b \oplus c_{xy}})$ for some $c_{xy}\in\{0,1\}$. This corresponds to deterministic predicates such that for every pair of questions there is a unique parity $a\oplus b$ that leads to a win for the players. (This additional restriction is not essential, and all results discussed in this section extend to general $V(a,b,x,y)=V(a\oplus b,x,y)\in[0,1]$.)

Given an arbitrary $P\in\prob(\seta\setb|\setx\sety)$, based on the definition~\eqref{eq:val-def-bell} the value achieved by $P$ in $G$ can be expressed as 
\begin{align*}
\omega(G;P) &= \sum_{x,y;a,b} \pi(x,y)V(a,b,x,y)P(a,b|x,y)\\
&= \sum_{x,y;a,b} \pi(x,y)\frac{1+(-1)^{a\oplus b \oplus c_{xy}}}{2}P(a,b|x,y)\\
&= \frac{1}{2} + \frac{1}{2}\sum_{x,y} \pi(x,y) (-1)^{c_{xy}} \big(P(0,0|x,y)+ P(1,1|x,y) - P(0,1|x,y)-P(1,0|x,y)\big).
\end{align*}
This last expression motivates the introduction of the \emph{bias} $\beta(G;P)=2\omega(G;P)-1\in[-1,1]$ of an XOR game, a quantity that will prove more convenient to work with than the value. Optimizing over all classical strategies one obtains the \emph{classical bias} $\beta(G)$ of an XOR game $G$, 
\begin{align}
\beta(G) &= \sup_{P\in\class(\seta\setb|\setx\sety)} |\beta(G;P)|\notag\\
&= \sup_{A\in\prob(\seta|\setx),\,B\in\prob(\setb|\sety)} \Big|\sum_{x,y} \pi(x,y)(-1)^{c_{xy}} \big(A(0|x)-A(1|x)\big)\big(B(0|y)-B(1|y)\big)\Big|\notag\\
&=\sup_{a\in\Ball(\ell_\infty^\cardx(\R)),\,b\in\Ball(\ell_\infty^\cardy(\R))} \Big|\sum_{x,y} \pi(x,y)(-1)^{c_{xy}}a_xb_y\Big|,\label{eq:xor-classical-bias}
\end{align}
where for the last equality recall that $\Ball(\ell_\infty^n(\R))$ denotes the unit ball of $\ell_\infty^n(\R)$, the set of all sequences $(x_i)\in\R^n$ such that $|x_i|\leq 1$ for all $i\in\{1,\ldots,n\}$. To the game $G$ we may associate a bilinear form, which (abusing notation slightly) we will denote $G:  \ell_\infty^\cardx(\R) \times \ell_\infty^\cardy(\R) \to \R$, defined on basis elements by $G(e_x,e_y)=\pi(x,y)(-1)^{c_{xy}}$. Thus~\eqref{eq:xor-classical-bias} relates the classical bias of the game $G$ to the norm of the associated bilinear form:
\begin{equation}\label{eq:xor-norm-bil}
\beta(G) =  \|G\|_{\Bil( \ell_\infty^\cardx(\R), \ell_\infty^\cardy(\R) )},
\end{equation}
marking our first connection between the theory of games and that of Banach spaces. Note that in~\eqref{eq:xor-norm-bil}
 the norm of $G$ is taken with respect to the real Banach spaces $\ell_\infty(\R)$. It will often be mathematically convenient to consider instead complex spaces, in which case we have the following.

\begin{lemma}\label{lem:xor-real-complex}
Let $G\in\Bil(\ell_\infty^\cardx(\R),\ell_\infty^\cardy(\R))$ be a real bilinear form. Then 
\begin{equation}\label{eq:xor-real-complex}
\|G\|_{\Bil(\ell_\infty^\cardx(\R),\ell_\infty^\cardy(\R))}\leq \|G\|_{\Bil(\ell_\infty^\cardx(\C),\ell_\infty^\cardy(\C))}\leq \sqrt{2}\|G\|_{\Bil(\ell_\infty^\cardx(\R),\ell_\infty^\cardy(\R))}.
\end{equation}
Moreover, for each inequality there exists a $G$ for which it is an equality.
\end{lemma}

While the first inequality in the lemma is clear, the second is non-trivial and due to Krivine~\cite{Krivine:79a}. A $G$ for which the inequality is tight is 
\begin{equation}\label{eq:chsh-def}
G = G_{CHSH} : (e_x,e_y)\mapsto \frac{1}{4}(-1)^{x\wedge y}\qquad \forall x,y\in\{0,1\}.
\end{equation}
Independently identified by Krivine, this bilinear form underlies the famous \emph{CHSH inequality} introduced a decade earlier by Clauser et al.\ \cite{Clauser:69a}. Following the pioneering work of Bell~\cite{Bell:64a}, Clauser et al.\ were the first to state a simple, finite inequality, $\beta(G_{CHSH})\leq 1/2$, for which quantum mechanics predicts a noticeable violation, $\beta^*(G_{CHSH})\geq \sqrt{2}/{2}$.\footnote{In the next section it will be seen that equality holds. The ``violation'' $\beta^*(G_{CHSH})= \sqrt{2}/{2} > 1/2$ has been experimentally demonstrated~\cite{Aspect81} and remains a fundamental benchmark of quantum nonlocality.} Here $\beta^*$ denotes the \emph{entangled bias} of an XOR game, defined as 
\begin{align}
\beta^*(G) &= \sup_{P\in\ent(\seta\setb|\setx\sety)} |\beta(G;P)|\notag\\
&= \sup_{\substack{d;\,\ket{\psi}\in \Ball(\C^d\otimes \C^d),\\
 A_x^a,B_y^b\in\Pos(\C^d),\,\sum_a A_x^a = \sum_b B_y^b=\Id }} \Big|\sum_{x,y} \pi(x,y)(-1)^{c_{xy}} \bra{\psi} \big(A_x^0-A_x^1\big)\otimes\big(B_y^0-B_y^1\big)\ket{\psi}\Big|\label{eq:xor-norm-os-0a}\\
&=\sup_{d;\,A\in\Ball(\ell_\infty^\cardx(M_d)),\,B\in\Ball(\ell_\infty^\cardy(M_d))} \Big\|\sum_{x,y} \pi(x,y)(-1)^{c_{xy}}A_x\otimes B_y\Big\| \label{eq:xor-norm-os-0}\\
&= \|G\|_{\cBil(\ell_\infty^\cardx(\C),\ell_\infty^\cardy(\C))}, \label{eq:xor-norm-os}
\end{align}
where for the last equality recall that $\Ball(\ell_\infty^n(M_d))$ denotes the unit ball of $\ell_\infty^n(M_d)$, those sequences $(x_i)\in (M_d)^n$ such that $\|x_i\|\leq 1$ for every $i\in\{1,\ldots,n\}$. In~\eqref{eq:xor-norm-os-0} the $A_x,B_y$ are allowed to be non-Hermitian, whereas as dictated by quantum mechanics in~\eqref{eq:xor-norm-os-0a} the supremum is restricted to quantum observables\footnote{An {observable} is a Hermitian matrix that squares to identity.} of the form $A_x=A_x^0-A_x^1$ for positive $A_x^0$ and $A_x^1$; nevertheless taking advantage of the unrestricted dimension $d$ the mapping $A\to\begin{pmatrix} 0 & A \\ A^* & 0\end{pmatrix}$ shows that restricting the supremum in~\eqref{eq:xor-norm-os-0} to Hermitian operators leaves it unchanged. Thus~\eqref{eq:xor-norm-os} shows that the entangled bias equals the completely bounded norm of $G$ when seen as a bilinear form on $\ell_\infty^\cardx \times\ell_\infty^\cardy$ equipped with the natural o.s.s. obtained by embedding them as diagonal matrices (q.v.~\eqref{infinity o.s.s.} in Section~\ref{sec:prelims-os}). 

As described in Section~\ref{sec:bil-tensor}, expression~\eqref{eq:xor-norm-bil} and~\eqref{eq:xor-norm-os} for the classical and entangled bias as the norm and completely bounded norm  of a bilinear form respectively may be equivalently stated in terms of the associated tensor norms. For any XOR game $G$ consider the tensor
$$\hat{G} = \sum_{x,y}\pi(x,y) (-1)^{c_{xy}} e_x\otimes e_y\in \ell_1^\cardx(\R)\otimes\ell_1^\cardy(\R).$$
From~\eqref{eq:xor-norm-bil} and~\eqref{eq:xor-norm-os} it follows immediately that 
\begin{equation}\label{eq:xor-norm-tensor}
\beta(G) = \|\hat{G}\|_{\ell_1^\cardx(\R)\otimes_\epsilon \ell_1^\cardy(\R) }\qquad\text{and}\qquad \beta^*(G) = \|\hat{G}\|_{\ell_1^\cardx \otimes_{\min} \ell_1^\cardy}.
\end{equation}

To conclude this section we observe that the correspondence goes both ways. To any tensor $G\in\ell_1^\cardx\otimes \ell_1^\cardy$ with real coefficients that satisfies the mild normalization condition $\sum_{x,y}|G_{x,y}|=1$ we may associate an XOR game by defining $\pi(x,y)=| G_{xy}|$ and $(-1)^{c_{xy}} = \text{sign}(G_{xy})$. In particular any Bell functional $M:\ell_\infty^\cardx\times\ell_\infty^\cardy\to\R$ can, up to normalization, be made into an equivalent XOR game. Such functionals are called ``Bell correlation functionals'', and in this setting there is no difference of substance between the viewpoints of Bell functionals and of games. 

\subsection{Grothendieck's theorem as a fundamental limit on nonlocality}
\label{sec:grothendieck}

Tsirelson is the first to have applied results in Banach space theory to the study of nonlocal games, and this section describes some of his contributions. First we introduce the pioneering work of Grothendieck, who initiated the systematic study of norms on the algebraic tensor product of two Banach spaces. Grothendieck introduced a notion of ``reasonable'' tensor norms,\footnote{Of course Grothendieck gave a precise meaning to ``reasonable'': it is required that the tensor norms satisfy simple compatibility conditions with the Banach space structure of the spaces to be combined.} with the two ``extremal'' such norms playing a distinguished role in the theory: the  ``smallest'' reasonable norm, the injective norm $\epsilon$, and the ``largest'' reasonable norm, the projective norm $\pi$. The injective norm is defined for a general tensor product in Section~\ref{sec:bil-tensor}: it is the norm which makes the identification (\ref{bilinear-tensor}) between tensors and bilinear forms isometric, in the sense that $\|B\|=\|\hat{B}\|_{X^*\otimes_\epsilon Y^*}$ for every bilinear form $B:X\times Y\rightarrow \C$. In turn, the projective norm $\pi$ is defined so that the identification (\ref{eq:bil-tensor}) between bilinear forms and linear maps is isometric. Here we will be concerned with an additional ``reasonable'' tensor norm, the $\gamma_2^*$ norm. Although it can be defined more generally (see e.g.~\cite[Section 3]{PisierGT}), for the case of $\ell_1^\cardx(\C)\otimes \ell_1^\cardy(\C)$ of interest here the $\gamma_2^*$  norm can be expressed as 
\begin{equation}\label{eq:xor-facto-norm}
\|C\|_{\gamma_2^*} = \sup_{d\in\N;\,a_x,b_y\in \Ball(\C^d)} \Big|\sum_{xy} c_{xy}\,a_x \cdot b_y\Big|,
\end{equation}
where without loss of generality the supremum on $d$ can be restricted to $d\leq\min(\cardx,\cardy)$. (If furthermore $C$ has real coefficients then the supremum in~\eqref{eq:xor-facto-norm} can be restricted to real Hilbert spaces without changing its value.) Grothendieck's ``Th\'eor\`eme fondamental de la th\'eorie des espaces m\'etriques'' relates this norm to the injective norm as follows:
 
\begin{theorem}[Grothendieck's inequality~\cite{Gro53}]\label{thm:grothendieck}
There exist universal constants $K_G^\R$ and $K_G^\C$ such that for any integers $n,m$ and $C_1\in\R^{n\times m}$, $C_2\in\C^{n\times m}$,
$$ \|C_1\|_{\gamma_2^*} \leq K_G^\R\, \|C_1\|_{\ell_1^\cardx(\R)\otimes_\epsilon\ell_1^\cardy(\R)}\qquad\text{and}\qquad\|C_2\|_{\gamma_2^*} \leq K_G^\C \,\|C_2\|_{\ell_1^\cardx(\C)\otimes_\epsilon\ell_1^\cardy(\C)}.$$
\end{theorem}

For an extensive discussion of Grothendieck's theorem and many generalizations we refer to the survey~\cite{PisierGT}. The precise values of $K_G^\R$ and $K_G^\C$ are not known, but $1<K_G^\C < K_G^\R < \pi/(2\log (1+\sqrt{2})) (=1.782...)$. Eq.~\eqref{eq:xor-norm-tensor} shows that if $\hat{G}$ is the tensor associated to an XOR game $G$ the injective norm $\|\hat{G}\|_\epsilon$ equals the classical bias $\beta(G)$. The following crucial observation, due to Tsirelson, relates the entangled bias, equal to the minimal norm of $\hat{G}$ by~\eqref{eq:xor-norm-tensor}, to the $\gamma_2^*$ norm: 
\begin{align}
\|\hat{G}\|_{\ell_1^\cardx \otimes_{\min} \ell_1^\cardy}&= \sup_{d;\,A_x,B_y\in\Ball(M_d(\C))} \Big\|\sum_{xy} \hat{G}_{xy}\,A_x\otimes B_y\Big\|\notag\\
&= \sup_{\substack{d;\,\ket{\psi}\in\Ball(\C^d\otimes\C^d),\\A_x,B_y\in\Ball(M_d(\C))}}\Big|\sum_{xy} \hat{G}_{xy}\, \bra{\psi}A_x\otimes B_y\ket{\psi}\Big|\notag\\
&= \sup_{\substack{d;\,\ket{\psi}\in\Ball(\C^d\otimes\C^d),\\A_x,B_y\in\Ball(M_d(\C))}}\Big|\sum_{xy} \hat{G}_{xy}\,\big( \bra{\psi}A_x\otimes\Id\big)\cdot\big(\Id\otimes B_y\ket{\psi}\big)\Big|\notag\\
&\leq \sup_{d;\,a_x,b_y\in \Ball(\C^d)} \Big|\sum_{xy} \hat{G}_{xy} \,a_x\cdot b_y\Big|\label{eq:xor-min-facto-0}\\
&= \|\hat{G}\|_{\gamma_2^*}.\label{eq:xor-min-facto}
\end{align}
Tsirelson further showed that for the case of real tensors $\hat{G}$ inequality~\eqref{eq:xor-min-facto-0} is an equality: for any \emph{real} unit vectors $(a_x)$ and $(b_y)$ there exists a state $\ket{\psi}$ and observables $A_x,B_y$ such that $a_x\cdot b_y = \bra{\psi}A_x\otimes B_y \ket{\psi}$ for all $x,y$. Tsirelson's argument is based on representations of the Clifford algebra and we refer to the original work~\cite{tsirel1987quantum} for details of the construction. With this correspondence in hand Grothendieck's fundamental inequality, Theorem~\ref{thm:grothendieck}, has the following  consequence for nonlocality.

\begin{corollary}[Tsirelson~\cite{tsirel1987quantum}]\label{cor:tsirelson}
Let $G$ be an XOR game. The largest bias achievable by entangled players is bounded as 
$$ \beta^*(G) \,\leq\, K_G^\R \,\beta(G).$$
Furthermore, for any $\eps>0$ there exists an XOR game $G_\eps$ with $2^{\poly(\eps^{-1})}$ questions per player such that $\beta^*(G_\eps) \geq (K_G^\R-\eps)\beta(G_\eps)$.  
\end{corollary} 

The ``furthermore'' part of the theorem follows from the aforementioned observation of Tsirelson that the inequality $\beta^*(G)\leq \|\hat{G}\|_{\gamma_2^*}$ obtained by combining~\eqref{eq:xor-norm-tensor} and~\eqref{eq:xor-min-facto} is an equality together with the fact, shown in~\cite{raghavendra2009towards}, that for any $\eps>0$ there exists a real tensor $\hat{G}_\eps\in\ell_1^n(\R)\otimes\ell_1^n(\R)$, where $n=2^{\poly(\eps^{-1})}$, such that $\|\hat{G}_\eps\|_{\gamma_2^*}\geq (K_G^\R-\eps)\|\hat{G}_\eps\|_\epsilon$.

\medskip

In the remainder of this section we describe some consequences of Tsirelson's characterization of the entangled bias of an XOR game through the $\gamma_2^*$ norm. An important hint that the formulation~\eqref{eq:xor-min-facto} for the entangled bias will prove amenable to  analysis comes from it being \emph{efficiently computable}, as we now explain. Using~\eqref{eq:xor-facto-norm} it is a simple exercise to verify that
\begin{equation}\label{eq:xor-facto-sdp}
\|C\|_{\gamma_2^*} = \sup_{Z\in \Pos(\C^{\cardx+\cardy}),\,Z_{ii}\leq 1\,\forall i} \Tr\big( \tilde{C} Z\big),
\end{equation}
where the supremum is taken over positive semidefinite $Z$ that have all their diagonal coefficients at most $1$ and $\tilde{C}=\frac{1}{2}\begin{pmatrix} 0 & C\\C^T & 0\end{pmatrix}$. The benefit of rewriting~\eqref{eq:xor-facto-norm} in this form is that the supremum has been linearized: the right-hand side of~\eqref{eq:xor-facto-sdp} is now the optimization of a linear function over the positive semidefinite cone, under linear constraints. This is precisely the kind of optimization problem known as a \emph{semidefinite program}, and it follows from general results in semidefinite optimization~\cite{NesterovN94interior} that for any $\eps>0$ the optimum of~\eqref{eq:xor-facto-sdp} can be approximated to within additive error $\pm\eps$ in time polynomial in $\cardx+\cardy$, the size of the coefficients of $C$, and $\log(1/\eps)$. This stands in stark contrast to the classical bias, which is known to be NP-hard to approximate to within a constant multiplicative factor --- indeed, assuming the famous ``unique games conjecture'', to within any factor smaller than the real Grothendieck constant~\cite{raghavendra2009towards}. 

Tsirelson's characterization leads to a quantitative understanding of the amount of entanglement needed to play optimally, or near-optimally, in an XOR game. Using that the  supremum in~\eqref{eq:xor-facto-norm} is always achieved by unit vectors $a_x,b_y$ of dimension $d=\min(\cardx,\cardy)$, Tsirelson's construction yields observables of dimension $2^{\lfloor \min(\cardx,\cardy)/2\rfloor}$ which, together with a maximally entangled state of the same dimension, can be used to implement an optimal strategy for the players. Based on a slightly more involved argument Tsirelson~\cite{tsirel1987quantum} showed that it is always possible to obtain strategies achieving $\bias^*(G)$ using a maximally entangled state of dimension $2^{\lfloor r/2 \rfloor}$, where $r$ is the largest integer such that ${r+1 \choose 2} < \cardx+\cardy$; Slofstra exhibits a family of games for which this bound is tight~\cite{Slofstra11xor}. Using a dimension reduction argument based on the Johnson-Lindenstrauss lemma it is possible to show that approximately optimal strategies, achieving a bias at least $\beta^*(G)-\eps$, can be found in a dimension $2^{O(\eps^{-2})}$ depending only on $\eps$ and not on the size of the game; this will be shown in Lemma~\ref{lem:xor-dim} in Section~\ref{sec:xor-entanglement}. 

We end this section with an application to the problem of \emph{parallel repetition}. Given a game $G$ (not necessarily an XOR game), define the $\ell$-th parallel repeated game $G^{(\ell)}$ as follows. The referee selects $\ell$ pairs of questions $(x_1,y_1),\ldots,(x_\ell,y_\ell)\in\setx\times\sety$ for the players, each pair chosen independently according to $\pi$. He sends $(x_1,\ldots,x_\ell)$ to the first player, Alice, and $(y_1,\ldots,y_\ell)$ to the second player, Bob. Upon receiving answers $(a_1,\ldots,a_\ell)\in\seta^\ell$ and $(b_1,\ldots,b_\ell)\in\setb^\ell$ respectively the referee accepts with probability $\prod_{i=1}^\ell V(a_i,b_i,x_i,y_i) \in [0,1]$.\footnote{Note that in case $V(a,b,x,y)\in\{0,1\}$ this corresponds to requiring that the players ``win'' each of the $\ell$ instances of the game played in parallel.} The parallel repetition problem is the following: assuming $\omega(G)<1$, does $\omega(G^{(\ell)})\to_{\ell\to\infty} 0$, and if so at which rate? 

Naturally if the players decide upon their answers independently it will be the case that their success probability in $G^{(\ell)}$ is the $\ell$-th power of their success probability in $G$. But in some cases they can do better, taking advantage of the fact that all questions are received at once: perhaps surprisingly, there exists a simple game $G$ such that $\omega(G)=\omega(G^{(2)})<\omega^*(G)=\omega^*(G^{(2)})<1$~\cite{CleveSUU08xor}. This situation is, in fact, generic, in the sense that perfect parallel repetition is known to hold only in very specific cases --- a prime example of which is the entangled bias of XOR games, as we now explain. 

For the case of XOR games it is more natural to consider what is known as the \emph{direct sum} property.\footnote{General reductions relating the property of parallel repetition to that of direct sum are known; see e.g.~\cite{impagliazzo2010constructive}.} Let $G^{(\oplus \ell)}$ be defined as $G^{(\ell)}$ except that the players win if and only if the parity of all their answers equals the required parity across the $\ell$ instances, i.e. $(a_1\oplus b_1)\oplus\cdots \oplus(a_\ell\oplus b_\ell) = c_{x_1y_1}\oplus \cdots \oplus c_{x_\ell y_\ell}$. Recalling the definition of the tensor $\hat{G}_{xy}=\pi(x,y)(-1)^{c_{xy}}$ associated to $G$ we see that $\hat{G}^{(\oplus\ell)} = \hat{G}\otimes \cdots \otimes \hat{G} = \hat{G}^{\otimes \ell}$. Together with the relations~\eqref{eq:xor-norm-tensor} the direct sum problem for XOR games is equivalent to the problem of relating $\|\hat{G}^{\otimes \ell}\|_\epsilon$ to $\|\hat{G}\|_\epsilon$ (for the classical bias) and $\|\hat{G}^{\otimes \ell}\|_{\min}$ to $\|\hat{G}\|_{\min}$ (for the entangled bias).

The reader familiar with the $\gamma_2^*$ norm will be aware that for the case of $\ell_1^\cardx\otimes \ell_1^\cardy$ it possesses the tensoring property
$$ \|C_1\otimes C_2\|_{\gamma_2^*} = \|C_1\|_{\gamma_2^*} \cdot \|C_2\|_{\gamma_2^*}$$
for any $C_1 $ and $C_2$, as can be verified directly from the definition (see~\cite{CleveSUU08xor} for a duality-based proof). Thus the entangled bias of XOR games satisfies a perfect direct sum property, $\beta^*(G^{(\oplus\ell)})=(\beta^*(G))^\ell$.

Interestingly, the combination of Grothendieck's inequality and the direct sum property for the entangled bias implies that the classical bias does \emph{not} itself satisfy the direct sum property. Indeed, consider any XOR game $G$ such that $\bias^*(G) > \bias(G)$, such as the CHSH game~\eqref{eq:chsh-def}. If the classical bias were multiplicative the ratio $\bias^*(G^{(\oplus\ell)})/ \bias(G^{(\oplus\ell)})$ would go to infinity with $\ell$, violating Grothendieck's inequality. The behavior of $\bias(G^{(\oplus\ell)})^{1/\ell}$ as $\ell$ goes to infinity is not completely understood; see~\cite{barak2008rounding} for an analysis of the limiting behavior of $\omega(G^{(\ell)})^{1/\ell}$ when $\omega(G)$ is close to $1$.

\subsection{Three-player XOR games: unbounded violations}
\label{sec:three-xor}

The correspondence established in the previous section between the classical and entangled bias of a two-player XOR game and the bounded and completely bounded norm of the associated bilinear form respectively  suggests a tempting ``recipe'' for obtaining large violations of locality in quantum mechanics, i.e. constructing games such that $\omega^*(G)\gg\omega(G)$. First, find operator spaces $X$ and $Y$ such that the injective and minimal norms on their tensor product are not equivalent, in the sense that there exist families of tensors for which the ratio of the latter over the former grows arbitrarily. Second, investigate whether tensors on $X\otimes Y$ can be associated to games in a way that the injective and  minimal norms correspond to natural quantities of the obtained games --- ideally, the classical or entangled value. This second step is, of course, rather subtle, and most combinations of norms will in general \emph{not} give rise to a natural quantity from the point of view of quantum information. 

The first to have applied (indeed, uncovered) the ``recipe'' outlined above are Perez-Garcia et al.~\cite{PerezWPVJ08tripartite}, who consider the case of \emph{three-player} XOR games. Motivated by the study of possible trilinear extensions of Grothendieck's inequality, they prove the following. 

\begin{theorem}\label{thm:3xor-largeviolation}
For every integer $n$ there exists\footnote{The proof given in~\cite{PerezWPVJ08tripartite} is highly non-constructive and only guarantees the existence of $T$.} an $N$ and a trilinear form $T:\ell_\infty^{2^{n^2}}\times\ell_\infty^{2^{N^2}}\times\ell_\infty^{2^{N^2}}\to\C$ such that
$$\|T\|_{\cBil(\ell_\infty^{2^{n^2}},\ell_\infty^{2^{N^2}},\ell_\infty^{2^{N^2}})} \geq \big\|T\otimes \Id_{M_{n}}\otimes \Id_{M_{N}}\otimes \Id_{M_{N}}\big\| =\Omega\big(\sqrt{n} \big) \|T\|_{\Bounded(\ell_\infty^{2^{n^2}},\ell_\infty^{2^{N^2}},\ell_\infty^{2^{N^2}})}.$$
Moreover, $T$ can be taken with real coefficients.
\end{theorem}

The separation between the injective and the minimal norm established in this theorem can be interpreted as the \emph{absence} of a tripartite generalization of Grothendieck's inequality. Nevertheless, as shown in Lemma~\ref{lem:3xor-questionbound} below the violation cannot grow arbitrarily fast, and it  cannot exceed the product of the square roots of the dimension of each space.

It is easily verified that the correspondences~\eqref{eq:xor-norm-bil} and~\eqref{eq:xor-norm-os} from the previous section extend directly to the case of three (or more) copies of the space $\ell_\infty$. In particular to the trilinear map $T$ whose existence is promised by Theorem~\ref{thm:3xor-largeviolation} can be associated a three-player XOR game $T$ whose classical and entangled biases equal the bounded and completely bounded norm of $T$ respectively (up to normalization by $\sum_{xyz}|T(e_x,e_y,e_z)|$): in this game, answers $a_x,b_y,c_z\in\{0,1\}$ to questions $x,y,z$ are accepted if and only if $(-1)^{a_x\oplus b_y \oplus c_z} = \hat{T}_{xyz} = {\text{sign}(T(e_x,e_y,e_z))}$. 

Thus Theorem~\ref{thm:3xor-largeviolation} implies that tripartite Bell correlation inequalities can be violated by unbounded amounts in quantum mechanics. We state this important consequence as the following corollary using the language of games, with an improved dependence on the number of questions obtained in~\cite{BV12}. 

\begin{corollary}\label{cor:3xor-largeviolation}
There exists a $C>0$ such that for any $n$ there is a three-player XOR game $G$ with $n^2$ questions per player such that 
$$\beta^*(G) \geq C\frac{\sqrt{n}}{\log^{3/2} n}\,\beta(G).$$
Furthermore there exists a quantum strategy achieving this lower bound using an entangled state of local dimension $n$ per player. 
\end{corollary}

Both Theorem~\ref{thm:3xor-largeviolation} and Corollary~\ref{cor:3xor-largeviolation} extend to any number $k\geq 3$ of players, giving a violation of order $\tilde{\Omega}(n^{\frac{k-2}{2}})$, where the $\tilde{\Omega}$ notation suppresses polylogarithmic factors. As already mentioned it is possible to show that this dependence of the violation on the number of questions is at most quadratically far from optimal. 

\begin{lemma}\label{lem:3xor-questionbound}
Let $G$ be a $k$-player XOR game in which the number of questions to the $i$-th player is $n_i$. Then 
$$\beta^*(G) = O(\sqrt{n_1\cdots n_{k-2}}\big)\beta(G).$$
\end{lemma}

The proof of Lemma~\ref{lem:3xor-questionbound} is based on a simple ``decoupling trick'' that allows one to track the evolution of the maximum bias as the players are successively constrained to apply classical strategies until only two players are quantum, at which point Grothendieck's inequality can be applied. We refer to~\cite{BV12} for details. 

In the remainder of this section we sketch the proof of Corollary~\ref{cor:3xor-largeviolation} (which implies Theorem~\ref{thm:3xor-largeviolation} with the improved parameters), with some further simplifications from~\cite{Pisiernote}.\footnote{See~\cite[Section 2]{palazuelos2015random} for a discussion of the differences between these proofs.}
Fix an integer $n$. The argument is probabilistic: we give a randomized construction for a game $G=G_n$ for which the claimed violation holds with high probability. The construction of $G$ proceeds in two steps. 
\begin{enumerate}
\item The first step is a probabilistic argument for the existence of a tensor $T\in (\C^{n})^{\otimes 6}$ that has certain useful spectral properties: $T$ has a substantially larger norm when interpreted as a tensor $T\in \C^{n^3}\otimes \C^{n^3}$ than when interpreted as a tensor $T\in \C^{n^2}\otimes \C^{n^2}\otimes \C^{n^2}$. Denote the former norm by $\|T\|_{2,\epsilon}$  and the latter by $\|T\|_{3,\epsilon}$.  
\item Starting from any tensor $T\in (\C^{n})^{\otimes 6}$, the second step gives a construction of a game $G$ such that the ratio $\beta^*(G)/\beta(G)$ is lower bounded as a function of the ratio of the norms $\|T\|_{2,\epsilon}/\|T\|_{3,\epsilon}$. (This construction is different from the direct correspondence between tensor and XOR game described earlier on.) The construction of $G$ from $T$ is explicit and deterministic. 
\end{enumerate}

We give more details on each step. 

\paragraph{Step 1.}
 Consider the re-ordering map $J$
\beq\label{eq:defJ}
 (\C^{n_1} \otimes \C^{m_1}) \otimes_\epsilon (\C^{n_2} \otimes \C^{m_2})\otimes_\eps (\C^{n_3} \otimes \C^{m_3}) \,\overset{J}{\to}\, (\C^{n_1} \otimes \C^{n_2}\otimes \C^{n_3}) \otimes_\epsilon (\C^{m_1} \otimes \C^{m_2}\otimes \C^{m_3}), 
\eeq
where we distinguished the dimension of each space in order to make their identification easier. The domain of $J$ is normed with $\|\cdot\|_{3,\epsilon}$ and its range with $\|\cdot\|_{2,\epsilon}$. A simple inductive argument, considering each of the six complex spaces in sequence, shows that the norm of $J$ is bounded as $\|J\|\leq \sqrt{n_2m_2}\sqrt{n_3m_3}$. The first step of the proof shows that this upper bound is almost tight by constructing a tensor $T$ such that $\|T\|_{2,\epsilon} = \tilde{\Omega}(\sqrt{n_2m_2}\sqrt{n_3m_3})\|T\|_{3,\epsilon}$. 

The definition of $T$ is as follows. Let $g\in \Ball(\C^{n_1n_2n_3})$ and $g'\in\Ball(\C^{m_1m_2m_3})$ be uniformly random unit vectors chosen according to the Haar measure,\footnote{Alternatively, random Gaussian vectors will achieve the same effect.} and let $T_{ii'jj'kk'}=g_{ijk}g'_{i'j'k'}$. Since $\sum T_{ii'jj'kk'} \overline{g_{ijk}}\overline{g'_{i'j'k'}} = \|g\|^2\|g'\|^2=1$ it is immediate that $\|T\|_{2,\epsilon} \geq 1$. 
It remains to give an upper bound on 
\begin{equation}\label{eq:3xor-ub-1}
\|T\|_{3,\epsilon}=\sup_{U\in\Ball(\C^{n_1m_1}),V\in\Ball(\C^{n_2m_2}),W\in\Ball(\C^{n_3m_3})} \Big|\sum_{ii'jj'kk'} g_{ijk} g'_{i'j'k'} U_{ii'} V_{jj'} W_{kk'}\Big|.
\end{equation}
A bound that holds with high probability over the choice of $g,g'$ can be derived  based on a delicate concentration argument that we now sketch. Assume for simplicity that $n_1=n_2=n_3=m_1=m_2=m_3=n$. Consider first the supremum in~\eqref{eq:3xor-ub-1} restricted to $U,V,W$ in $\Ball(\C^{n^2})$ such that, interpreted as matrices in $M_n$, $U,V$ and $W$ have all their singular values in $\{0,1\}$ and are of fixed ranks $r,s$ and $t$ respectively. Under this assumption it is possible to write $U=r^{-1/2}\sum_a u^a (\hat{u}^a)^*,V=s^{-1/2}\sum_b v^b(\hat{v}^b)^*,W=t^{-1/2}\sum_c w^c (\hat{w}^c)^*$ for some choice of orthonormal families $\{u^a\}$ and $\{\hat{u}^a\}$, $\{v^b\}$ and $\{\hat{v}^b\}$, and $\{w^c\}$ and $\{\hat{w}^c\}$ in $\C^n$ respectively. The expression appearing on the right-hand-side of~\eqref{eq:3xor-ub-1}  factors as
\begin{equation}\label{eq:3xor-ub-2}
\sum_{ii'jj'kk'} g_{ijk}g'_{i'j'k'} U_{ii'} V_{jj'} W_{kk'} = \frac{1}{\sqrt{rst}}\sum_{a,b,c} \Big(\sum_{ijk} g_{ijk} u^a_iv^b_jw^c_k\Big)\Big(\sum_{i'j'k'} g'_{i'j'k'} \overline{\hat{u}_{i'}^a\hat{v}_{j'}^b\hat{w}_{k'}^c}\Big).
\end{equation}
For fixed $U,V,W$ and a uniformly random choice of $g,g'$ the expression in~\eqref{eq:3xor-ub-2} has expectation $0$. Moreover, it follows from a concentration bound due to Lata{\l}a~\cite[Corollary 1]{Lataa2006} that its modulus has tails that decay at an exponential rate governed by the largest singular value as well as the Frobenius norm of $\sum_{a,b,c} (u^a\otimes v^b\otimes w^c)(\hat{u}^a\otimes \hat{v}^b\otimes \hat{w}^c)^*\in M_{n^3}$. Combining the tail bound with a union bound over a suitable $\eps$-net over the set of all projections of rank $r,s$ and $t$ the supremum of~\eqref{eq:3xor-ub-1} (restricted to projections of the specified rank) is bounded by $C/n^{2}$ for some universal constant $C$. The precise combination of the tail bound and the union bound is rather delicate, and a  careful case analysis over the possible values of $(r,s,t)$ is required, which is the reason for making that distinction in the first place. 

Once the supremum in~\eqref{eq:3xor-ub-1} has been bounded for the case of projections the extension to all $U,V,W$ of norm $1$ is obtained by decomposing an arbitrary $U$ (resp. $V$,$W$) as $U = \sum_x \alpha_x U_x$ where each $U_x$ is a projection and $\sum_x |\alpha_x| \leq 4\sqrt{\ln N}$. Combining the three factors one incurs a loss of a factor $4^3(\ln N)^{3/2}$, ultimately establishing an upper bound $\|T\|_{3,\eps} = O(n^{-2}\ln^{3/2} n)$ that holds with high probability over the choice of the coefficients of $T$ as described above. 

\paragraph{Step 2.}
In the second step, given any tensor $T$ a game $G$ is constructed such that 
\begin{equation}\label{eq:3xor-ub-3}
\beta^*(G) = \Omega(n^{-3/2})\frac{\|T\|_{2,\epsilon}}{\|T\|_{3,\epsilon}}\beta(G).
\end{equation}
For simplicity assume that $T$ is Hermitian; considering $T+T^*$ or $T-T^*$ allows a reduction to this case. Fix an orthogonal basis $\mathcal{S}=\{P_i\}_{i=1,\ldots,n^2}$ of $M_n$ for the Hilbert-Schmidt inner product such that each $P_i$ is Hermitian and squares to identity (for example, the standard Pauli basis). For any triple $(P,Q,R)\in\mathcal{S}^3$ define 
$$G_{P,Q,R}=\langle P \otimes Q \otimes R,T\rangle = \sum_{ii'jj'kk'} T_{ii'jj'kk'} \overline{P_{ii'} Q_{jj'} R_{kk'}} \in\R.$$
 The game proceeds as follows: the referee asks the triple of questions $(P,Q,R)$ with probability  $Z^{-1}|G_{P,Q,R}|$, where $Z=\sum |G_{P,Q,R}|$ is the appropriate normalization factor. Each player answers with a single bit. The referee accepts the answers if and only if their parity matches the sign of $G_{P,Q,R}$.
The classical bias of $G$ is 
\begin{align}
 \beta(G) &= \max_{x_P,y_Q,z_R\in\{\pm 1\}} Z^{-1}\Big| \sum_{P,Q,R\in\mathcal{S}} G_{P,Q,R} x_P y_Q z_R \Big| \notag\\
&=\max_{x_P,y_Q,z_R\in\{\pm 1\}} Z^{-1}\Big| \sum_{ii'jj'kk'} T_{ii'jj'kk'} \Big(\sum_{P\in\mathcal{S}}  x_P \overline{P}\Big)_{ii'}\Big(\sum_{Q\in\mathcal{S}} y_Q \overline{Q}\Big)_{jj'} \Big(\sum_{R\in\mathcal{S}} z_R \overline{R}\Big)_{kk'} \Big|,\label{eq:3xor-ub-4}
\end{align}
from which it follows that $\beta(G)\leq Z^{-1} n^{9/2}\|T\|_{3,\epsilon}$. To lower bound the entangled bias $\beta^*(G)$ it suffices to exhibit a specific strategy for the players. For the entangled state we take any state $\ket{\psi}$ that is an eigenvector of the map $T:\C^{n^3}\to\C^{n^3}$ associated with its largest eigenvalue. Each player's observable upon receiving question $P\in\mathcal{S}$ is the observable $P$ itself (this is the motivation for imposing the condition that elements of $\mathcal{S}$ are Hermitian and square to identity). With this choice 
\begin{align}
 \beta^*(G) &\geq Z^{-1} \Big|\sum_{P,Q,R\in\mathcal{S}}  G_{P,Q,R} \, \bra{\psi}(P\otimes Q \otimes R) \ket{\psi}\Big|\notag\\
&= n^3 Z^{-1} \big|\bra{\psi} T\ket{\psi}\big|\notag\\
&= n^3 Z^{-1}\|T\|_{2,\epsilon},\label{eq:3xor-ub-5}
\end{align}
where the second line follows from the normalization assumption $\Tr(P^2)=\Tr(\Id)=n$ on elements of $\mathcal{S}$. Combining~\eqref{eq:3xor-ub-4} and~\eqref{eq:3xor-ub-5} gives~\eqref{eq:3xor-ub-3}, completing the second step of the proof. 

For games with a fixed number of questions per player there is a quadratic gap between the violation obtained in Corollary~\ref{cor:3xor-largeviolation} and the best upper bound known (Lemma \ref{lem:3xor-questionbound}). Thus the following question is of interest: 

\begin{question}
Find a tripartite Bell correlation inequality $G$ with $n$ inputs per party such that $\beta^*(G)/\beta(G) = \Omega (\sqrt{n}).$
\end{question}

It is also noteworthy that the construction described in this section is probabilistic: no deterministic construction is known, and in view of the possibility for experiments it would be highly desirable to obtain a deterministic construction that is as simple as possible. 

Another open question has to do with the problem of parallel repetition. As we saw for the case of two-player games a perfect direct product theorem holds, but no such result is known for $k>2$ players. 

\begin{question}[Parallel repetition of $\beta^*$]
Does the entangled bias of tripartite XOR games obey a (perfect) parallel repetition theorem? 
(Recall that XOR games must satisfy the normalization $\sum_{x_1,x_2,x_3}|\hat{G}_{x_1,x_2,x_3}|=1$, a condition which can be written as $\|\hat{G}\|_{\ell_1^{\cardx_1}\otimes_\pi \ell_1^{\cardx_2}\otimes_\pi\ell_1^{\cardx_3}}=1$, where $\pi$ is the projective norm. This question is thus equivalent to asking about the behavior of $\|\hat{G}^{\otimes \ell}\|_{\ell_1^{\ell \cdot \cardx_1}\otimes_{min}\ell_1^{\ell \cdot \cardx_2}\otimes_{min}\ell_1^{\ell \cdot \cardx_3}}$ for arbitrary real tensors $\hat{G}$.)
\end{question}
Little is known about this problem, which may be difficult. General results on the parallel repetition of two-player games suggest that perfect parallel repetition may not hold, while an exponential decay of the entangled bias with the number of repetitions could still be generic. Interestingly, it is known that the classical bias of XOR games with three or more players does not obey any direct product property, in the strongest possible sense: as discussed in Section~\ref{sec:xor-entanglement}, there is a three-player XOR game such that $\bias(G)<1$ but $\bias(G^{(\oplus\ell)})$ is bounded from below by a positive constant independent of $\ell$.

\subsection{XOR games with quantum messages}
\label{sec:quantum-xor}

In this section we pursue the exploration of possible extensions of the framework of XOR games that are motivated by interesting operator space structures. In the previous section we considered games with more than two players, encountering a situation where no constant-factor Grothendieck inequality holds. Here we consider a different extension, remaining in the two-player setting but allowing for quantum questions to be sent to the players. As we will see in the Banach/operator space picture this corresponds to replacing the commutative space $\ell_1^n$ associated with the players' questions in a classical XOR game with its non-commutative extension $S_1^n$, the space of $n\times n$ matrices endowed with the Schatten $1$-norm. In this setting a Grothendieck inequality \emph{does} hold and again has interesting consequences for the nonlocal properties of the underlying \emph{quantum XOR games}. 

For the sake of exposition we start by giving the operator space point of view, subsequently deriving the game from the underlying bilinear form or tensor. Consider the ``non-commutative'' generalization of the class of bilinear forms $G: \ell_\infty^\cardx \times \ell_\infty^\cardy\to \C$ associated with classical XOR games, to  bilinear forms $G: M_\cardx \times M_\cardy \to \C$. To $G$ we associate a tensor $\hat{G} \in  S_1^\cardx\otimes S_1^\cardy$ by $\hat{G}_{xy,x'y'} = G(E_{xx'},E_{yy'})$, where an operator space structure on $S_1^\cardx$ (resp. $S_1^\cardy)$ is naturally obtained by identifying it as the dual of $M_\cardx$ (resp. $M_\cardy$). 
Suppose for convenience that $\hat{G}$ is Hermitian, i.e. $\hat{G}_{xy,x'y'} = \overline{\hat{G}_{x'y',xy}}$. What kind of game, if any, does $G$ correspond to? To answer this question we express the bounded and completely bounded norms of $G$ as a bilinear form, and see if they can be given a natural interpretation as quantities associated to a quantum game. Consider thus a possible definition for the bias as
\begin{equation}\label{eq:qxor-qval}
\beta(G)\,=\, \|G\|_{\Bounded(M_\cardx,M_\cardy)} = \sup_{A\in \Ball(M_\cardx),\,B\in \Ball(M_\cardy)} \big| \Tr\big(\hat{G}\cdot(A\otimes B)\big)\big|.
\end{equation}
Up to the Hermiticity requirement the norm-$1$ operators $A$ and $B$ appearing in the supremum above can  be interpreted as quantum observables, a good sign that a quantum connection may not be far. Using the assumption that the tensor $\hat{G}$ is Hermitian, according to the Hilbert-Schmidt decomposition we may write $\hat{G} = \sum_i \mu_i \ket{\phi_i}\bra{\phi_i}$. Without loss of generality we may also normalize $\|\hat{G}\|_{S_1^{\cardx\cardy}}=1$, in which case $\sum_i |\mu_i|=1$ and we can write $\mu_i = p_i(-1)^{c_i}$ for a distribution $(p_1,\ldots,p_{\cardx\cardy})$ and $c_i\in\{0,1\}$. Thus
\begin{equation}\label{eq:qxor-qval-2}
\beta(G)\,=\, \sup_{A\in \Ball(M_\cardx),\,B\in \Ball(M_\cardy)} \Big| \sum_i p_i(-1)^{c_i} \bra{\phi_i} A\otimes B \ket{\phi_i}\Big|.
\end{equation}
Consider the following game. A (quantum) referee chooses an $i\in\{1,\ldots,\cardx\cardy\}$ with probability $p_i$, prepares the bipartite state $\ket{\phi_i}\in \C^\cardx\otimes\C^\cardy$, and sends the first (resp. second) register of $\ket{\phi_i}$ to the first (resp. second) player. The players measure their half of $\ket{\phi_i}$ using observables $A,B$ respectively, obtaining outcomes $a,b$ that they send back as their answers. The verifier accepts the answers if and only if the parity $a\oplus b = c_i$ holds. The maximum bias of arbitrary quantum players (not sharing any entanglement) in such an \emph{quantum XOR game} is precisely given by~\eqref{eq:qxor-qval-2}, where the supremum on the right-hand side should be restricted to Hermitian $A,B$. For the purposes of clarifying the connection with operator spaces we adopt~\eqref{eq:qxor-qval-2} as the definition of the \emph{unentangled bias} of a quantum XOR game,\footnote{The players in an XOR game with quantum messages are always quantum, and it would be less natural to talk of the ``classical'' bias.} but it is important to keep in mind that (as already demonstrated by the example of the CHSH game described after the statement of Lemma~\ref{lem:xor-real-complex}) restricting $A,B$ to Hermitian observables leads to a value that can be a constant factor $\sqrt{2}$ smaller than the one in~\eqref{eq:qxor-qval-2}.

Next we consider the completely bounded norm of $G$, and investigate whether it also corresponds to a natural quantity associated to the game just described. Let
\begin{equation}\label{eq:qxor-eval}
\beta^*(G)\,:=\, \|G\|_{\cBil(M_\cardx,M_\cardy)} \,=\, \sup_{\substack{d\in\N;\,A\in \Ball(M_\cardx\otimes\Bounded(\C^d)),\\B\in \Ball(M_\cardy\otimes\Bounded(\C^d))}} \big\| \Tr_{\C^\cardx\otimes\C^\cardy}(\hat{G}\otimes \Id_{\C^d} \otimes \Id_{\C^d} \cdot(A\otimes B))\big\|.
\end{equation}
Interpreting the two spaces $\C^d$ in this equation as additional spaces in which the players may hold an entangled state, the quantity $\beta^*(G)$ can be interpreted as the largest bias achievable in the game $G$ by quantum players allowed to share an arbitrary quantum state.\footnote{Here one can verify, as was already the case for the entangled bias of classical XOR games, that restricting the supremum in~\eqref{eq:qxor-eval} to Hermitian $A,B$ does not affect its value.} Thus just as for the case of classical XOR games there is a direct correspondence between the largest value achievable by players not using any entanglement and the bounded norm (Banach space level), and players using entanglement and the completely bounded norm (operator space level). 

Following on the tracks of our investigation of classical XOR games, the following questions are natural: what is the largest ratio that is achievable between the unentangled and entangled biases? Can either be computed efficiently? How much entanglement is needed to achieve optimality? In the following two subsections we briefly review what is known about these questions, highlighting how the answers relate to results in operator space theory. The interested reader is referred to~\cite{RegevV12a}, as well as~\cite{cooney} where a closely related class of games called \emph{rank-one quantum games} is introduced; these games are such that the \emph{square} $\|G\|_{\cBil(M_\cardx, M_{\cardy})}^2$ has a natural interpretation as the quantum (entangled) value of a so-called \emph{rank-one quantum game}.   

\subsubsection{The unentangled bias and the non-commutative Grothendieck inequality}

We first look at the unentangled bias $\beta(G)$ of a quantum XOR game, the maximum success probability of players not sharing any entanglement, as defined in~\eqref{eq:qxor-qval}. For the case where $G$ is diagonal, i.e. the associated bilinear form satisfies $G(E_{xx'},E_{yy'})=0$ whenever $x\neq x'$ or $y\neq y'$, the game reduces to a classical XOR game and $\beta(G)$ to the classical bias. As previously mentioned this quantity is NP-hard to compute, and even, assuming the Unique Games conjecture, to approximate within any factor smaller than the real Grothendieck constant. Moreover, in that case we saw that a best-possible efficiently computable approximation was given by the $\gamma_2^*$ norm, with Grothendieck's inequality providing the required bound on the approximation guarantee. It turns out that these observations extend to the case of general $G$, where now a best-possible (assuming P$\neq$NP) efficiently computable approximation can be obtained through a variant of Grothendieck's inequality that applies to arbitrary $C^*$-algebras. 

\begin{theorem}\cite{Pisier78NCGT,Haagerup85NCGT}\label{thm:ncgt}
Let $A,B$ be $C^*$-algebras and $M\in \Bil(A, B)$. Then
\begin{equation}\label{eq:ncgt}
\sup_{(x_i)\subset A,\,(y_i)\subset B} \Big|\sum_i M(x_i,y_i)\Big| \leq 2 \sup_{x\in\Ball(A),\,y\in\Ball(B)} \big|M(x,y)\big| = 2\|M\|_{\Bounded(A,B)},
\end{equation}
where the supremum on the left-hand side is taken over all sequences $(x_i)\subset A$ and $(y_i)\subset B$ such that $\|\sum_i x_ix_i^*\|+\|\sum_i x_i^*x_i\|\leq 2$ and $\|\sum_i y_iy_i^*\|+\|\sum_i y_i^*y_i\|\leq 2$.
\end{theorem}

For $M$ taken as the bilinear form associated to a quantum XOR game $G$, let $\beta^{nc}(G)$ denote the quantity on the left-hand side of~\eqref{eq:ncgt}. Since it is clear that $\|G\|_{\Bounded(M_\cardx,M_\cardy)}\leq \beta^{nc}(G)$ always holds, Theorem~\ref{thm:ncgt} states that $\bias^{nc}(G)$ is an approximation to $\bias(G)$ to within a multiplicative factor at most $2$. Based on a similar re-writing as previously done for the entangled bias of a classical XOR game it is possible to express $\bias^{nc}(G)$ as the optimum of a semidefinite program, and therefore approximated to within an additive $\pm \eps$ in time polynomial in the size of $G$ and $\log(1/\eps)$. As shown in~\cite{briet2014tight}, for general (not necessarily Hermitian) $M$ it is NP-hard to obtain approximations of $\beta(M)$ within any constant factor strictly less than $2$, thus here again the Grothendieck inequality expressed in Theorem~\ref{thm:ncgt} has (up to the Hermitianity restriction) a striking application as providing the best efficiently computable approximation to the value of a quantum game.  

Theorem~\ref{thm:ncgt} has  consequences for the problem of bounding the maximum bias achievable by quantum players allowed to share entangled states of a specific form --- in this case, a maximally entangled state of arbitrary dimension. Denote by $\beta^{me}(G)$ the associated bias:
\begin{align} 
\beta^{me}(G) &= \sup_{\substack{d\in\N;\,A\in\Ball(M_\cardx\otimes\Bounded(\C^d)),\\B\in \Ball(M_\cardy\otimes\Bounded(\C^d))}} \frac{1}{d}\Big|\sum_{i,j=1}^d \Tr \big(G\otimes \ket{j}\bra{i} \otimes \ket{j}\bra{i} \cdot(A\otimes B)\big)\Big|\label{eq:bias-qxor-me}\\
&=\sup_{\substack{d\in\N;\,A\in \Ball(M_\cardx\otimes\Bounded(\C^d)),\\B\in \Ball(M_\cardy\otimes\Bounded(\C^d))}} \Big|\sum_{i,j=1}^d \Tr \big(G \cdot (A_{i,j}\otimes B_{i,j})\big)\Big|,\notag
\end{align}
where $A_{i,j}=d^{-1/2} (\Id_{\C^\cardx} \otimes \bra{i})A(\Id_{\C^\cardx} \otimes \ket{j})$ and $B_{i,j}=d^{-1/2} (\Id_{\C^\cardy} \otimes \bra{i})B(\Id_{\C^\cardy} \otimes \ket{j})$. Recall that the expression~\eqref{eq:bias-qxor-me} was introduced as the tracially bounded norm of the bilinear form $G$ in Section~\ref{sec:bil-tensor}. From their definition 
the sequences $(A_{i,j})$ and $(B_{i,j})$ are easily seen to satisfy the constraints on $(x_i)$ and $(y_i)$ in~\eqref{eq:ncgt} and it follows that the inequality $\beta^{me}(G)\leq \beta^{nc}(G)$ always holds. Hence Theorem~\ref{thm:ncgt} implies that the largest advantage that can be gained from using maximally entangled states in a quantum XOR game is bounded by a constant factor $2\sqrt{2}$, where the additional factor $\sqrt{2}$ accounts for the restriction that quantum strategies are Hermitian. As will be seen in the next section no such bound holds when the players are allowed to share an arbitrary entangled state. 

\begin{question}
The best separation known between $\bias^{me}$ and $\bias$ is a factor $K_G^\R$ that follows from the consideration of classical XOR games, and the best separation known between $\bias^{nc}$ and $\bias$ is a factor $2$~\cite{RegevV12a}. The best upper bounds on either ratio are the factor-$2\sqrt{2}$ bounds mentioned here. What are the optimal separations?
\end{question}

\subsubsection{The entangled bias and the operator space Grothendieck inequality}

Next we turn to the entangled bias $\beta^*(G)$, defined in~\eqref{eq:qxor-eval}. As previously we observe that for the case of a diagonal $G$ this bias reduces to the entangled bias of a classical XOR game. While the latter was seen to be efficiently computable through its connection with the $\gamma_2^*$ norm, it is not known whether this holds for general $G$ (see Problem~\ref{prob:qxor-2} below). As in the previous section, here again an extension of Grothendieck's inequality, this time to operator spaces, will provide us with the best known polynomial-time approximation to $\bias^*$. 

\begin{theorem}\cite{PS02OSGT,HM08}\label{thm:osgt}
Let $A,B$ be $C^*$-algebras and $M\in \Bil(A, B)$. Then
\begin{align}
\sup_{(x_i)\subset A,\,y_i\subset B,\,(t_i)>0} \Big|\sum_i M(x_i,y_i) \Big| &\leq 2 \sup_{d\in\N;\,x\in\Ball(M_d(A)),\,y\in\Ball(M_d(B))} \big\|M\otimes \Id_d\otimes\Id_d(x,y)\big\|\notag\\
&= 2\,\|M\|_{\cBil(A,B)}
\label{eq:osgt},
\end{align}
where the supremum on the left-hand side is taken over all integers $d$ and sequences $(x_i)\subset M_d(A)$, $(y_i)\subset M_d(B)$ and positive reals $(t_i)$ such that $\|\sum_i t_i^2 x_ix_i^*\|+\|\sum_i x_i^*x_i\|\leq 2$ and $\|\sum_i y_iy_i^*\|+\|\sum_i t_i^{-2} y_i^*y_i\|\leq 2$.
\end{theorem}

In Section~\ref{sec:os-emb} we sketch a proof of Theorem~\ref{thm:osgt} that is inspired by quantum information theory. For $M$ taken as the bilinear form associated to a quantum XOR game $G$, let $\beta^{os}(G)$
denote the quantity appearing on the left-hand side of~\eqref{eq:osgt}. As for the case of $\bias^{nc}(G)$ it turns out that this quantity can be expressed as a semidefinite program and computed in polynomial time. 

In contrast to the case of classical XOR games, there exists a family of quantum XOR games $C_n\in\Bil(M_{n+1},M_{n+1})$ such that $\bias^*(C_n)/\bias(C_n)\to_{n\to\infty}\infty$. In terms of operator spaces, this corresponds to the observation that there exist bilinear forms on $M_{n+1}\times M_{n+1}$ that are bounded but not completely bounded. One can give more precise quantitative bounds: the  family $(C_n)$ can be taken such that $\bias^*(C_n)=1$ for all $n$ but $\bias(C_n)= 1/\sqrt{n}$. In addition, $\beta^*(C_n)=1$ can only be achieved in the limit of infinite-dimensional entanglement, a topic to which we return in Section~\ref{sec:Entanglement}. (Note also that necessarily $\beta^{me}(C_n)\leq 2\sqrt{2/n}$, and in fact $\beta^{me}(C_n)=\beta(C_n)$, thus in this case maximally entangled states are no more useful than no entanglement at all.)

\begin{question}\label{prob:qxor-2}
Theorem~\ref{thm:osgt} shows that $\beta^{os}$ is a factor $2$ approximation to $\bias^*$, but the best separation known between these two quantities is of a small constant factor larger than $1$~\cite{RegevV12a}. What is the optimal constant? 

Is there a tighter, efficiently computable approximation to $\bias^*$ than $\bias^{os}$? It is not known if $\bias^*$ itself is hard to compute.  

What is the maximum ratio $\beta^*(G)/\beta(G)$, as a function of the dimension of $G$ or the dimension of the shared entangled state? The best lower bound is the factor $\sqrt{n}$ mentioned above, obtained for a $G\in\Bil(M_{n+1},M_{n+1})$.
\end{question}


\section{Measuring nonlocality via tensor norms}\label{sec:2p1r}

The class of XOR games investigated in the previous section demonstrates a very tight connection between the different values, or biases, associated to these games and appropriately defined Banach and operator space norms. In Section~\ref{sec:positive-coefficients} we pursue this connection for the general setting of two-player games, further demonstrating that operator spaces provide a natural framework for the study of their nonlocal properties. In Section~\ref{sec:signed-coefficients} we extend the correspondence to arbitrary Bell functionals, with small losses in its ``tightness''. Finally in Section~\ref{sec:games-bounds} we describe some upper and lower bounds on the maximum ratio $\omega^*/\omega$, both for the case of games and of Bell functionals.  
\subsection{Two-player games}\label{sec:positive-coefficients}
Given a game $G=(\setx,\sety,\seta,\setb,\pi,V)$, a classical strategy for the players is described by an element $P_{\seta\setb|\setx\sety}\in \class(\seta\setb|\setx\sety)$, defined in~\eqref{eq:def-classical-prob} as the convex hull of the set of product strategies $\prob(\seta|\setx)\times\prob(\setb|\setx)$. The normalization condition $\sup_x \sum_a |P(a|x)|\leq 1$ suggests the introduction of the Banach space $\ell_\infty^\cardx(\ell_1^\carda)$, defined as $\C^{\cardx\carda}$ equipped with the norm
\begin{equation}\label{eq:linfl1norm}
\|(R(x,a))_{x,a} \|_{\infty,1}\,=\,\sup_x\sum_a|R(x,a)|.
\end{equation}
The game $G$ can be seen as a bilinear form $G:\ell_\infty^{\cardx}(\ell_1^{\carda})\times \ell_\infty^{\cardy}(\ell_1^{\cardb})\rightarrow \C$ defined by
 $$G(P, Q)=\sum_{x,y;a,b}G_{x,y}^{a,b}P(x,a)Q(y,b),$$
where $G_{x,y}^{a,b} = \pi(x,y)V(a,b,x,y)$, and with norm
\begin{align}\label{bilinear general games}
\|G\|=\sup\Big\{\Big|\sum_{x,y;a,b}G_{x,y}^{a,b}P(x,a)Q(y,b)\Big |: \|P\|_{\ell_\infty^{\cardx}(\ell_1^{\carda})}, \|Q\|_{\ell_\infty^{\cardy}(\ell_1^{\cardb})}\leq 1\Big\}.
\end{align}
The correspondence between bilinear forms and tensors invites us to consider the dual space $(\ell_\infty^\cardx(\ell_1^\carda))^*=\ell_1^\cardx(\ell_\infty^\carda)=(\C^{\cardx\carda}, \|\cdot \|_{1,\infty})$, where 
$$\big\|(R(x,a))_{x,a} \big\|_{1,\infty}\,=\,\sum_x\sup_a|R(x,a)|.$$
According to~\eqref{eq:bil-epsilon-norm} the tensor $\hat{G}=\sum_{x,y;a,b}G_{x,y}^{a,b}(e_x\otimes e_a) \otimes (e_y\otimes e_b)\in  \ell_1^{\cardx}(\ell_\infty^{\carda}) \otimes  \ell_1^{\cardy}(\ell_\infty^{\cardb})$ associated to $G$ verifies 
$$\|G\|=\|\hat{G}\|_{\ell_1^{\cardx}(\ell_\infty^{\carda})\otimes_\epsilon \ell_1^{\cardy}(\ell_\infty^{\cardb})}.$$
While~\eqref{eq:linfl1norm} and~\eqref{bilinear general games} make it clear that $\omega(G)\leq \|G\|$, since the space $\ell_\infty^\cardx(\ell_1^\carda)$ allows for elements with complex coefficients there could a priori be cases where the inequality is strict. As will be seen in the next section this can indeed happen for general Bell functionals $M$; however for the case of a game $G$ it always holds that 
\begin{equation}\label{eq:game-eps-norm}
\omega (G)=\|\hat{G}\|_{\ell_1^{\cardx}(\ell_\infty^{\cardy})\otimes_\epsilon \ell_1^{\cardy}(\ell_\infty^{\cardb})}.
\end{equation}
To see this, for any $P,Q$ such that $ \|P\|_{\ell_\infty^{\cardx}(\ell_1^{\carda})}$, $\|Q\|_{\ell_\infty^{\cardy}(\ell_1^{\cardb})}\leq 1$ write
$$\Big|\sum_{x,y;a,b}G_{x,y}^{a,b}P(x,a)Q(y,b)\Big|\leq \sum_{x,y;a,b}G_{x,y}^{a,b}|P(x,a)||Q(y,b)|\leq \omega(G),$$
where the last inequality follows from $ \||P|\|_{\ell_\infty^{\cardx}(\ell_1^{\carda})}$, $\||Q|\|_{\ell_\infty^{\cardy}(\ell_1^{\cardb})}\leq 1$ and the fact that $|P|$, $|Q|$ as well as $G$ have non-negative coefficients.

We proceed to analyze entangled strategies for the players, i.e.\ the set $\ent(\seta\setb|\setx\sety)$, and their relation to the completely bounded norm of $G:\ell_\infty^{\cardx}(\ell_1^{\carda})\times \ell_\infty^{\cardy}(\ell_1^{\cardb})\rightarrow \C$. Towards this end we need to define an o.s.s.\ on $\ell_\infty^{\cardx}(\ell_1^{\carda})$. Using the o.s.s.\ on $\ell_1^\carda$ introduced in~\eqref{o.s.s. ell_1}, together with the natural o.s.s.\ on $\ell_\infty^\cardx$, one can verify that the sequence of norms
\begin{align}\label{eq:os-l1}
\Big\|\sum_{x,a}T_x^a\otimes (e_x\otimes e_a)\Big\|_{M_d(\ell_\infty^{\cardx}(\ell_1^{\carda}))}=\sup_x\Big\|\sum_{a}T_x^a\otimes e_a\Big\|_{M_d(\ell_1^{\carda})}, \text{    }\text{    } d\geq 1,
\end{align}
defines a suitable o.s.s.\ on $\ell_\infty^{\cardx}(\ell_1^{\carda})$. Moreover, a corresponding o.s.s.\ can be placed on $\ell_1^\cardx(\ell_\infty^\carda)=(\ell_\infty^\cardx(\ell_1^\carda))^*$ using duality. With these structures in place we may express the completely bounded norm of $G$ as 
\begin{align}\label{bilinear general games quantum}
\|G\|_{\cBil(\ell_\infty^{\cardx}(\ell_1^{\carda}),\ell_\infty^{\cardy}(\ell_1^{\cardb}))}&=\sup_d\big\|G\otimes \Id_{M_d} \otimes \Id_{M_d} \big \|_{\Bil(M_d(\ell_\infty^{\cardx}(\ell_1^{\carda})), M_d(\ell_\infty^{\cardy}(\ell_1^{\cardb})))}\\&\nonumber=\sup_d\big\|\sum_{x,y;a,b}G_{x,y}^{a,b}\,T_x^a\otimes S_y^b\big\|_{M_{d^2}},
\end{align}
where the supremum is taken over all $d$ and $T_x^a,S_y^b\in M_d$ such that 
$$\max\Big\{\Big\|\sum_{x,a}T_x^a\otimes (e_x\otimes e_a)\Big\|_{M_d(\ell_\infty^{\cardx}(\ell_1^{\carda}))},\Big\|\sum_{y,b}S_y^b\otimes (e_y\otimes e_b)\Big\|_{M_d(\ell_\infty^{\cardy}(\ell_1^{\cardb}))}\Big\}\,\leq\, 1.$$
Using the completely isometric correspondence described in Section~\ref{sec:bil-tensor} this norm coincides with the minimal norm of the tensor $\hat{G}$ associated to $G$,   
$$\|G\|_{\cBil}=\|\hat{G}\|_{\ell_1^{\cardx}(\ell_\infty^{\carda})\otimes_{min}  \ell_1^{\cardy}(\ell_\infty^{\cardb})}.$$
As for the case of the classical value it turns out that the entangled value of a two-player game equals the completely bounded norm of the associated bilinear form, or equivalently the minimal norm of the corresponding tensor. Since this fact seems not to have previously appeared in the literature we state it as a lemma. 

\begin{lemma}\label{quantum value 2P1R}
Given a two-player game $G$,\footnote{We will see in the next section that the lemma is false for general Bell functionals.}
\begin{equation}\label{eq:game-min-norm}
\omega^* (G)=\|\hat{G}\|_{\ell_1^{\cardx}(\ell_\infty^{\carda})\otimes_{min} \ell_1^{\cardy}(\ell_\infty^{\cardb})}.
\end{equation}
\end{lemma}

\begin{proof}
Given a family of POVM $\{E_x^a\}_{a\in \seta}$ in $M_d$, for every $x\in\setx$ we have 
$$\Big\|\sum_{x,a}E_x^a\otimes (e_x\otimes e_a)\Big\|_{M_d(\ell_\infty^{\cardx}(\ell_1^{\carda}))}\,=\,1.$$
 According to~\eqref{eq:os-l1} this follows from the fact that for every $x$,
$\|\sum_{a}E_x^a\otimes e_a\|_{M_d(\ell_1^{\carda})}
=1$. Indeed, since the map $T_x:\ell_\infty^{\carda}\rightarrow M_d$ defined by $T_x(e_a)=E_x^a$ for every $a$ is completely positive and unital, by Lemma~\ref{lem:cp-unital} $\|T_x\|_{cb}=1$. Proceeding similarly with Bob's POVM, we deduce from~(\ref{bilinear general games quantum}) that 
$$\omega^*(G)\leq \|\hat{G}\|_{\ell_1^{\cardx}(\ell_\infty^{\carda})\otimes_{min}  \ell_1^{\cardy}(\ell_\infty^{\cardb})}.$$
It remains to show the converse inequality. According to~\eqref{bilinear general games quantum}, given $\eps>0$ there exists an integer $d$, $T_x^a,S_y^b\in M_d$ satisfying $\|\sum_{x,a}T_x^a\otimes (e_x\otimes e_a)\|_{M_d(\ell_\infty^{\cardx}(\ell_1^{\carda}))}\leq 1$, $\|\sum_{y,b}S_y^b\otimes (e_y\otimes e_b)\|_{M_d(\ell_\infty^{\cardy}(\ell_1^{\cardb}))}\leq 1$, and unit vectors $|u\rangle$, $|v\rangle$ in $\C^{d^2}$ such that 
$$\bra{u}\sum_{x,y;a,b}G_{x,y}^{a,b}T_x^a\otimes S_y^b\ket{v}\,>\, \|\hat{G}\|_{\ell_1^{\cardx}(\ell_\infty^{\carda})\otimes_{min} \ell_1^{\cardy}(\ell_\infty^{\cardb})}-\eps.$$
By definition of the o.s.s.\ on $\ell_1$ via duality (\ref{dual o.s.s.}), the condition 
$$\Big\|\sum_{x,a}T_x^a\otimes (e_x\otimes e_a)\Big\|_{M_d(\ell_\infty^{\cardx}(\ell_1^{\carda}))}=\sup_x\Big\|\sum_a T_x^a\otimes e_a\Big\|_{M_d(\ell_1^{\carda})}\leq 1$$ is equivalent to  
\begin{align}\label{cond dual Lemma}
\|T_x:\ell_\infty^{\carda}\rightarrow M_d\|_{cb}\leq 1 \text{     }\text{     } \text{for every} \text{     } x,
\end{align}where $T_x(e_a)=T_x^a$ for every $a$. The same bound applies to the operators $S_y^b$.

The main obstacle to conclude the proof is that the  $T_x^a$, $S_y^b$ are not necessarily positive, or even Hermitian. In order to recover a proper quantum strategy we appeal to the following. 
\begin{theorem}\label{positivity theorem}
Let $ A$ be a C$^*$-algebra with unit and  let $T: A\rightarrow \Bounded(\h)$ be completely bounded. Then there exist completely positive maps
 $\psi_i: A\rightarrow \Bounded(\h)$, with $\|\psi_i\|_{cb}=\|T\|_{cb}$ for $i=1,2$, such that the map $\Psi:A\rightarrow M_2(\Bounded(\h))$ given by
 $$\Psi(a)=\left( \begin{array}{cc}
\psi_1(a) & T(a)\\
 T^*(a) & \psi_2(a)
\end{array}
\right), \text{    } \text{    }  a\in A$$
is completely positive. Moreover, if $\|T\|_{cb}\leq 1$, then we may take $\psi_1$ and $\psi_2$ unital.
\end{theorem}

Theorem \ref{positivity theorem} is a direct consequence of \cite[Theorem 8.3]{PaulsenBook}, where the same statement is proved with the map $\Psi$ replaced by the map $\eta:M_2( A)\rightarrow M_2(\Bounded(\h))$ given by
 $$\eta\left( \begin{pmatrix}
a & b\\
 c & d
 \end{pmatrix}
\right)=\begin{pmatrix}
\psi_1(a) & T(b)\\
 T^*(c) & \psi_2(d)
\end{pmatrix}.$$
The complete positivity of $\eta$ implies that the map $\Psi$ defined in Theorem~\ref{positivity theorem} is completely positive. In fact, it is an equivalence~\cite[Exercise 8.9]{PaulsenBook}. (While in the moreover case $\eta$ is unital, so $\|\eta\|_{cb}=1$, in general one can only obtain $\|\Psi\|_{cb}\leq 2$.)

In our setting we take $ A=\ell_\infty^\carda$. Since this is a commutative C$^*$-algebra, a map $T:\ell_\infty^\carda\rightarrow B(\h)$ is completely positive if and only if it is positive; that is, $T(a)\in B(\h)$ is a positive element for every positive element $a\in \ell_\infty^\carda$. For every $x\in\setx$, applying  Theorem \ref{positivity theorem} to the map $T_x:\ell_\infty^{\carda}\rightarrow M_d$ defined in (\ref{cond dual Lemma}), we find completely positive and unital maps $\psi_x^i: \ell_\infty^{\carda}\rightarrow M_d$, $i=1,2$ such that the map $\Psi_x:\ell_\infty^{\carda}\rightarrow M_2(M_d)$ defined by
$$\Psi_x(a)=\left( \begin{array}{cc}
\psi_x^1(a) & T_x(a)\\
 T_x^*(a) & \psi_x^2(a) 	
\end{array}
\right), \text{    } \text{    }  a\in \ell_\infty^{\carda}$$
 is completely positive. Similarly, for every $y\in\sety$ we define $S_y:\ell_\infty^{\cardb}\rightarrow M_d$ and find completely positive and unital maps $\varphi_y^i:\ell_\infty^{\cardb}\rightarrow M_d$, $i=1,2$ and $\Phi_y:\ell_\infty^{\cardb}\rightarrow M_2(M_d)$. Since these maps are positive, the element 
$$\Gamma=\sum_{x,y;a,b}G_{x,y}^{a,b}\Psi_x(e_a)\otimes \Phi_y(e_b)\in M_2(M_d)\otimes M_2(M_d)$$
 is positive. Consider the unit vectors $\tilde{u}=(u, 0, 0, 0)\in \C^{4d^2}$ and $\tilde{v}=(0, 0, 0, v)\in \C^{4d^2}$; we have
\begin{align*}
&\Big|\langle u|\sum_{x,y;a,b}G_{x,y}^{a,b}T_x^a\otimes S_y^b|v\rangle \Big|=|\langle \tilde{u}|\Gamma |\tilde{v}\rangle|\leq |\langle \tilde{u}|\Gamma | \tilde{v}\rangle|^{\frac{1}{2}}|\langle \tilde{u}|\Gamma|\tilde{v}\rangle|^{\frac{1}{2}}\\
&=\Big|\bra{u}\sum_{x,y;a,b}G_{x,y}^{a,b}\psi_x^1(e_a) \otimes \varphi_y^1(e_b) \ket{u}\Big |^{\frac{1}{2}}\Big|\bra{v}\sum_{x,y;a,b}G_{x,y}^{a,b}\psi_x^2(e_a) \otimes \varphi_y^2(e_b) \ket{v}\Big|^{\frac{1}{2}}\leq \omega_d^*(G),
\end{align*}
where the first inequality follows from the Cauchy-Schwarz inequality and the second inequality follows from the fact that the corresponding maps are completely positive and unital.
\end{proof}

\subsection{Bell functionals with signed coefficients}
\label{sec:signed-coefficients}

It is sometimes interesting to consider arbitrary Bell functionals $M=\{M_{x,y}^{a,b}\}_{x,y;a,b}$, that may not directly correspond to games because of the presence of signed coefficients. This additional freedom leads to phenomena with no equivalent in games, and we will see some examples later in this survey (see Section \ref{sec:general-me}). 

We note that the possibility for considering signed coefficients is the reason behind the introduction of an absolute value in the definitions~\eqref{eq:class-value-def} and~\eqref{eq:def-entangled-value} of $\omega$ and $\omega^*$ respectively. Indeed, while this absolute value  is superfluous in the case of games in general it is needed for the quantity $\omega^*(M)/\omega(M)$ to be meaningful: without it this quantity could be made to take any value in $[-\infty, \infty]$ simply by shifting the coefficients $M_{x,y}^{a,b} \rightarrow M_{x,y}^{a,b}+c$. The presence of the absolute value in the definition allows one to show that the ratio $\omega^*(M)/\omega(M)$ can always be obtained as the ratio between the quantum and classical \emph{biases}, defined as the maximum possible deviation respectively of the quantum and classical value from $1/2$, of an associated game~\cite[Section 2]{BuhrmanRSW12}.

Unfortunately the use of signed coefficients makes the geometry of the problem more complicated. In particular the correspondence between the classical and entangled values of a game and the injective and minimal norms of the associated tensor described in the previous section no longer holds. As a simple example, consider $\{M_{x,y}^{a,b}\}_{x,y;a,b=1}^2$ such that $M_{1,1}^{a,b}=1=-M_{2,2}^{a,b}$ for every $a,b\in\{1,2\}$ and $M_{x,y}^{a,b}=0$ otherwise. It is possible to think of $M$ as a game in which the players' ``payoff'' may be negative on certain answers; here the payoff would be $1$ on questions $(1,1)$, $-1$ on questions $(2,2)$ and $0$ otherwise, irrespective of the players' answers. Then clearly $\omega(M)=\omega^*(M)=0$, but since $M\neq 0$ it must be that
$$ \Big\|\sum_{x,y,a,b} M_{x,y}^{a,b}(e_x\otimes e_a) \otimes (e_y\otimes e_b)\Big\|_{\ell_1^{\cardx}(\ell_\infty^{\carda})\otimes_{\epsilon} \ell_1^{\cardy}(\ell_\infty^{\cardb})}\neq 0.$$
By considering a small perturbation of $M$ it is possible to construct functionals $\tilde{M}$ such that $\omega(\tilde{M})\neq 0$ but the quotient $\|\tilde{M}\|_{\ell_1^{\cardx}(\ell_\infty^{\carda})\otimes_{\epsilon} \ell_1^{\cardy}(\ell_\infty^{\cardb})}/\omega(\tilde{M})$ is arbitrarily large; the same effect can be obtained for $\omega^*(M)$ with the minimal norm.

There are two different ways to circumvent this problem. The first, considered in \cite{JungePPVW10CMP, JungePPVW10}, consists in a slight modification of the definition of the classical and entangled values $\omega(M)$ and $\omega^*(M)$. Intuitively, since payoffs may be negative it is natural to allow the players to avoid ``losing'' by giving them the possibility to refuse to provide an answer, or alternatively, to provide a ``dummy'' answer on which the payoff is always null. This leads to a notion of ``incomplete'' strategies which, aside from being mathematically convenient, is also natural to consider in the setting of Bell inequalities, for instance as a way to measure detector inefficiencies in experiments.\footnote{We refer to~\cite[Section 5]{JungePPVW10CMP} for more on this. Incomplete distributions have also been used in~\cite{lancien2015parallel} to obtain a parallel repetition theorem for the no-signaling value of multiplayer games.}

Formally, a family of incomplete conditional distributions is specified by a  vector $P_A=(P_A(a|x))_{x,a}\in \mathbb R^{\cardx\carda}$ of non-negative reals such that $\sum_aP_A(a|x)\leq 1$ for every $x\in\setx$. Define $\omega_{inc}(M)$ as in~\eqref{eq:class-value-def} where the supremum is extended to the convex hull of products of incomplete conditional distributions for each player. An analogous extension can be considered for the entangled bias, leading to a value $\omega^*_{inc}(M)$ obtained by taking a supremum over all distributions that can be obtained from a bipartite quantum state $|\psi\rangle\in \Ball(\C^{d}\otimes \C^{d})$ and families of positive operators $\{E_x^a\}_{a\in \seta}$ and $\{F_y^b\}_{b\in \setb}$ in $M_{d}$ verifying $\sum_aE_x^a\leq \Id$ and $\sum_bF_y^b\leq \Id$. Note that in general it will always be the case that $\omega(M)\leq\omega_{inc}(M)$ and $\omega^*(M)\leq\omega_{inc}^*(M)$, and the example given at the beginning of the section can be used to show that the new quantities can be arbitrarily larger than the previous ones. The following lemma shows that considering incomplete distributions allows us to restore a connection with operator spaces, albeit only up to constant multiplicative factors. 
\begin{lemma}\label{Lemma incomplete values}
Let $M$ be a Bell functional and $\hat{M}=\sum_{x,y;a,b}M_{x,y}^{a,b}(e_x\otimes e_a) \otimes (e_y\otimes e_b)\in \R^{\cardx\carda}\otimes \R^{\cardy\cardb}$ the associated tensor. Then
\begin{align*}
&\omega_{inc}(M)\leq \|\hat{M}\|_{\ell_1^{\cardx}(\ell_\infty^{\carda}(\R))\otimes_{\epsilon} \ell_1^{\cardy}(\ell_\infty^{\cardb}(\R))}\leq 4\,\omega_{inc}(M), \text{    }\text{    }\text{  and  }\\
&\omega^*_{inc}(M)\leq \|\hat{M}\|_{\ell_1^{\cardx}(\ell_\infty^{\carda})\otimes_{min} \ell_1^{\cardy}(\ell_\infty^{\cardb})}\leq 4\,\omega^*_{inc}(M).
\end{align*}
Here, in the first inequality the norm of $\hat{M}$ is taken over real spaces, while in the second the spaces are complex.
\end{lemma}
The first estimate in Lemma~\ref{Lemma incomplete values} is not hard to obtain (see \cite[Proposition 4]{JungePPVW10CMP}), and readily extends to complex spaces with a factor of $16$ instead of $4$. The second estimate is proved in \cite[Theorem 6]{JungePPVW10CMP} with a constant 16. The constant $4$ stated in the theorem can be obtained by using the map $\Psi$ introduced in Theorem \ref{positivity theorem} and the fact that $\|\Psi\|_{cb}\leq 2$. 

Lemma~\ref{Lemma incomplete values} provides us with a method to translate constructions in operator space theory to Bell functionals for which the ratio of the entangled and classical values is large. Indeed, the correspondence established in the lemma implies that if $\hat{M}$ is a tensor such that $\|\hat{M}\|_{min}/\|\hat{M}\|_\eps$ is large the associated functional $M$ will be such that $\omega^*_{inc}(M)/\omega_{inc}(M)$ is correspondingly large (up to the loss of a factor $4$). Increasing the number of possible outputs by $1$, we may then add a ``dummy'' output for which the payoff is always zero. This results in a functional $\tilde{M}$ for which $\omega(\tilde{M})=\omega_{inc}(M)$ and $\omega^*(\tilde{M})=\omega_{inc}^*(M)$, so that any large violation for the incomplete values of $M$ translates to a large violation for the values of $\tilde{M}$.

Unfortunately the lemma is not sufficient to obtain bounds in the other direction, upper bounds on the ratio $\omega^*(M)/\omega(M)$: as shown by the example described earlier, the values $\|M\|_{\epsilon}$ and $\omega(M)$, and $\|M\|_{min}$ and $\omega^*(M)$, are in general incomparable. 
In order to obtain such upper bounds a more sophisticated approach was introduced in~\cite{JungeP11low}. Consider the space
\begin{align*}
\mathcal{N}_{\K}(\seta|\setx)=\big\{\{R(a|x)\}_{x,a=1}^{\cardx,\carda}\in \K^{\cardx\carda}: \sum_{a=1}^\carda R(a|x)=\sum_{a=1}^\carda R(a|x')\,\forall x,x'\in\setx\big\},
\end{align*}
where $\K=\R$ or $\K=\C$. This space is introduced to play the role of $\ell_\infty^\cardx(\ell_1^\carda)$ above, while allowing a finer description of classical strategies. In particular, note that $\dim_{\R}(\mathcal{N}_{\R}(\seta|\setx))=\cardx\carda-\cardx+1$, while $\dim_\R(\ell_\infty^\cardx(\ell_1^\carda))=\cardx\carda$. The space $\mathcal{N}_{\R}(\seta|\setx)$ (resp. $\mathcal{N}_{\C}(\seta|\setx)$) can be endowed with a norm (resp. o.s.s.) such that the following holds. 
\begin{lemma}[\cite{JungeP11low}]\label{NSG classical-quantum}
Let $M$ be a Bell functional and $\hat{M}=\sum_{x,y;a,b}M_{x,y}^{a,b}(e_x\otimes e_a) \otimes (e_y\otimes e_b)$ the associated tensor, viewed as an element of $\mathcal{N}_\K(\seta|\setx)^* \otimes \mathcal{N}_\K(\setb|\sety)^*$. Then 
\begin{align*}
\omega(M)= \|\hat{M}\|_{\mathcal{N}_\R(\seta|\setx)^*\otimes_{\epsilon} \mathcal{N}_\R(\setb|\sety)^*} \qquad\text{and}\qquad\omega^*(M)= \|\hat{M}\|_{\mathcal{N}_\C(\seta|\setx)^*\otimes_{\min} \mathcal{N}_\C(\setb|\sety)^*}.
\end{align*}
\end{lemma}
Lemma~\ref{NSG classical-quantum} shows that it is possible to describe a Banach space and an o.s.s.\ on it that precisely capture the classical and entangled values of arbitrary Bell functionals. Unfortunately there is a price to pay, which is that the relatively well-behaved space $\ell_1^{\cardx}(\ell_\infty^{\carda})$ is replaced by a more complex object, $\mathcal{N}_\K(\seta|\setx)^*$. We refer to~\cite[Section 5]{JungeP11low} for more details on the structure of these spaces and their use in placing upper bounds on the ratio $\omega^*(M)/\omega(M)$. 
\subsection{Bounds on the largest violations achievable in two-player games}\label{sec:games-bounds}

In the previous two sections we related the classical and entangled values of two-player games and Bell functionals to the bounded and completely bounded norms of the associated bilinear form respectively. This correspondence allows the application of tools developed for the study of these norms in operator space theory to quantify the ratio $\omega^*/\omega$, a quantity that can be interpreted as a measure of the nonlocality of quantum mechanics. For the case of classical two-player XOR games it was shown in Section~\ref{sec:xor-games} that this ratio is always bounded by a constant independent of the size of the game. In the same section we considered extensions of XOR games where there are more than two players, or two players but the messages can be quantum; in either case the ratio could become arbitrarily large as the size of the game was allowed to increase. The situation for general, non-XOR two-player games is similar. In Section~\ref{sec:violations-ub} we discuss upper bounds on the ratio $\omega^*/\omega$ that depend on the number of questions or answers in the game. In Section~\ref{sec:violations-lb} we describe examples of games that come close to saturating these bounds, leading to violations that scale (almost) linearly with the number of answers and as the square root of the number of questions. 

\subsubsection{Upper bounds}\label{sec:violations-ub}
The proof of the following proposition is new. It follows the same ideas as in~\cite{JungeP11low}, where similar estimates were given for general Bell functionals but with worse constant factors. The proof is a good example of the application of estimates from the theory of operator spaces to bounds on the entangled and classical values of a multiplayer game. 
\begin{proposition}\label{Prop upper bounds}
The following inequalities hold for any two-player game $G$:
\begin{enumerate}
\item [1.] $\omega^*(G)\leq \min\{\cardx,\cardy\}\,\omega(G)$,
\item [2.] $\omega^*(G)\leq K_G^\C\sqrt{\carda\cardb}\,\omega(G)$, where $K_G^\C$ is the complex Grothendieck constant.
\end{enumerate}
\end{proposition}
\begin{proof}
The proof of each item is based on a different way of bounding the norm of the identity map 
\begin{equation}\label{eq:upperbounds-0}
id\otimes id:\, \ell_1^{\cardx}(\ell_\infty^{\carda})\otimes_{\epsilon} \ell_1^{\cardy}(\ell_\infty^{\cardb})\rightarrow \ell_1^{\cardx}(\ell_\infty^{\carda})\otimes_{\min} \ell_1^{\cardy}(\ell_\infty^{\cardb}).
\end{equation}
Using~\eqref{eq:game-eps-norm} and~\eqref{eq:game-min-norm} any such bound immediately implies the same bound on the ratio $\omega^*/\omega$. 

For the first item, assume without loss of generality that $\cardx\leq \cardy$. The identity map~\eqref{eq:upperbounds-0} can be decomposed as a sequence 
\begin{equation}\label{eq:2p1rd-dim-1}
\ell_1^{\cardx}(\ell_\infty^{\carda})\otimes_{\epsilon} \ell_1^{\cardy}(\ell_\infty^{\cardb})\rightarrow \ell_\infty^{\cardx\carda}\otimes_{\epsilon} \ell_1^{\cardy}(\ell_\infty^{\cardb})\rightarrow \ell_\infty^{\cardx\carda}\otimes_{min} \ell_1^{\cardy}(\ell_\infty^{\cardb})\rightarrow \ell_1^{\cardx}(\ell_\infty^{\carda})\otimes_{min} \ell_1^{\cardy}(\ell_\infty^{\cardb}),
\end{equation}
where all arrows correspond to the identity. It follows directly from the definition of the $\epsilon$ and the $min$ norms that the first and the third arrow in~\eqref{eq:2p1rd-dim-1} have norm $1$ and $\cardx$ respectively. Using that $ \ell_\infty^{\cardx\carda}$ is a commutative C$^*$-algebra, the second arrow has norm $1$~\cite[Proposition 1.10 (ii)]{PisierBook}, so that the desired result is proved by composing the three norm estimates. Motivated by this decomposition, we give a self-contained proof relating the quantum and classical values. Start with the third arrow in~\eqref{eq:2p1rd-dim-1}. Given a family of POVM $\{E_x^a\}_{a}$, $x\in\setx$, for Alice, $\{\frac{1}{\cardx}E_x^a\}_{x,a}$ can be interpreted as a single POVM with $\cardx\carda$ outcomes. Thus,
\begin{align*}
\omega^*(G)\leq \cardx\sup \Big|\sum_{x,y;a,b}G_{x,y}^{a,b}\,\langle\psi|E^{x,a} \otimes F_y^a|\psi\rangle\Big|,
\end{align*}
where the supremum is taken over all families of POVMs $\{F_y^b\}_b$ for Bob, a single POVM $\{E^{x,a}\}_{x,a}$ for Alice with  $\cardx\carda$ outputs, and all bipartite states $|\psi\rangle$. Now that Alice is performing a single measurement, she may as well apply it before the game starts and her strategy can be assumed to be classical probabilistic. Thus
\begin{align}\label{eq:upperbounds-1}
\omega^*(G)\leq \cardx\sup \Big|\sum_{x,y;a,b}G_{x,y}^{a,b}\,P(x,a) Q(b|y)\Big|,
\end{align}
where the supremum is taken over all $P\in\class(\seta\setx)$ and $Q\in \class(\setb|\sety)$. This corresponds to the second arrow in~\eqref{eq:2p1rd-dim-1}. Finally, the fact that the first map in~\eqref{eq:2p1rd-dim-1} has norm $1$ corresponds, in this setting, to the observation that the distribution $P$ can be transformed into an element $\tilde{P}\in\class(\seta|\setx)$ such that $\tilde{P}(a|x)\geq P(x,a)$ for every $x,a$. Since $G$ has positive coefficients this can only increase the value. Thus the supremum on the right-hand side of~\eqref{eq:upperbounds-1} is at most $\omega(G)$, completing the proof of the first item in the proposition. 

The proof of the second item makes use of the Fourier transform 
$$\mathcal{F}_N: \C^N\rightarrow \C^N,\qquad \mathcal{F}_N:e_j\mapsto\sum_{k=1}^N e^{\frac{2\pi i j k}{N}}e_k\qquad\forall j\in\{1,\ldots,N\}.$$
The identity map~\eqref{eq:upperbounds-0} can be decomposed as 
\begin{equation}\label{eq:upperbounds-3}
\ell_1^{\cardx}(\ell_\infty^{\carda})\otimes_{\epsilon} \ell_1^{\cardy}(\ell_\infty^{\cardb})\rightarrow \ell_1^{\cardx\carda}\otimes_{\epsilon} \ell_1^{\cardy\cardb}\rightarrow \ell_1^{\cardx\carda}\otimes_{min} \ell_1^{\cardy\cardb}\rightarrow \ell_1^{\cardx}(\ell_\infty^{\carda})\otimes_{min} \ell_1^{\cardy}(\ell_\infty^{\cardb}),
\end{equation}
where the first arrow is $\frac{1}{\carda\cardb}(id_{\cardx}\otimes \mathcal F_{\carda})\otimes (id_{\cardy}\otimes \mathcal F_{\cardb})$, the second is the identity, and the third arrow is $\carda \cardb (id_{\cardx}\otimes \mathcal F^{-1}_{\carda})\otimes (id_{\cardy}\otimes \mathcal F^{-1}_{\cardb})$. Here it is not hard to verify that the norm of the first and the third maps are $\sqrt{\carda\cardb}$ and $1$ respectively. Grothendieck's theorem can be used to show that the second map has norm at most $K_G^\C$. Composing the three estimates proves the second item. As for the first item, we give a self-contained proof directly relating the classical and entangled values. Start with the third arrow in~\eqref{eq:upperbounds-3}. Given a family of POVM $\{E_x^a\}_{a}$, $x\in\setx$, for Alice, the (not necessarily self-adjoint) operators 
$$A_{x,k}=\sum_{a\in \seta}e^{-\frac{2\pi i a k}{\carda}} E_x^a$$
verify that $\|A_{x,k}\|\leq 1$ for every $x\in \setx$, $k\in \seta$. To see this, for any unit $\ket{u},\ket{v}$ write
\begin{align*}
\big|\bra{u}A_{x,k}\ket{v}\big| &= \Big|\sum_a\, e^{-\frac{2\pi i a k}{\carda}} \bra{u}E_x^a\ket{v}\Big|\\
&\leq \sum_a \,\big|\bra{u}E_x^a\ket{v}\big|\\
&\leq \Big( \sum_a \,\bra{u}E_x^a\ket{u} \Big)^{1/2} \Big(\sum_a \,\bra{v}E_x^a\ket{v} \Big)^{1/2}\\
&= 1,
\end{align*}
where the second inequality uses $E_x^a \geq 0$ and the Cauchy-Schwarz inequality, and the last follows from $\sum_a E_x^a = \Id$. 
The same transformation can be applied to obtain operators $B_{y,k'}$ from Bob's POVM, thus
$$\omega^*(G)\,\leq\,\frac{1}{\carda \cardb}\sup \Big|\sum_{x,y;k,k'}\Big(\sum_{a,b}G_{x,y}^{a,b}e^{\frac{2\pi i a k}{\carda}}e^{\frac{2\pi i b k'}{\cardb}}\Big)\,\langle\psi|A_{x,k} \otimes B_{y,k'}|\psi\rangle\Big|,$$
where the supremum on the right-hand side is taken over all $d$, states $\ket{\psi}\in\Ball(\C^d\otimes \C^d)$ and $A_{x,k},B_{y,k'}\in M_d$ of norm at most $1$. For the next step, interpret the coefficients 
$$\Big(\sum_{a,b}G_{x,y}^{a,b}\,e^{\frac{2\pi i a k}{\carda}}e^{\frac{2\pi i b k'}{\cardb}}\Big)_{x,k;y,k'}$$
 as a complex $\cardx\carda\times \cardy\cardb$ matrix and apply Grothendieck's inequality (Theorem~\ref{thm:grothendieck}) to obtain
\begin{equation}\label{eq:upperbounds-5}
\omega^*(G)\,\leq\, \frac{1}{\carda \cardb}K_G^\C\,\sup\Big|\sum_{x,y;k,k'}\Big(\sum_{a,b}G_{x,y}^{a,b}\,e^{\frac{2\pi i a k}{\carda}}e^{\frac{2\pi i b k'}{\cardb}}\Big) t_{x,k}s_{y,k'}\Big|,
\end{equation}
where the supremum is taken over all $(t_{x,k})\in\ell_\infty^{\cardx\carda}$ and $(s_{y,k'})\in\ell_\infty^{\cardy\cardb}$ of norm at most $1$. This gives the second arrow in~\eqref{eq:upperbounds-3}. To obtain the first, observe that the expression appearing inside the supremum on the right-hand side of~\eqref{eq:upperbounds-5} may be rewritten as 
$$\Big|\sum_{x,y;a,b}G_{x,y}^{a,b}\Big(\sum_{k}e^{\frac{2\pi i a k}{\carda}}t_{x,k}\Big)\Big(\sum_{k'}e^{\frac{2\pi i b k'}{\cardb}}s_{y,k'}\Big)\Big|.$$
For any family of complex numbers $(t_{x,k})_{x,k}$ such that $\sup_{x\in \setx,k\in \carda}|t_{x,k}|\leq 1$ it follows from Parseval's identity and the Cauchy-Schwarz inequality that the complex numbers $P(x,a)=\sum_k\exp(2\pi i a k/\carda)t_{x,k}$, $x\in \setx,a\in \seta$ satisfy $\sum_{a}|P(x,a)|\leq \carda^{3/2}$ for every $x\in\setx$. Using that the coefficients $G_{x,y}^{a,b}$ are positive, the supremum in~\eqref{eq:upperbounds-5} is at most $(\carda\cardb)^\frac{3}{2} \omega(G)$, concluding the proof of the second item in the proposition. 
\end{proof}

The bounds stated in Proposition~\ref{Prop upper bounds} can be extended in several ways. First, the same estimates as stated in the proposition apply to the quantities $\|\hat{M}\|_\eps$ and $\|\hat{M}\|_{min}$, for an arbitrary tensor $\hat{M}\in  \ell_1^{\cardx}(\ell_\infty^{\carda})\otimes \ell_1^{\cardy}(\ell_\infty^{\cardb})$, instead of $\omega(G)$ and $\omega^*(G)$ respectively. The proof of Proposition~\ref{Prop upper bounds} goes through the complex spaces $\ell_1^{\cardx}(\ell_\infty^{\carda})\otimes \ell_1^{\cardy}(\ell_\infty^{\cardb})$, thus using Lemma~\ref{Lemma incomplete values} and the comments that follow it similar bounds can be derived that relate the quantities $\omega_{inc}(M)$ and $\omega^*_{inc}(M)$ for arbitrary Bell functionals $M$. Bounds (slightly weakened by a constant factor) can then be obtained for the values $\omega(M)$ and $\omega^*(M)$ by using Lemma~\ref{NSG classical-quantum} and the fact that the space $\mathcal{N}_\K(\seta|\setx)^*$ is ``very similar'' to the space $\ell_1^\cardx(\ell_\infty^{\carda-1})$; we refer to~\cite[Section 5]{JungeP11low} for details on this last point.

Second, we note that the proof of the first item in the proposition can be slightly modified to show that for any game $G$ it holds that $\omega^*(G)\leq \min\{\carda,\cardb\}\omega(G)$.
The key point to show this is that 
$$\big\|id\otimes id: \ell_1^{\cardx\carda}\otimes_{\epsilon} \ell_1^{\cardy}(\ell_\infty^{\cardb})\rightarrow \ell_1^{\cardx\carda}\otimes_{\min} \ell_1^{\cardy}(\ell_\infty^{\cardb})\big\|_+=1,$$
where $\|\cdot \|_+$ denotes the norm of a map when it is restricted to positive elements. Note that this bound is always stronger than the one provided in the second item in the proposition. However, as discussed above the latter also applies to general Bell functionals, while the above bound is no longer true. Indeed, in~\cite{PalazuelosY15} the authors show the existence of a family of Bell functionals $(M_n)$ with $\carda = 2$, $\cardb = n$, $\cardx = \cardy = 2^n$ such that $\omega^*(M_n)/\omega(M_n) \geq C\sqrt{n}/\log^2n$, where $C$ is a universal constant.

Third, the proof of the second item in Proposition \ref{Prop upper bounds} can be adapted to derive the same bound when the quantum value $\omega^*$ is replaced by the $\gamma_2^*$ norm (q.v.~\eqref{eq:xor-facto-norm} for the definition in the special case of $\ell_1^\cardx\otimes\ell_1^\cardy$). Since the $\gamma_2^*$ norm is known to be at least as large as the value returned by the ``basic semidefinite relaxation''~\cite{raghavendra2008optimal} for $\omega^*(G)$, the bound from Proposition~\ref{Prop upper bounds} applies to that value as well.
The $\gamma_2^*$ norm was used in \cite{Dukaric11} to obtain a similar upper bound as the one in the second item in Proposition \ref{Prop upper bounds}. The author obtained a constant factor $K=\pi/(2\ln (1+\sqrt{2}))$, which is slightly worse than ours (q.v. comments following Theorem \ref{thm:grothendieck}). Finally, we note that the use of the spaces $\mathcal{N}_\K(\carda|\cardx)$ allows to extend the result on $\gamma_2^*$ to the case of general Bell functionals, at the price of a worse universal constant. 
\subsubsection{Lower bounds}\label{sec:violations-lb}
In \cite{BuhrmanRSW12} the authors analyze a family of two-player games which shows that the bound provided by the second item in Proposition \ref{Prop upper bounds} is essentially optimal. The games were originally introduced by Khot and Vishnoi~\cite{KhotV15} to obtain the first integrality gap between the classical value of a unique game and the value returned by its ``basic semidefinite relaxation''. The family of games, known as the Khot-Vishnoi games, is parametrized by an integer $\ell$ and denoted by $(KV_\ell)$; for each $\ell$ letting $n=2^\ell$ the game $KV_\ell$ has $2^n/n$ questions and $n$ answers per player. In~\cite{BuhrmanRSW12} it is shown that for every $n=2^\ell$, 
\begin{equation}\label{eq:kvbound}
\frac{\omega^*({KV}_\ell)}{\omega({KV}_\ell)}\,\geq\, C\,\frac{n}{\log^2 n},
\end{equation}
where $C$ is a universal constant. 
The inequality~\eqref{eq:kvbound} is derived by separately establishing an upper bound on the classical value of $KV_\ell$ and a lower bound on its entangled value. The main ingredient in the proof of the bound on the classical value is a clever use of the Beckner-Bonami-Gross inequality. A lower bound on the entangled value is obtained by explicitly constructing a quantum strategy, obtained by combining a ``quantum rounding'' technique introduced in~\cite{KRT10} with the original argument of Khot and Vishnoi to prove a lower bound on the value of  the semidefinite relaxation of the game $KV_\ell$. The result is a strategy which uses a maximally entangled state of dimension $n$ and achieves a value of order $1/\log^2 n$. We refer to the paper~\cite{BuhrmanRSW12} for more details. 

 A drawback of the Khot-Vishnoi game is that it requires exponentially many questions per player, so that the violation~\eqref{eq:kvbound} is very far from the upper bound provided by the first item from Proposition~\ref{Prop upper bounds}. In~\cite{JungeP11low} the authors show the existence of a family of Bell functionals $(JP_n)_{n\in\N}$ such that for each $n$, $JP_n$ has $n$ inputs and $n$ outputs per system and satisfies
\begin{equation}\label{eq:jpbound}
\frac{\omega_n^*(JP_n)}{\omega(JP_n)}\,\geq\, C'\,\frac{\sqrt{n}}{\log n},
\end{equation}
where $C'$ is a universal constant. Thus the family $(JP_n)$ is only quadratically far from matching the best upper bounds in both parameters simultaneously. No better estimate in terms of the number of inputs is known:
\begin{question}
Does there exist a family of Bell functionals $M_n$ with $n$ inputs per system such that $\frac{\omega^*(M_n)}{\omega(M_n)}$ is $\Omega(n)$?
\end{question} 
In order to describe the family $(JP_n)$, consider first 
the linear map $G_n:\ell_2^n\rightarrow \ell_1^n(\ell_\infty^n)$ given by 
$$G_n:e_i \rightarrow  \frac{1}{n\sqrt{\log n}}\sum_{x,a=1}^ng_{i,x,a}e_x\otimes e_a,$$
where $(g_{i,x,a})_{i,x,a=1}^n$ is a family of independent real standard Gaussian random variables. It follows from Chevet's inequality \cite[Theorem 43.1]{TomczakBook} and standard concentration arguments that there exists a universal constant $C_1$ such that $\|G_n\|\leq C_1$ with high probability. It follows from  this estimate and the fact that the $\epsilon$ norm has the \emph{metric mapping property}~\cite[Section 4.1]{Def93} that the map $M_n$ defined by  
\begin{align*}
M_n&= (G_n\otimes G_n)\Big(\sum_{i=1}^n e_i\otimes e_i\Big)\\
&= \frac{1}{n^2\log n}\sum_{x,y,a,b=1}^n\sum_{i=1}^n  g_{i,x,a}g_{i,y,b}e_x\otimes e_a\otimes e_y\otimes e_b\\
&=\frac{1}{n^2\log n}\sum_{x,y,a,b=1}^n\langle u_x^a|v_y^b \rangle e_x\otimes e_a\otimes e_y\otimes e_b,
\end{align*}
where 
\begin{equation}\label{eq:def-ux-vy}
u_x^a=\sum_{i=1}^ng_{i,x,a}e_i\qquad \text{and}\qquad v_y^b=\sum_{i=1}^ng_{i,y,b}e_i,
\end{equation}
verifies $\|M_n\|_{\ell_1^{\cardx}(\ell_\infty^{\carda})\otimes_{\epsilon} \ell_1^{\cardy}(\ell_\infty^{\cardb})}\leq C_1^2$ with high probability over the choice of the $(g_{i,x,a})$ and $(g_{i,y,b})$.  Introducing a ``dummy'' output (as explained in Section~\ref{sec:signed-coefficients}), from $M_n$ it is straightforward to define a Bell functional $JP_n$ with $n$ inputs and $n+1$ outputs per system such that $\omega(JP_n)\leq C_1^2$. To lower bound the entangled value, a strategy based on the state 
\begin{equation}\label{eq:jp-state}
\ket{\psi}=\frac{1}{\sqrt{2}}\Big(|11\rangle+\frac{1}{\sqrt{n}}\sum_{i=2}^{n+1}|ii\rangle\Big)\in\Ball\big(\C^{n+1}\otimes\C^{n+1}\big)
\end{equation}
 and POVM $\{E_x^a\}_{a}$ and $\{E_y^b\}_{b}$ in $M_{n+1}$ for Alice and Bob respectively (depending on the Gaussian variables $g_{i,x,a}$ and $g_{i,y,b}$) can be found that achieves a value at least $C_2\sqrt{n}/\log n$ in $JP_n$ with high probability, where $C_2$ is a universal constant; the construction is rather intuitive and we leave it as an exercise for the reader. Together, these estimates on the classical and entangled value of $JP_n$ prove~\eqref{eq:jpbound}.

The construction of $JP_n$ is probabilistic, and only guaranteed to satisfy the claimed bound with high probability over the choice of the $(g_{i,x,a})$ and $(g_{i,y,b})$; no deterministic construction of a functional satisfying the same bound is known. The state~\eqref{eq:jp-state} is very different from the maximally entangled state, and it is not known whether the latter can be used to obtain a bound similar to~\eqref{eq:jpbound}; we refer to Section~\ref{sec:general-me} for some limitations of the maximally entangled state in achieving large violations. 

We end this section by giving an interpretation of $JP_n$ as a two-player game $R_n$ due to Regev~\cite{Regev2012}, who obtains the improved bound 
$$\frac{\beta^*(R_n)}{\beta(R_n)}\,\geq\, C\sqrt{\frac{n}{\log n}},$$
 where $\beta=2\omega-1$ and $\beta^*=2\omega^*-1$ denote the classical and entangled biases respectively, through a more careful analysis.\footnote{The presentation of $JP_n$ given above is in fact based on the construction by Regev, who in particular suggested the use of independent coefficients $g_{i,x,a}$ and $g_{i,y,b}$ where~\cite{JungeP11low} used the same coefficients on both sides.} In the game $R_n$, each player is sent a uniformly random $x,y\in\{1,\ldots,n\}$ respectively. Each player must return an answer $a,b\in\{1,\ldots,n+1\}$. If any of the players returns the answer $n+1$ their payoff is $1/2$. Otherwise, the payoff obtained from answers $a,b$ on questions $x,y$ is defined to be $\frac{1}{2}+\delta\langle u_x^a, u_y^b\rangle$, where the vectors $u_x^a$ and $v_y^b$ are defined from independent standard Gaussian random variables as in~\eqref{eq:def-ux-vy} and $\delta$ is a suitable factor chosen so that $\delta\langle u_x^a, u_y^b\rangle \in [-1/2,1/2]$ with high probability. It is then straightforward to verify that both the classical and entangled biases of the game $R_n$ are linearly related to the classical and entangled values of the Bell functional $JP_n$.


\section{Entanglement in nonlocal games}\label{sec:Entanglement}

The study of multiplayer games in quantum information theory is motivated by the desire to develop a quantitative understanding of the nonlocal properties of entanglement. The previous two sections focus on the ratio of the entangled and classical values of a multiplayer game as a measure of nonlocality. In this section we refine the notion by investigating how other measures of entanglement, such as the Schmidt rank or the entropy of entanglement,\footnote{The Schmidt rank of a bipartite pure state is the rank of its reduced density matrix on either subsystem. Its entanglement entropy is the Shannon entropy of the singular values of the reduced density matrix.} are reflected in the properties of nonlocal games. To study the Schmidt rank, for any Bell functional $M$ and integer $d$ define
\begin{equation}\label{eq:def-omega-d}
\omega^*_d(M) \,=\, \sup\, \Big|\sum_{x,y;a,b}\, M_{xy}^{ab} \,\bra{\psi} A_x^a \otimes B_y^b \ket{\psi}\Big|,
\end{equation}
where the supremum is taken over all $k\leq d$, $\ket{\psi}\in\Ball(\C^k\otimes \C^k)$ and families of POVM $\{A_x^a\}_a$ and $\{B_y^b\}_b$ in $M_k$. Clearly $(\omega_d^*(M))_d$ forms an increasing sequence that converges to $\omega^*(M)$ as $d\to\infty$. The quantity 
\begin{equation}\label{eq:def-supm}
\sup_M\,\frac{\omega^*_d(M)}{\omega(M)}
\end{equation} 
thus asks for the largest violation of a Bell inequality achievable by states of Schmidt rank at most $d$. In Section~\ref{sec:dim-bounds} we describe known bounds on this quantity, first for the case when $M$ corresponds to a multiplayer XOR game and then for general two-player games. In Section~\ref{sec:maximally-entangled} we refine our discussion by considering the largest bias achievable by a specific class of states, maximally entangled states. In Section~\ref{sec:os-emb} we introduce a class of ``universal'' states, \emph{embezzlement states}, and give a surprising application of these states to a proof of the operator space Grothendieck inequality, Theorem~\ref{thm:osgt}. Finally in Section~\ref{sec:infinite} we discuss the issue of quantifying the rate of convergence $\omega_d^*\to\omega^*$, present evidence for the existence of a game $G$ such that $\omega^*_d(G)<\omega^*(G)$ for all $d$, and relate these questions to Connes' Embedding Conjecture. 

\subsection{Dimension-dependent bounds}
\label{sec:dim-bounds}

\subsubsection{XOR games}\label{sec:xor-entanglement} 

As discussed in Section~\ref{sec:grothendieck}, Grothendieck's inequality implies that the entangled bias of a two-player XOR game can never be more than a constant factor larger than its classical bias. Moreover, for any XOR game $G$ the optimal violation can always be achieved using a maximally entangled state of local dimension at most $2^{O(\sqrt{\cardx+\cardy})}$, and there exists games for which this bound is tight. Can the optimal bias be approached, up to some error $\eps$, using strategies that have much lower dimension?
 For any integer $d$ let $\beta^*_d(G)$ be defined analogously to~\eqref{eq:def-omega-d} by restricting the supremum in the definition~\eqref{eq:xor-norm-os-0} of the entangled bias to strategies with local dimension at most $d$. The following lemma is attributed to Regev in~\cite{CHTW04}.

\begin{lemma}\label{lem:xor-dim}
Let $G$ be an XOR game. Then for any $\eps>0$ it holds that $\beta^*_{ 2^{\lfloor 2\eps^{-2}\rfloor}}(G)\geq (1-\eps)\beta^*(G)$. 
\end{lemma}

Lemma~\ref{lem:xor-dim} shows that the optimal violation in XOR games can always be achieved, up to a multiplicative factor $(1-\eps)$, using a (maximally entangled) state of dimension that scales exponentially with $\eps^{-2}$.\footnote{Unpublished work of the second author shows that the exponent of $\eps$ can be improved from $2$ to $1$, and that an exponential dependence on $\poly(\eps^{-1})$ is necessary.} Since we are not aware of a proof of Lemma~\ref{lem:xor-dim} appearing explicitly in the literature we include a short argument based on the Johnson-Lindenstrauss lemma.

\begin{proof}[Proof of Lemma~\ref{lem:xor-dim}]
Tsirelson's characterization of the entangled bias using the $\gamma_2^*$ norm (q.v.~\eqref{eq:xor-min-facto}) shows
$$\beta^*(G) = \sup_{d;\,a_x,b_y\in\Ball(\R^d)} \Big|\sum_{x,y} \hat{G}_{xy} \,a_x\cdot b_y\Big|,$$
where $\hat{G}\in\R^{\cardx\cardy}$ is the tensor associated to the game $G$. 
Let $\eta>0$ and fix $d$ and $a_x,b_y\in\Ball(\R^d)$ such that $\sum_{x,y} \hat{G}_{xy}\, a_x\cdot b_y \geq (1-\eta)\beta^*(G)$. For $i=1,\ldots,k$, where $k$ is a parameter to be determined later, let $g_i= (g_{i1},\ldots,g_{id})\in\R^d$ be a vector of independent standard Gaussian random variables. For $x\in \setx$ define $\tilde{a}_x = (g_1\cdot a_x,\ldots,g_k\cdot a_x)\in\R^k$, and for $y\in\sety$ let $\tilde{b}_y = (g_1\cdot b_y,\ldots,g_k\cdot b_y)\in\R^k$. Then
\begin{align}
\beta_k^*(G) &\geq \Ex\Big[\sum_{x,y} \hat{G}_{xy} \,\frac{\tilde{a}_x}{\|\tilde{a}_x\|} \cdot \frac{\tilde{b}_y}{\|\tilde{b}_y\|}\Big]\notag\\
&= \frac{1}{k} \Ex\Big[\sum_{x,y} \hat{G}_{xy}\, \tilde{a}_x \cdot \tilde{b}_y\Big] \notag\\
&\qquad+ 
\frac{1}{\sqrt{k}}\Ex\Big[\sum_{x,y} \hat{G}_{xy} \,\tilde{a}_x \cdot \tilde{b}_y \Big(\frac{1}{\|\tilde{b}_y\|}-\frac{1}{\sqrt{k}}\Big)\Big] -\Ex\Big[\sum_{x,y} \hat{G}_{xy} \,\tilde{a}_x \Big(\frac{1}{\sqrt{k}}-\frac{1}{\|\tilde{a}_x\|}\Big) \cdot \frac{\tilde{b}_y}{\|\tilde{b}_y\|}\Big]\label{eq:xor-dim-0}\\
&\geq (1-\eta)\beta^*(G)- \frac{1}{\sqrt{k}}\beta^*(G) - \frac{1}{\sqrt{k}}\beta^*(G).\notag
\end{align} 
Here to obtain the last inequality the two terms in~\eqref{eq:xor-dim-0} are bounded by interpreting the random variables $k^{-1/2}\tilde{a}_x$, $\tilde{b}_y(\|\tilde{b}_y\|^{-1}-k^{-1/2})$ and $\tilde{a}_x(\|\tilde{a}_x\|^{-1}-k^{-1/2})$ as vectors in $L_2$ with squared norm $\Ex[\|k^{-1/2}\tilde{a}_x\|^2]=1$,
\begin{align*}
 \Ex\Big[\Big\|\tilde{b}_y\Big(\frac{1}{\|\tilde{b}_y\|}-\frac{1}{\sqrt{k}}\Big)\Big\|^2\Big] &=  2-2\frac{\Ex\big[\|\tilde{b}_y\|\big]}{\sqrt{k}}\\
&= 2-2\sqrt{1-\frac{1}{2k}}\\
&\leq \frac{1}{k},
\end{align*}
where the second line follows from standard estimates on the $\chi^2$ distribution, and similarly for $\tilde{a}_x(\|\tilde{a}_x\|^{-1}-k^{-1/2})$. Thus these vectors can lead to a value at most $k^{-1/2}\|\hat{G}\|_{\gamma_2^*}$ when evaluated against the game tensor $\hat{G}$ as in~\eqref{eq:xor-dim-0}. Choosing $k=4\eps^{-2}$ and using Tsirelson's construction to obtain quantum strategies of dimension $2^{\lfloor k/2\rfloor}$ proves the lemma.
\end{proof}

The situation becomes more interesting once XOR games with more than two players are considered. As shown in Section~\ref{sec:three-xor} the ratio of the entangled and classical biases can become unbounded as the number of questions in the game increases. We may again ask the question, what is the largest violation that is achievable when restricting to entangled states of local dimension at most $d$? From the proof of Corollary~\ref{cor:3xor-largeviolation} it follows that for every $d$ there exists a three-player XOR game for which there is an entangled strategy of local dimension $d$ that achieves a violation of order $\sqrt{d}$ (up to logarithmic factors). As it turns out this dependence is essentially optimal:\footnote{The state that can be extracted from the strategy constructed in the proof of Corollary~\ref{cor:3xor-largeviolation} is rather complicated, with $d^3$ coefficients chosen, up to normalization, as i.i.d. standard Gaussians. As will be shown in Section~\ref{sec:general-xor} this is to some extent necessary, as the natural generalization of maximally entangled states to multiple parties (\emph{Schmidt states}, a class that encompasses the GHZ state) can only lead to bounded violations in multiplayer XOR games.} as shown in~\cite{PerezWPVJ08tripartite}, for any $k$-player XOR game it holds that 
\begin{equation}\label{eq:kxor-upperbound}
\beta_d^*(G) \,=\, O\big(d^{\frac{k-2}{2}}\big)\beta(G).
\end{equation}
The main ingredient in the proof of~\eqref{eq:kxor-upperbound} is a noncommutative Khinchine inequality from~\cite{HM07}; we refer to~\cite{PerezWPVJ08tripartite} for details (see also~\cite{BV12} for an elementary proof). The upper bound can be interpreted as stating the existence of certain ``weak'' forms of Grothendieck's inequality for tripartite tensor products: it implies a bound on the rate at which the norm of the amplifications $G\otimes \Id_{\C^d}\otimes\Id_{\C^d}\otimes\Id_{\C^d}$, for $G:\ell_\infty^n\times\ell_\infty^n\times\ell_\infty^n\to\C$, grows with $d$. 
\subsubsection{Two-player games}

In Section~\ref{sec:violations-lb} we described a family of games $(KV_\ell)$ for which a violation of order $2^\ell/\ell^2$ could be achieved using a maximally entangled state of local dimension $2^\ell$. The following proposition shows that this dependence on the dimension is close to optimal.\footnote{A slightly better bound is known for the case of the maximally entangled state; q.v.~\eqref{upper bound max} in Section~\ref{sec:general-me}.}

\begin{proposition}\label{dimension-general}
Let $G$ be a two-player game and $M$ a bipartite Bell functional. Then for every $d\geq 1$,
$$\omega^*_d(G)\leq d\,\omega(G) \qquad\text{and}\qquad\omega^*_d(M)\leq 2d\,\omega(M).$$
\end{proposition}

\begin{proof}
We start with the first estimate. Consider families of POVMs $\{E_x^a\}_{x,a}$, $\{E_y^b\}_{y,b}$ in $M_d$ for Alice and Bob respectively and a pure state $|\psi\rangle\in \Ball(\C^d\otimes \C^d)$. Absorbing local unitaries in the POVM elements, write the Schmidt decomposition as $|\psi\rangle=\sum_{i=1}^d \lambda_i |ii\rangle$, with $\sum_{i=1}^d|\lambda_i|^2=1$. Thus 
\begin{align}
\sum_{x,y;a,b}G_{x,y}^{a,b}\langle\psi|E_x^a \otimes F_y^b|\psi\rangle&=\sum_{i,j=1}^n {\lambda}_i  \lambda_j \sum_{x,y;a,b}G_{x,y}^{a,b}\langle i|E_x^a|j\rangle\langle i|F_y^b|j\rangle\notag\\
&\leq d\max_{i,j}\Big|\sum_{x,y;a,b}G_{x,y}^{a,b}\langle i|E_x^a|j\rangle\langle i|F_y^b|j\rangle\Big|,\label{eq:dimension-general-1}
\end{align}
where for the inequality we used that $|\sum_{i,j=1}^n {\lambda}_i  \lambda_j|=| \sum_{i=1}^d \lambda_i|^2\leq d$ since $\sum_{i=1}^d |\lambda_i |^2=1$. For fixed $i,j$ and any $x\in\setx$ we have 
\begin{align*}
\sum_a|\langle i|E_x^a|j\rangle|\leq \sum_a|\langle i|E_x^a|i\rangle|^{\frac{1}{2}}|\langle j|E_x^a|j\rangle|^{\frac{1}{2}}\leq \Big(\sum_a|\langle i|E_x^a|i\rangle|\Big)^{\frac{1}{2}}\Big(\sum_a|\langle j|E_x^a|j\rangle|\Big)^{\frac{1}{2}} \leq 1,
\end{align*}
where the first two inequalities follow from the Cauchy-Schwarz inequality and the positivity of $E_x^a$, and the last uses $\sum_a E_x^a\leq\Id$. The same bound applies to Bob's POVM, hence for fixed $i,j$, using that $G$ has non-negative coefficients
\begin{align*}
\Big|\sum_{x,y;a,b}G_{x,y}^{a,b}\langle i|E_x^a|j\rangle\langle i|F_y^b|j\rangle\Big|\leq \sum_{x,y;a,b}G_{x,y}^{a,b}|\langle i|E_x^a|j\rangle||\langle i|F_y^b|j\rangle|\leq \omega(G).
\end{align*}
Together with~\eqref{eq:dimension-general-1} this proves the first estimate in the proposition. 

For the second estimate, proceed as above to obtain (\ref{eq:dimension-general-1}). To conclude it will suffice to show that for every $i,j$,
 $$\Big|\sum_{x,y;a,b}M_{x,y}^{a,b}\langle i|E_x^a|j\rangle\langle i|F_y^b|j\rangle\Big|\leq 2\omega^*(M).$$
Fix $i,j$.  Decompose the rank-$1$ operator $\rho=|i\rangle\langle j|$ as $\rho=\rho_1+\sqrt{-1}\rho_2$, where $\rho_1$ and $\rho_2$ are self-adjoint operators with $\|\rho_k\|_{S_1^d}\leq 1$ for $k=1,2$. According to the Hilbert-Schmidt decomposition each $\rho_k$ can be expressed as $\rho_k=\sum_{s=1}^d\alpha_s^k|f_s^k\rangle\langle f_s^k|$, where $\alpha_s^k\in \R$, $\sum_{s=1}^d|\alpha_s^k|\leq 1$ and $(|f_s^k\rangle)_s$ is an orthonormal basis of $\C^d$ for $k=1,2$. Thus
\begin{align*}
\Big|\sum_{x,y;a,b}M_{x,y}^{a,b}\langle i|E_x^a|j\rangle\langle i|F_y^b|j\rangle\Big|
&\leq \sum_{k=1}^2 \sum_{s=1}^d\big|\alpha_s^k\big| \Big|\sum_{x,y;a,b} M_{x,y}^{a,b}\langle f_s^k|E_x^a|f_s^k\rangle\langle i|F_y^b|j\rangle\Big|\\
&
\leq 2\sup_{P\in \class(\seta\setb|\setx\sety)}\Big|\sum_{x,y;a,b}M_{x,y}^{a,b}P(a|x)\langle i|F_y^b|j\rangle\Big|,
\end{align*}
where in the last inequality we used that $(P(a|x))_a=(\langle f|E_x^a|f\rangle)_a$ is a classical distribution for every unit vector $|f\rangle$. Since the coefficients $M_{x,y}^{a,b}$ are real, $\sum_{x,y;a,b}M_{x,y}^{a,b}P(a|x)F_y^b$ is a self-adjoint operator in $M_d$ for every $P\in \class(\seta\setb|\setx\sety)$, and its norm is attained at a pure state. Hence for any $P\in \class(\seta\setb|\setx\sety)$ we get
\begin{align*}
\Big|\sum_{x,y;a,b}M_{x,y}^{a,b}P(a|x)\langle i|F_y^b|j\rangle\Big|&\leq \Big\|\sum_{x,y;a,b}M_{x,y}^{a,b}P(a|x)F_y^b\Big\|\\
&=\sup_{|\psi\rangle\in \Ball(\C^d)}\Big|\sum_{x,y;a,b}M_{x,y}^{a,b}P(a|x)\langle \psi |F_y^b|\psi\rangle\Big|\\
&\leq \omega(M).
\end{align*}
\end{proof}

Proposition~\ref{dimension-general} can be extended to $k$-player games $G$ and $k$-partite Bell functionals $M$ for $k>2$ to give the following bounds: 
\begin{align*}
\omega^*_d(G)\leq d^{k-1}\,\omega(G)\qquad \text{and}\qquad\omega^*_d(M)\leq (2d)^{k-1}\,\omega(M).
\end{align*} 
The proof is a straightforward adaptation of the proof of the proposition, based on the decomposition
\begin{align}\label{claim upper bounds}
|\psi\rangle=\sum_j\lambda_j|f_1^j\rangle\otimes\cdots \otimes  |f_k^j\rangle
\end{align}
of a $k$-partite state $\ket{\psi}\in(\C^d)^{\otimes k}$, where the real coefficients $\lambda_j$ verify $\sum_j|\lambda_j|\leq d^{\frac{k-1}{2}}$ and the vectors $|f_i^j\rangle\in \Ball(\C^d)$ for every $k,j$. The existence of such a decomposition follows from the estimate $\| id:\bigotimes_2^k\ell_2^d\rightarrow \bigotimes_\pi^k\ell_2^d\|\leq d^{\frac{k-1}{2}}$, and can also easily be obtained directly. Using~\eqref{claim upper bounds} we obtain
$$\Big|\sum_{\substack{x_1,\cdots, x_k\\a_1,\cdots, a_k}}G_{x_1,\cdots, x_k}^{a_1,\cdots, a_k}\langle\psi|E_{x_1}^{a_1} \otimes\cdots \otimes E_{x_k}^{a_k}|\psi\rangle\Big|\leq d^{k-1}\sup_{j,j'}\Big|\sum_{\substack{x_1,\cdots, x_k\\a_1,\cdots, a_k}}G_{x_1,\cdots, x_k}^{a_1,\cdots, a_k}\langle f_1^{j'}|E_{x_1}^{a_1}|f_1^j \rangle\cdots \langle f_k^{j'}|E_{x_k}^{a_k}|f_k^j \rangle\Big|$$
in replacement of~\eqref{eq:dimension-general-1}, and the remainder of the proof follows with only minor modifications.

\subsection{Maximally entangled states}\label{sec:maximally-entangled}

The $d$-dimensional maximally entangled state $|\psi_d\rangle=\frac{1}{\sqrt{d}}\sum_{i=1}^d|ii\rangle$ maximizes the entropy of entanglement among all $d$-dimensional bipartite pure states. This fact, together with the wide use of this state as a resource for many fundamental protocols in quantum information theory, such as quantum teleportation or superdense coding, motivate the study of the family $(|\psi_d\rangle)_d$ in the context of Bell inequalities. In this section we review some results in this direction, first for the case of XOR games and then for general Bell functionals. 

\subsubsection{XOR games}\label{sec:general-xor}

As described in Section~\ref{sec:xor-games} optimal strategies in two-player XOR games can always be assumed to be based on the use of a maximally entangled state $\ket{\psi_d}$, for some $d$ which grows with the number of questions in the game. The situation is more subtle for the case of games with more than two players. Indeed there is no good candidate multipartite ``maximally entangled'' state. A natural choice would be the \emph{GHZ state}
$$ \ket{GHZ_d} = \frac{1}{\sqrt{d}} \sum_{i=1}^d \ket{i \cdots i},$$
which at least bears a syntactic resemblance to the bipartite maximally entangled state. This state has the maximum possible entanglement entropy between any one party and the remaining $(k-1)$, but this is no longer true as soon as one considers more balanced partitions. A slight generalization of the GHZ state is the family of \emph{Schmidt states}, states of the form 
\begin{equation}\label{eq:schmidt-state}
\ket{SCH_{\alpha,d}} = \sum_{i=1}^d \alpha_i \ket{i \cdots i},
\end{equation}
where the $\alpha_i$ are any positive reals satisfying $\sum \alpha_i^2 =1$. Somewhat surprisingly, it is shown in~\cite{Briet2011} that the largest possible violation achievable by a Schmidt state in a $k$-player XOR game is bounded by a constant independent of the dimension: for any $k$-player XOR game $G$ it holds that
\begin{equation}\label{eq:kxor-schmidtbound}
\beta_S^*(G)\,\leq\,2^{3(k-2)/2} K_G^\C\, \beta(G),
\end{equation}
where $\beta_S^*(G)$ denotes the largest bias achievable by strategies restricted to use a state of the form~\eqref{eq:schmidt-state} for any $d$, and $K_G^\C$ is the complex Grothendieck constant (q.v. Theorem~\ref{thm:grothendieck}).
As described in Section~\ref{sec:xor-entanglement} it is possible to achieve violations that scale as $\sqrt{d}$ using entangled states of local dimension $d$ in three-player XOR games, thus the states achieving such unbounded violation must be very different from Schmidt states.\footnote{As shown in Section~\ref{sec:quantum-xor}, two-player \emph{quantum} XOR games provide a bipartite scenario in which the maximally entangled state can be far from optimal, as demonstrated by the family of games $C_n$ for which $\beta^*(C_n)=1$ but $\beta^{me}(C_n)=\beta(C_n)\leq 2\sqrt{2/n}$.} 

This observation has consequences in quantum information and functional analysis, and it has an interesting application in computer science. In quantum information, as already described it implies that a large family of natural multipartite entangled states has bounded nonlocality in XOR games, and are thus far from being the most nonlocal. Due to their simple form GHZ states are prime candidates for an experiment attempting to demonstrate multipartite entanglement, but the bound~\eqref{eq:kxor-schmidtbound} suggests that they may not be the most well-suited to demonstrating large violations.

In functional analysis the result is connected to a question in the theory of Banach algebras that dates back to work of Varopoulos~\cite{varopoulos:1972}. The bound~\eqref{eq:kxor-schmidtbound} implies, via a reduction shown in~\cite{davie:1973,PerezWPVJ08tripartite}, that the Banach algebra $S_\infty$ of compact operators on $\ell_2$, endowed with the Schur product, is a $Q$-algebra, i.e. a quotient of a uniform algebra; we refer to~\cite{Briet2011} for details.

Finally the same result has an interesting application to the problem of parallel repetition, or more precisely direct sum, for multiplayer XOR games (q.v. Section~\ref{sec:grothendieck} for an introduction to this question). It is known that there are three-player XOR games for which $\beta^*(G)=1$ but $\beta(G)<1$; Mermin~\cite{mermin:1990} describes such a game for which there is a perfect quantum strategy based on the GHZ state. Taking the $\ell$-fold direct sum $G^{(\oplus\ell)}$, the same strategy played $\ell$ times independently will still succeed with probability $1$, hence $\beta^*(G^{(\oplus\ell)})=1$. The bound~\eqref{eq:kxor-schmidtbound} implies that the classical bias can never be lower than a constant independent of $k$: $\lim\inf_{\ell\to\infty} \beta(G^{(\oplus\ell)}) > 0$. Thus in this game even though the classical bias is strictly less than $1$ the bias of the repeated game does not go to $0$ as the number of repetitions goes to infinity, a surprising consequence. In fact the proof of~\eqref{eq:kxor-schmidtbound} given in~\cite{Briet2011} provides an explicit procedure by which an optimal quantum strategy may be turned into a classical strategy achieving the claimed bias, so that in principle a good correlated strategy for classical players in $G^{(\oplus \ell)}$ can be extracted from an optimal quantum strategy for Mermin's game.

\subsubsection{General Bell inequalities}\label{sec:general-me}

Denote by $\omega^{me}_{d}(M)$ the entangled value of a Bell inequality $M$, when the supremum in the definition~\eqref{eq:def-entangled-value} of $\omega^*$ is restricted to strategies using a maximally entangled state $|\psi_k\rangle = \frac{1}{\sqrt{k}}\sum_{i=1}^k \ket{ii}$ of local dimension $k\leq d$,
 and let $\omega^{me}(M)=\sup_d \omega_{d}^{me}(M)$. In Section~\ref{sec:xor-games} we saw that for two-player XOR games it always holds that $\omega^{me}(G)=\omega^*(G)$, while in the previous section we saw that for games with more players the natural analogue to $\omega^{me}(G)$ defined from Schmidt states (\ref{eq:schmidt-state}) can be strictly smaller than $\omega^*(G)$. Which is the case for general two-player games?
The following theorem provides an answer.

\begin{theorem}[\cite{JungeP11low}]\label{theorem restric MES}
There exists a universal constant $C$ and a family of Bell functionals $(JP'_n)_{n\in\N}$ with $2^{n^2}$ inputs and $n$ outputs per system such that for every $n\geq 2$,
$$\omega^*_n(JP_n')\,\geq\,\frac{ \sqrt{n}}{\log n}\qquad\text{and}\qquad\omega^{me}(JP_n')\,\leq\, C.$$
\end{theorem}

The family $(JP'_n)$ can be obtained as a modification of the family $(JP_n)$ introduced in Section~\ref{sec:violations-lb}. This modification is directly inspired by constructions in operator space theory for which we attempt to give a flavor below. Based on his analysis of $(JP_n)$ Regev~\cite{Regev2012} gave a slightly different construction from the one in Theorem~\ref{theorem restric MES}, with the additional properties that it is explicit and it improves the number of inputs to $2^n/n$; the proof relies on techniques from quantum information theory. It is interesting to note that both constructions yield a family of Bell functionals, with some coefficients negative. Indeed, as shown in~\cite[Theorem 10]{JungeP11low} one cannot hope to obtain a similar result by considering only two-player games, for which it always holds that 
\begin{align}\label{MES positive coeff}
\omega^*_d(G)\leq (4\log d)\,  \omega^{me}_d(G).
\end{align}
While there do exist two-player games for which $\omega^*(G)>\omega^{me}(G)$~\cite{VidickW11}, it is not known if an asymptotically growing violation can be obtained. 
\begin{question}
Does there exist a universal constant $C$ such that $\omega^*(G)\leq C \omega^{me}(G)$ for every two-player game $G$?
\end{question}
Similarly to~\eqref{eq:def-supm}, the nonlocality of the maximally entangled state can be measured through the quantity
$$\sup_M\,\frac{\omega_d^{me}(M)}{\omega(M)}.$$
How does this quantity compare to the one in~\eqref{eq:def-supm}? In~\cite{Palazuelos14largest} it is shown that for every Bell functional $M$ and every fixed dimension $d\geq 2$ one has
\begin{align}\label{upper bound max}
\omega^{me}_d(M) \leq C'\frac{d}{\sqrt{\log d}}\, \omega (M) ,
\end{align}
where $C'$ is a universal constant.  Even though the best upper bound when allowing for non-maximally entangled states, provided by Proposition~\ref{dimension-general}, is only of a factor $O(d)$, the largest violation known is the one obtained by the Khot-Vishnoi game, for which it is achieved using a maximally entangled state: $\omega^{me}_{d}(KV_{\ell}) \geq C(d/\log^2 d)\omega(G)$, where $\ell=\lceil\log d\rceil$. It is an interesting open question to close the gap by finding a game for which a violation of order $d$ (or slightly weaker, in which case one may also seek to improve the bound~(\ref{upper bound max})) can be obtained by using a non-maximally entangled state.

\medskip

We end this section with a brief overview of the role played by operator space theory in obtaining the family of functionals $(JP'_n)$ in Theorem \ref{theorem restric MES}. 
Recall that $JP_n$ is defined by a simple modification of the tensor
$$M_n\,=\,(G_n\otimes G_n) \Big(\sum_{i=1}^n e_i\otimes e_i\Big)\in \ell_1^n(\ell_\infty^n)\otimes \ell_1^n(\ell_\infty^n),$$
where $G_n: \ell_2^n \rightarrow \ell_1^n(\ell_\infty^n)$ is such that the norm of $M_n$ satisfies $\|M_n\|_{\ell_1^n(\ell_\infty^n)\otimes_{min} \ell_1^n(\ell_\infty^n)}\geq \sqrt{n}/\log n$ with high probability. 

It is not hard to extend Lemma~\ref{quantum value 2P1R} relating the entangled value to the minimal norm to show that an analogous (up to the precise constants) relation holds between the maximally entangled value $\omega^{me}$ and the tensor norm $\|\cdot \|_{\ell_1^{\cardx}(\ell_\infty^{\carda})\otimes_{\psi-min}  \ell_1^{\cardy}(\ell_\infty^{\cardb})}$ associated to the tracially bounded norm of a bilinear form in~\eqref{Def tb norm}; the relation holds both for games and general Bell functionals. 
Thus to prove Theorem~\ref{theorem restric MES} it would suffice to establish a bound of the form $\|M_n\|_{\psi-min}\leq C^2$. Below we show how this can be reduced to the problem of obtaining an estimate on the completely bounded norm of the map $G_n$ of the form 
\begin{equation}\label{eq:j-bound}
\big\|G_n:\, R_n\cap C_n\rightarrow \ell_1^n(\ell_\infty^n)\big\|_{cb}\,\leq\, C.
\end{equation}
 Unfortunately, we do not know how to obtain such an estimate on $G_n$ directly, and the map has to be modified to an $R_n:\ell_2^n\rightarrow \ell_1^{N}(\ell_\infty^n)$, where $N=2^{n^2}$. Putting this aspect aside (q.v.~\cite[Section 4]{JungeP11low} for details) it remains to relate the  norms $\|M_n\|_{\psi-min}$ and $\|G_n\|_{cb}$. To this end, consider the o.s.s.\ on $\ell_2^n$ obtained by combining the row and the column structure: for every $d\geq 1$,
\begin{align*}
\Big\|\sum_{i=1}^nA_i\otimes |i\rangle\Big\|_{M_d(R_n\cap C_n)}=\max \Big\{\Big\|\sum_{i=1}^nA_iA_i^*\Big\|^\frac{1}{2}_{M_d},\Big\|\sum_{i=1}^nA_i^*A_i\Big\|^\frac{1}{2}_{M_d}\Big\}.
\end{align*}
Let $R: \ell_2^n \rightarrow \ell_1^N(\ell_\infty^n)$ be any linear map such that $\|R:R_n\cap C_n\rightarrow \ell_1^N(\ell_\infty^n)\|_{cb}\leq C$. Letting
$$M=(R\otimes R) \Big(\sum_{i=1}^ne_i\otimes e_i\Big)\in \ell_1^N(\ell_\infty^n)\otimes \ell_1^N(\ell_\infty^n)$$
it follows from the definition of the tracially bounded norm in~(\ref{Def tb norm}) that 
\begin{align*}
\big\|M\big\|_{\psi-min} &= \sup_{k\leq d;\,u,v\in\Lin(\ell_1^N(\ell_\infty^n),M_k):\,\|u\|_{cb},\|v\|_{cb}\leq 1} \big|\langle\psi_k|(u\otimes v)(M)|\psi_k\rangle\big|\\
&\leq C^2\sup_{k\leq d;\,T,S\in\Lin(R_n\cap C_n,M_k):\,\|T\|_{cb},\|S\|_{cb}\leq 1} \big|\langle\psi_k|(T\otimes S) \big(\sum_{i=1}^ne_i\otimes e_i\big)|\psi_k\rangle\big|,
\end{align*}
where the second inequality follows by setting $T=u\circ R$ and $S=v\circ R$. 
As shown in~\cite[Lemma 2]{JungeP11low}, for every linear maps $T,S\in\Lin(R_n\cap C_n,M_d)$ such that $\|T\|_{cb},\|S\|_{cb}\leq 1$ one has that 
\begin{align}\label{Fact cb embedding}
\Big|\big\langle \psi_d\big|(T\otimes S)\Big(\sum_{i=1}^n e_i\otimes e_i\Big)\big|\psi_d\big\rangle\Big|\leq 1.
\end{align}
Choosing $R=G_n$ would let us obtain $\|M_n\|_{\psi-min}\leq C^2$, as desired, provided~\eqref{eq:j-bound} was true. We do not know if this is the case, but as mentioned above the introduction of a slightly more complicated map $R_n$ leads to the required estimate and Theorem~\ref{theorem restric MES} follows.


\subsection{A universal family of entangled states}
\label{sec:os-emb}

The definition~\eqref{eq:def-eval} of the entangled value of a multiplayer game considers the supremum over all possible shared entangled states, of any dimension. As discussed in Section~\ref{sec:maximally-entangled} the natural choice of the maximally entangled state, even of arbitrarily large dimension, often turns out not to provide an optimal choice for the players. In the first part of this section a different family of states, \emph{embezzlement states}, is introduced which \emph{does} systematically lead to optimal strategies.\footnote{A compactness argument together with a simple discretization of $\Ball(\C^d\otimes \C^d)$ guarantees that one may always find a countable family of states that is guaranteed to enable strategies achieving values arbitrarily close to the optimum; here we are interested in the existence of ``natural'' such families.} In the second part  it will be explained how embezzlement states suggest an elementary proof of the operator space Grothendieck inequality,~Theorem~\ref{thm:osgt}.

\subsubsection{Quantum embezzlement} 

The following family of quantum states is introduced in~\cite{DamH03embezzlement}. For any integer $d$, define the $d$-dimensional embezzlement state $\ket{\Gamma_d}$ as
\begin{equation}\label{eq:emb-state}
\ket{\Gamma_d} = Z_d^{-1/2}\sum_{i=1}^d \frac{1}{\sqrt{i}}\ket{ii}\,\in\Ball(\C^d\otimes\C^d),
\end{equation}
where $Z_d = \sum_{i=1}^d i^{-1}$ is the appropriate normalization constant. 

\begin{theorem}[\cite{DamH03embezzlement}]\label{thm:embezzlement}
Let $\ket{\psi} \in \Ball(\C^d\otimes \C^d)$, and $\eps>0$. There exists a $d' = d^{O(\eps^{-1})}$ and unitaries $U,V\in\Unitary(\C^{d'})$ (depending on $\ket{\psi}$) such that 
$$ \big\| U\otimes V \ket{\Gamma_{d'}} - \ket{\psi}\otimes \ket{\Gamma_{d'/d}}\big\|\,\leq\,\eps,$$
where the state $\ket{\psi}\otimes \ket{\Gamma_{d'/d}}$ should be understood as a bipartite state on $(\C^d\otimes \C^{d'/d})\otimes (\C^d\otimes \C^{d'/d})$.
\end{theorem}

The proof of Theorem~\ref{thm:embezzlement} is not difficult, and based on comparing the Schmidt coefficients of the states $\ket{\Gamma_{d'}}$ and $\ket{\psi}\otimes\ket{\Gamma_{d'/d}}$; indeed it is not hard to see that an optimal choice of unitaries $U,V$ consists in matching the Schmidt vectors of the former state to those of the latter, when the associated Schmidt coefficients are sorted in non-increasing order. Although the family of states described in Theorem~\ref{thm:embezzlement} is not the only family having the properties claimed in the theorem, the constraint of ``universal embezzlement'' imposes rather stringent conditions on the Schmidt coefficients associated to any such family~\cite{leung2014characteristics}.

It follows directly from Theorem~\ref{thm:embezzlement} that for any two-player game there exists a sequence of strategies for the players that uses an embezzlement state of increasing dimension and whose success probability converges to the entangled value of the game. Although a generic discretization argument shows the existence of such a family for any number of parties, it is not known if there are ``natural'' families of embezzlement states for more than two parties. 

\subsubsection{A proof of the operator space Grothendieck inequality}

In Section~\ref{sec:quantum-xor} we considered a class of two-player games with quantum questions called ``quantum XOR games''. A quantum XOR game is represented by a Hermitian bilinear form $G\in\Bil(M_\cardx,M_\cardy)$, and the unentangled and entangled biases of $G$ are related to the norm and completely bounded norm of $G$ as 
\begin{equation}\label{eq:qxor-bias-norm}
 \bias(G) = \|G\|_{\Bounded(M_\cardx,M_\cardy)}\qquad\text{and}\qquad\bias^*(G) = \|G\|_{\cBil(M_\cardx,M_\cardy)}
\end{equation}
respectively. Although these two biases can be very far apart, based on the family of states $(\ket{\Gamma_d})$ introduced in the previous section it is possible to make the following observation: 
\begin{align}
\bias^*(G) &=\sup_{\substack{d\in\N;\,A\in \Ball(\Bounded(\C^\cardx\otimes\C^d)),\\B\in \Ball(\Bounded(\C^\cardy\otimes \C^d))}} \big\| \Tr_{\C^\cardx\otimes\C^\cardy}(G\otimes \Id_{\C^d} \otimes \Id_{\C^d} \cdot(A\otimes B))\big\|\notag\\
&=\sup_{\substack{d'\in\N;\,A\in \Ball(\Bounded(\C^\cardx\otimes\C^{d'})),\\B\in \Ball(\Bounded(\C^\cardy\otimes \C^{d'}))}} \big|\bra{\Gamma_{d'}} \Tr_{\C^\cardx\otimes\C^\cardy}(G\otimes \Id_{\C^{d'}} \otimes \Id_{\C^{d'}} \cdot(A\otimes B))\ket{\Gamma_{d'}}\big|\notag\\
&=\sup_{\substack{d'\in\N;\,A\in \Ball(\Bounded(\C^\cardx\otimes\C^{d'})),\\B\in \Ball(\Bounded(\C^\cardy\otimes \C^{d'}))}}  \big|\Tr_{(\C^\cardx\otimes \C^{d'})\otimes(\C^\cardy\otimes\C^{d'})}\big((G\otimes \ket{\Gamma_{d'}}\bra{\Gamma_{d'}}) \cdot(A\otimes B)\big)\big|\notag\\
&= \sup_{d'} \bias\big(G\otimes \ket{\Gamma_{d'}}\bra{\Gamma_{d'}}\big),\label{eq:emb-bias}
\end{align}
where $G' = G\otimes \ket{\Gamma_{d'}}\bra{\Gamma_{d'}}$ has a natural interpretation as a quantum XOR game in which the referee, in addition to the same questions as sent to the players in $G$, sends each of them half of an embezzlement state of dimension $d'$. In $G'$ the players no longer need any entanglement in order to achieve the entangled bias: as long as $d'$ is large enough they can simply ``embezzle'' whatever state they find useful from the copy of $\ket{\Gamma_{d'}}$ provided for free by the referee. 

The extensions of Grothendieck's inequality described in Theorem~\ref{thm:ncgt} and Theorem~\ref{thm:osgt} express a relationship between two quantities, introduced as $\beta^{nc}(M)$ and $\beta^{os}(M)$, and the bounded and completely bounded norm of $M\in\Bil(M_\cardx,M_\cardy)$ respectively: 
$$ \bias^{nc}(M) \leq 2 \|M\|_{\Bounded(M_\cardx,M_\cardy)}\qquad\text{and}\qquad\bias^{os}(M) \leq 2 \|M\|_{\cBil(M_\cardx,M_\cardy)}.$$
The correspondence between the entangled bias of a game $G$ and the unentangled bias of the game $G\otimes \ket{\Gamma_d}\bra{\Gamma_d}$ described above, combined with the relations~\eqref{eq:qxor-bias-norm}, suggests an intriguing possibility: if it were the case that $\bias^{os}(M) = \sup_d \bias^{nc}(M\otimes \ket{\Gamma_d}\bra{\Gamma_d})$ then the sequence of equalities~\eqref{eq:emb-bias} would (for the case of Hermitian $M$) translate to 
$$\bias^{os}(M) \stackrel{?}{=}  \sup_d \bias^{nc}(M\otimes \ket{\Gamma_d}\bra{\Gamma_d}) \leq 2\sup_d\bias(M\otimes \ket{\Gamma_d}\bra{\Gamma_d}) = 2\bias^{*}(M),$$
thereby providing a succinct derivation of  Theorem~\ref{thm:osgt} from Theorem~\ref{thm:ncgt}.\footnote{All known proofs of Theorem~\ref{thm:osgt} involve a reduction to Theorem~\ref{thm:ncgt}; see~\cite{Haagerup85NCGT} or~\cite{NRV13} for a more recent proof borrowing on techniques from computer science.}
This is precisely the approach taken in~\cite{RegevV12b}, where the equality $\bias^{os}(M) =  \sup_d \bias^{nc}(M\otimes \ket{\Gamma_d}\bra{\Gamma_d})$ is proved for bilinear forms $M$ defined on arbitrary operator spaces. The proof of the equality follows from the proof of Theorem~\ref{thm:embezzlement} sketched in the previous section but requires additional manipulations to relate the supremum appearing in the definition of $\bias^{os}$ in~\eqref{eq:osgt} to the one in the definition of $\bias^{nc}$ in~\eqref{eq:ncgt}. The resulting proof of Theorem~\ref{thm:osgt} has the benefit of being completely elementary, and is a good example of an application of techniques in quantum information to derive a result in the theory of operator spaces.

\subsection{The limit of infinite dimensions}
\label{sec:infinite}

The definition~\eqref{eq:def-eval} of the entangled value of a multiplayer game considers the supremum over entangled states of any dimension. Except in a few cases (such as two-player XOR games, as discussed in Section~\ref{sec:grothendieck}) it is not known whether the supremum is achieved. The question turns out to be intimately tied to important (and difficult) conjectures in the theory of $C^*$-algebras. This connection is discussed in Section~\ref{sec:connes} below. First we give evidence for the existence of nonlocal games for which the supremum may indeed not be achieved in any finite dimension. 

\subsubsection{Games requiring infinite entanglement}

There are two examples of families of games known for which the entangled bias is provably not attained by any strategy using a finite-dimensional entangled state. The first is the family of quantum XOR games $(T_n)$ already encountered in Section~\ref{sec:quantum-xor}. These games have messages of dimension $n$ per player and satisfy $\bias^*(T_n)=1$ for every $n\geq 2$, but any strategy achieving a bias at least $1-\eps$ requires the use of an entangled state of dimension $n^{\Omega(\eps^{-1/2})}$. 

The second is based on Hardy's paradox from quantum information theory~\cite{Hardy92paradox}. In~\cite{MancinskaV14unbounded} a two-player game $G$ with classical questions and answers is introduced such that the optimal value $\omega^*(G)=1$ can only be achieved with infinite-dimensional entanglement. Moreover, a success probability of $1-\eps$ requires an entangled state of dimension at least $\Omega(\eps^{-1/2})$. The game $G$ has only two questions per player but it has the peculiarity that answers can take values in a countably infinite set. Although this is not a requirement directly imposed by the referee in the game it is possible to show that optimal strategies always need to have the possibility of sending answers of arbitrary length (the probability of this happening decreases exponentially with the length of the answer). 

Due to their use respectively of quantum questions or an infinite number of possible answers these examples do not quite provide a construction of a ``standard'' multiplayer game for which the entangled value is not attained by any finite-dimensional strategy, and at the time of writing no such example has been found. Perhaps the best candidate is the $I_{3322}$ Bell inequality for which suggestive numeric investigations have been performed~\cite{pal2010maximal}. 

\begin{question}
Does there exist a multiplayer game $G$ such that the entangled value $\omega^*(G)$ is only attained in the limit of infinite-dimensional entanglement? 
\end{question}

It is interesting to investigate the rate at which the entangled value can be approached, as a function of the dimension. The games described above suggest some lower bounds, but no upper bound is known. 

\begin{question}
Given $\eps>0$ and $n\in\N$ provide an explicit\footnote{A compactness argument, combined with a simple discretization shows the existence of a finite $d=d(\eps,n)$, but does not give any estimate on the size of $d$.} bound on $d=d(\eps,n)$ such that for any game $G$ with at most $n$ questions and answers per player there exists a $d$-dimensional strategy with success probability at least $\omega^*(G)-\eps$.
\end{question}

We believe answering these two natural questions poses a significant challenge; indeed in spite of increased interest over the years little concrete progress has been made. 

\subsubsection{Tsirelson's Problem and Connes' Embedding Conjecture}
\label{sec:connes}

In Section \ref{Sec: Bell ineq and multiplayer} we introduced the set of quantum conditional distributions $\ent(\seta\setb|\setx\sety)$ as the set of distributions which can be implemented locally using measurements on a bipartite quantum state. For the purposes of this section we fix $\cardx=\cardy=N$ and $\carda= \cardb=K$, and write $\mathcal P_\mathcal Q(N,K)$ for $\ent(\seta\setb|\setx\sety)$. It is not known if the set $\mathcal P_\mathcal Q(N,K) $ is closed, and we introduce its closure
\begin{align*}
Q^{ten}(N,K) &= \text{cl}\Big\{ \big(\bra{\psi} A_x^a\otimes B_y^b \ket{\psi}\big)_{x,y;a,b}:\, \ket{\psi}\in \Ball(\mathcal H\otimes \mathcal K),\, A_x^a\in\Pos(\mathcal H),B_y^b\in\Pos(\mathcal K),\\\nonumber
&\hskip5cm\,\sum_{a=1}^K A_x^a = \sum_{b=1}^K B_y^b=\Id \,\,\forall(x,y)\in[N]\times [N]\Big\},
\end{align*}where the supremum is taken over all Hilbert spaces $\mathcal H$, $\mathcal K$.

The relevance of the tensor product in this definition is that it is the standard mathematical construction used to model spatially separated systems in non-relativistic quantum mechanics and most of quantum information theory. In other areas, such as algebraic quantum field theory, local measurements are represented by operators acting on a joint Hilbert space, with the independence condition modeled by imposing that operators associated with distinct spacial locations commute. This description naturally leads to the following definition:
\begin{align*}
Q^{com}(N,K) &= \Big\{ \big(\bra{\psi} A_x^aB_y^b \ket{\psi}\big)_{x,y;a,b}:\, \ket{\psi}\in \Ball(\mathcal H),\, A_x^a,B_y^b\in\Pos(\mathcal H), \, [A_x^a,B_y^b]=0,\\\nonumber
&\hskip5cm\,\sum_{a=1}^K A_x^a = \sum_{b=1}^K B_y^b=\Id \,\,\forall(x,y)\in[N]\times [N]\Big\},
\end{align*}
where here $[T,S]=TS-ST$ for operators $T$ and $S$. Since commuting operators acting on a finite dimensional Hilbert space can always be written in a tensor product form~\cite{SW08},\footnote{This observation was already known to Tsirelson.} $Q^{ten}(N,K)$ can be equivalently defined by replacing the  $\bra{\psi} A_x^a\otimes B_y^b \ket{\psi}$ with $\bra{\psi} A_x^aB_y^b \ket{\psi}$ for every $x,y,a,b$, as long as the restrictions that the POVM elements act on a finite dimensional Hilbert space $\mathcal H$ and that POVM elements associated with different parties commute are imposed.
Tsireson's Problem asks the following question:
\begin{align}\label{Tsirelson problem}
\text{Does }Q^{ten}(N,K)=Q^{com}(N,K) \text{       }\text{    for every   } N, K \text{   ?    }
\end{align} 
Since the definition of either set is directly related to the axiomatic description of composite systems, Tsirelson's Problem is a fundamental question in the foundations of quantum mechanics. A negative answer to~\eqref{Tsirelson problem} could lead to the experimental validation of the possibility that finite dimensional quantum models do not suffice to describe all bipartite correlations, even when only a finite number of inputs and outputs are considered.

In recent years this question has captured the attention of operator algebraists because of its connection with the so called Connes' Embedding Problem~\cite{Connes76}. This problem asks whether every von Neumann algebra can be approximated by matrix algebras in a suitable sense, and it has been (unexpectedly) related with many fundamental questions in the field of operator algebras~\cite{Ozawa13}. Formally, Connes' problem asks whether every finite von Neumann algebra $(\mathcal M, \tau)$ with separable predual embeds in the ultrapower $R^\omega$ of the hyperfinite $II_1$-factor $R$. 

In order to explain how Connes' problem is related with~\eqref{Tsirelson problem} we give an equivalent formulation due to Kirchberg~\cite{Kirchberg93qwep}. For an arbitrary $C^*$-algebra $A$ define the minimal and maximal norms of $x\in  A\otimes A$ by
\begin{align*}
\|x\|_{min}&=\sup\Big\{\|(\pi_1\otimes \pi_2)(x)\|_{\mathcal B(\mathcal H\otimes \mathcal K)}:\pi_1:A\rightarrow \mathcal B(\mathcal H), \pi_2:A\rightarrow \mathcal B(\mathcal K)\Big\},\\
\|x\|_{max}&=\sup\Big\{\|(\pi_1\cdot\pi_2)(x)\|_{\mathcal B(\mathcal H)}:\pi_1,\pi_2:A\rightarrow \mathcal B(\mathcal H)\Big\},
\end{align*}
where both suprema run over all $*$-homomorphims $\pi_1$, $\pi_2$ such that in addition for the case of the maximal norm $\pi_1$ and $\pi_2$ have commuting ranges.\footnote{Using that $C^*$-algebras are operator spaces one can verify that this definition for the minimal norm coincides with the one given in~\eqref{Def min-norm}.} More precisely, here $\pi_1\otimes\pi_2$ and $\pi_1\cdot\pi_2$ are defined on elementary tensors respectively by $(\pi_1\otimes\pi_2)(x_1\otimes x_2)=\pi_1(x_1)\otimes \pi_2(x_2)$ and $(\pi_1\cdot\pi_2)(x_1\otimes x_2)=\pi_1(x_1) \pi_2(x_2)$ for every $x_1,x_2\in A$, and extended to $A\otimes A$ by linearity. We denote by $A\otimes_{min} A$ (resp. $A\otimes_{max} A$) the completion of the algebraic tensor product $A\otimes A$ for the norm $\|\cdot\|_{min}$ (resp. $\|\cdot\|_{max}$). Kirchberg's reformulation of Connes' Embedding Problem asks whether $C^*(\mathbb F_n)\otimes_{min} C^*(\mathbb F_n)=C^*(\mathbb F_n)\otimes_{max} C^*(\mathbb F_n)$  for all (some) $n\geq 2$, where $C^*(\mathbb F_n)$ is the universal C$^*$-algebra associated to the free group with $n$ generators $\mathbb F_n$. We give an equivalent formulation of this question, which can be obtained from standard (though non-trivial) techniques in operator algebras. Let $\mathcal F_K^N:=\ell_\infty^K*\cdots * \ell_\infty^K$ denote the unital free product of $N$ copies of the commutative C$^*$-algebra $\ell_\infty^K$.\footnote{A formal definition of $\mathcal F_K^N$ is beyond the scope of this survey; for our purposes it will suffice to think of it as a C$^*$-algebra containing $N$ copies of $\ell_\infty^K$ and with the universal property that for every family of unital $*$-homomorphisms $\pi_i:\ell_\infty^K\rightarrow \mathcal B(\mathcal H)$, $i=1,\cdots ,N$ there is a unital $*$-homomorphism $\pi:\ell_\infty^K*\cdots * \ell_\infty^K\rightarrow \mathcal B(\mathcal H)$ which extends the $\pi_i$.} 
Kirchberg's theorem states that Connes' Embedding Problem is equivalent to the question
\begin{align}\label{Kirchberg III}
\text{Does }\mathcal F_K^N\otimes_{min} \mathcal F_K^N=\mathcal F_K^N\otimes_{max} \mathcal F_K^N \text{     } \text{  for all (some)}  \text{     } (N,K)\neq  (2,2) \text{      } \text{    ?}
\end{align}
The introduction of the minimal and maximal norms gives a good hint of the relationship with Tsirelson's Problem. To formulate the connection more precisely,~\eqref{Kirchberg III} can be related to a ``matrix-valued'' analogue of Tsirelson's Problem as follows. Given $n\in \mathbb N$, define 
\begin{align}\label{distributions tensor matrix}
Q^{ten}_n(N,K) &= \text{cl}\Big\{ \big[V^* A_x^a\otimes B_y^b V\big]_{x,y;a,b}:\, V:\ell_2^n\rightarrow \mathcal  H\otimes  \mathcal K,\, A_x^a\in\Pos(\mathcal  H), B_y^b\in\Pos(\mathcal  K),\\\nonumber
&\hskip5cm\,\sum_{a=1}^K A_x^a = \sum_{b=1}^K B_y^b=\Id \,\,\forall(x,y)\in[N]\times [N]\Big\},
\end{align}
where the supremum is taken over all Hilbert spaces $ \mathcal H$, $\mathcal  K$ and all isometries $V:\ell_2^n\rightarrow \mathcal  H\otimes  \mathcal K$, and
\begin{align}\label{distributions commutative matrix}
Q^{com}_n(N,K) &= \Big\{ \big[V^*A_x^aB_y^b V \big]_{x,y;a,b}:\, V:\ell_2^n\rightarrow \mathcal H,\, A_x^a, B_y^b\in\Pos(\mathcal H), \, [A_x^a,B_y^b]=0,\\\nonumber
&\hskip5cm\,\sum_{a=1}^K A_x^a = \sum_{b=1}^K B_y^b=\Id \,\,\forall(x,y)\in[N]\times [N]\Big\},
\end{align}
where the supremum is taken over all Hilbert spaces $ \mathcal H$ and all isometries $V:\ell_2^n\rightarrow \mathcal H$. Note that in both definitions we have $\big[V^* A_x^aB_y^b V\big]_{x,y,a,b}\in M_n$ for every $x,y,a,b$.
Let $e_i^{(m)}$ be the element $e_i\in \ell_\infty^K$, when seen as its $m^{th}$-copy in $\mathcal F_K^N$, $p_i^{(m)}=e_i^{(m)}\otimes 1\in \mathcal F_K^N\otimes \mathcal F_K^N$ and $q_j^{(l)}=1\otimes e_j^{(l)}\in \mathcal F_K^N\otimes \mathcal F_K^N$. Letting $S_n(A)$ denote the set of completely positive and unital maps from $A$ to $M_n$, a direct application of the  Stinespring factorization theorem and the universality property of the free product shows that 
\begin{align*}
Q^{ten}_n(N,K) =\Big\{ \big[\varphi(p_a^{(x)}q_b^{(y)})\big]_{x,y;a,b}:\varphi\in S_n(\mathcal F_K^N\otimes_{min} \mathcal F_K^N)\Big\}
\end{align*}and  
\begin{align*}
Q^{com}_n(N,K)=\Big\{ \big[\varphi(p_a^{(x)}q_b^{(y)})\big]_{x,y;a,b}:\varphi\in S_n(\mathcal F_K^N\otimes_{max} \mathcal F_K^N)\Big\}.
\end{align*}
From these last characterizations it is immediate that a positive answer to~\eqref{Kirchberg III} implies
\begin{align}\label{Matrix Tsirelson}
Q^{ten}_n(N,K) = Q^{com}_n(N,K)  \text{       } \text{     for every    } N, K, n,
\end{align}
whereas the characterizations~\eqref{distributions tensor matrix} and~\eqref{distributions commutative matrix} make it apparent that for the special case $n=1$,~\eqref{Matrix Tsirelson} implies a positive answer to Tsirelson's Problem. 
It was shown in \cite{JNPPSW11,Fritz12} that condition (\ref{Matrix Tsirelson}) for all $n$ implies a positive answer to Connes' Embedding Problem (see also \cite{FKPT14} for some simplifications). This was recently extended by Ozawa~\cite{Ozawa13} who shows that it is sufficient to assume condition (\ref{Matrix Tsirelson}) for $n=1$ in order to derive an affirmative answer to Connes' Embedding Problem. This is done by proving that the coincidence of the restriction of the max and min tensor norms to the space $L(N,K)\otimes L(N,K)$, where $L(N,K)$ is the subspace of $\mathcal F_K^N$ spanned by words of length $1$,\footnote{$L(N,K)$ is the same (operator) space as $\mathcal{N}_\C(N,K)^*$, introduced in Section \ref{sec:signed-coefficients}.} at the Banach space level already suffices to give an affirmative answer to Connes' Embedding Problem. 

\

\centerline{\textbf{Acknowledgments}}

The first author was supported by Spanish research projects MINECO (grant MTM2011-26912), Comunidad de Madrid (grant QUITEMAD+-CM, ref. S2013/ICE-2801) and \emph{Ram\'on y Cajal} program, and the European CHIST-ERA project CQC (funded partially by MINECO grant PRI-PIMCHI-2011-1071).

\newcommand{\etalchar}[1]{$^{#1}$}

\end{document}